\documentclass[a4paper,11pt]{article}
\usepackage{jheppub}

\usepackage[most]{tcolorbox}

\usepackage{amsmath}
\usepackage{amsmath, amsthm, amssymb, hyperref}

\usepackage{amsthm}

\theoremstyle{remark}

\theoremstyle{plain}
\newtheorem{theorem}{Theorem}

\theoremstyle{definition}
\newtheorem{definition}{Definition}
\newtheorem{eg}[theorem]{Example}

\newtheorem{claim}{Claim}
\newtheorem{proposition}{Proposition}

\usepackage{esint}
\usepackage{babel}
\usepackage{enumitem}

\usepackage{subfigure}

\usepackage{float}

\usepackage{amsfonts}
\usepackage{amssymb}
\usepackage{amsthm}
\usepackage{enumitem}
\usepackage{shuffle}
\usepackage{comment}
\usepackage{cancel}
\usepackage{blkarray}

\usepackage{savesym}
\savesymbol{div}
\usepackage{physics}
\restoresymbol{TXF}{div}
\usepackage{tikz}
\usepackage{tikz-3dplot}
\usetikzlibrary{patterns, decorations.markings, arrows}
\usetikzlibrary{arrows.meta}
\usepackage{mathabx}
\usepackage{mathtools}
\usepackage[export]{adjustbox}

\usepackage[numbers,sort&compress]{natbib}
\usepackage{hyperref}
\allowdisplaybreaks
\hypersetup{
colorlinks=true,
linkcolor=blue,
linktocpage=true,
citecolor=violet
}

\usepackage{xargs}
\usepackage[colorinlistoftodos,prependcaption,textsize=tiny]{todonotes}

\newcommand*{\LABELREVERSE}{{non-simple intersecting}}
\newcommand*{\LABELNONSIMPLE}{{non-simple non-intersecting}}

\newcommand{\munich}{Max Planck Institute for Physics, Garching bei M\"unchen, Germany}
\newcommand{\davis}{Center for Quantum Mathematics and Physics (QMAP), 
University of California, Davis, USA}

\newcommand{\la}{\langle}
\newcommand{\ra}{\rangle}
\newcommand{\ang}[1]{\langle #1\rangle}
\newcommand{\sq}[1]{\left[ #1\right]}

\begin{document}

\preprint{MPP-2025-44}

\title{All-loop Leading Singularities of Wilson Loops}

\author[a]{Taro V. Brown,}
\author[b]{Johannes M. Henn,}
\author[b]{Elia Mazzucchelli,}
\author[a]{Jaroslav Trnka}
\affiliation[a]{\davis}
\affiliation[b]{\munich}
\emailAdd{tvbrown@ucdavis.edu}
\emailAdd{henn@mpp.mpg.de}
\emailAdd{eliam@mpp.mpg.de}
\emailAdd{trnka@ucdavis.edu}

\abstract{
We study correlators of null, $n$-sided polygonal Wilson loops with a Lagrangian insertion in the planar limit of the ${\cal N}=4$ supersymmetric Yang-Mills theory. This finite observable is closely related to loop integrands of maximally-helicity-violating amplitudes in the same theory, and, conjecturally, to all-plus helicity amplitudes in pure Yang-Mills theory.
The resulting function has been observed to have an expansion in terms of functions of uniform transcendental weight, multiplied by certain rational prefactors, called leading singularities. In this work we prove several conjectures about the leading singularities: we classify and compute them at any loop order and for any number of edges of the Wilson loop, and show that they have a hidden conformal symmetry.
This is achieved by leveraging the geometric definition of the loop integrand via the Amplituhedron. The leading singularities can be seen as maximal codimension residues of the integrand, and the boundary structure of the Amplituhedron geometry restricts which iterative residues are accessible. Combining this idea with a further geometric decomposition of the Amplituhedron in terms of so-called negative geometries allows us to identify the complete set of leading singularities.}

\setcounter{tocdepth}{2}

\maketitle
\graphicspath{{figs/}}

\newpage

\section{Introduction}
\label{sec:introduction}

Scattering amplitudes in the planar limit of the ${\cal N}=4$ supersymmetric Yang-Mills theory (sYM) have been an amazing playground for testing new theoretical ideas. They have played an important role in the discovery of hidden mathematical structures in the heart of scattering amplitudes, and the formulation of new efficient methods used later in the context of QCD and gravitational wave calculations \cite{Henn:2020omi,Bern:2022jnl,Travaglini:2022uwo}. One of the most important features of ${\cal N}=4$ sYM amplitudes is the duality with Wilson loops \cite{Alday:2008yw,Henn:2009bd}. In particular, the $n$-particle scattering amplitude is dual to null polygonal Wilson loop with $n$ cusps. This duality underlies the dual conformal and Yangian symmetry \cite{Drummond:2008vq,Berkovits:2008ic,Drummond:2009fd} of ${\cal N}=4$ sYM amplitudes, which is a cornerstone of many other advances in this field, including the symbol bootstrap \cite{Caron-Huot:2020bkp}, cluster algebra structures \cite{Goncharov:2010jf}, connection to cells in the positive Grassmannian and on-shell diagrams \cite{Arkani-Hamed:2009ljj,Arkani-Hamed:2012zlh,Arkani-Hamed:2014bca,Paranjape:2022ymg,Brown:2022wqr}, and the Amplituhedron \cite{Arkani-Hamed:2013jha,Arkani-Hamed:2013kca,Arkani-Hamed:2017vfh,Damgaard:2019ztj,Ferro:2022abq}.

The Amplituhedron provides a novel, geometric definition of the loop integrand of planar sYM theory. (For recent reviews, see references \cite{Herrmann:2022nkh,De:2024bpk}, as well as the special volume \cite{Ranestad:2025qay} on Positive Geometry.) 
Starting from the kinematic data of the amplitudes (as well as additional auxiliary variables that encode the helicity information), i.e. the external and loop momenta, the Amplituhedron space is defined via certain inequalities. Conjecturally, this space admits a unique canonical form with logarithmic singularities at its boundaries, and the latter coincides with the loop integrand. While this has been proven rigorously only at one loop,
the construction is consistent with unitarity and on-shell recursion relations at all loop orders \cite{Arkani-Hamed:2013kca}. Remarkably, the latter are outputs of this definition, not inputs. This radically new picture for scattering amplitudes has led to many new insights and results for amplitudes in sYM and ABJM theories \cite{Franco:2014csa,Dennen:2016mdk,Arkani-Hamed:2018rsk,Herrmann:2020qlt,Dian:2022tpf,He:2023rou,Arkani-Hamed:2023epq,Ferro:2024vwn,De:2024bpk,Dian:2024hil}. The Amplituhedron has also been extensively studied from a mathematical perspective and is connected to recent developments in many disciplines, including cluster algebras or tropical geometry \cite{Galashin:2018fri,Lukowski:2020dpn,Parisi:2021oql,Even-Zohar:2023del,Akhmedova:2023wcf,Even-Zohar:2024nvw,Lam:2024gyg,Parisi:2024psm,Galashin:2024ttp}.

Scattering amplitudes in gauge theories are divergent in the infrared (IR) because of the presence of massless particles. This requires introducing a regulator, such as dimensional regularization with $D=4 - 2\epsilon$, which in principle breaks or at least obscures some of the nice features mentioned above. Nevertheless, the general structure of IR divergences in gauge theory is well understood, especially in the planar limit. In fact, studies in sYM theory have helped to elucidate these structures, and in even provided exact results. As an example, the leading infrared divergence of the amplitude, the cusp anomalous dimension, is predicted (in the planar limit) from integrability by the Beisert-Eden-Staudacher formula to arbitrary values of coupling \cite{Beisert:2006ez}. Moreover, the finite part of the amplitudes is constrained by dual conformal symmetry, which fixes its functional form \cite{Alday:2007hr,Drummond:2007au}
at four- and five-particles to agree with the Bern-Dixon-Smirnov conjecture \cite{Bern:2005iz}, and  for $n\ge 6$ constrains it to be a function of dual conformal cross-ratios \cite{Drummond:2008vq}, which has been intensively studied \cite{Caron-Huot:2020bkp}. 

These remarkable results relied in part on symmetry and integrability properties, as well as on certain bootstrap assumptions. We find it highly desirable to connect these advances to the Amplituhedron definition of integrands at any loop order. 
The necessity to introduce dimensional regularization at intermediate stages of calculations present a major obstacle to this. Indeed, the results for finite, four-dimensional parts of amplitudes only appear after carefully dealing with regulator-dependent objects. 
Ideally one would like to connect more directly predictions for observables, by which we mean finite, regulator-independent quantities, to the underlying structures.

One line of research with that philosophy is to consider four-dimensional scattering amplitudes on the Coulomb branch of sYM theory. The pattern of vacuum expectation values for (some of) the scalar fields in the theory introduces certain mass terms in the Lagrangian. As argued in reference \cite{Alday:2009zm}, this allows to define scattering amplitudes that are infrared finite in four dimensions, preserve dual conformal symmetry, and in the limit of massless particles reproduce the known massless amplitudes in a controlled way. 
Moreover, Coulomb branch amplitudes are conjectured to be related to a mass-deformed version of the Amplituhedron \cite{Arkani-Hamed:2023epq,Flieger:2025ekn}.

Here we follow a different approach, which consists in focusing on an observable that is closely related to scattering amplitudes.
A natural variant of the null polygonal Wilson loops, which are dual to amplitudes, are correlation functions of such Wilson loops with local operators. Such observables are natural in the AdS/CFT correspondence. Moreover, when normalizing those correlators by the (vacuum expectation value of the) Wilson loop itself, the result is finite in four dimensions. 
A particularly interesting choice of operator insertion is that of the stress-tensor multiplet, which contains the Lagrangian of the theory. We denote this finite ratio by $F_n$ \cite{Alday:2011pf,Alday:2011ga}.

This observable has a direct connection to scattering amplitudes. Thanks to the Wilson loop / scattering amplitudes duality, its integrand can be obtained from the Amplituhedron. The trick is to consider a derivative in the coupling of the logarithm of the amplitude / Wilson loop. The latter produces a Lagrangian insertion that is integrated over Minkowski space. In other words, the integrand of $F_n$ at $L-1$ loops is obtained from that of the (logarithm of the) $L$-loop amplitude.

An important comment is that carrying out the $L-1$ loop integrations yields a finite function, which depends on the external kinematics, as well as on the Lagrangian insertion point. 
The infrared divergences of the amplitudes appear if one thinks about performing the integration over the Lagrangian insertion point. In fact, this connection allowed the authors of ref. \cite{Henn:2019swt} to determine the first non-planar correction to the cusp anomalous dimension in sYM and QCD, at the four-loop level. In summary, $F_n$ is a finite observable, defined starting from the Amplituhedron, and can be thought of as a finite version of scattering amplitudes.

This object has been studied at both weak and strong couplings. The four-point case is known up to three loops in perturbation theory \cite{Alday:2012hy,Alday:2013ip,Henn:2019swt}, and at strong coupling \cite{Alday:2011ga}. The five-point case has recently been evaluated to two loops \cite{Chicherin:2022zxo}. Moreover, the authors of ref. \cite{Arkani-Hamed:2021iya} proposed an Amplituhedron-inspired geometric expansion of $F_n$ in terms of so-called `negative geometries'. The latter can be thought of as a certain partial triangulation or tiling of the Amplituhedron, with special properties. In this setup, it is possible to perform calculations at higher-loop orders and perform partial resummations, cf. \cite{Arkani-Hamed:2021iya,Brown:2023mqi,Chicherin:2024hes,Glew:2024zoh}, not only in the sYM theory but also in the ABJM theory \cite{He:2022cup,Henn:2023pkc,Lagares:2024epo}.

The perturbative results for $F_n$ suggest that its $L$-loop contribution is a function of uniform transcendental weight $2L$ \cite{Arkani-Hamed:2010pyv,Henn:2013pwa}, 
with prefactors that are algebraic in the kinematic data. The latter algebraic expressions are called leading singularities in the literature \cite{Cachazo:2008vp}. This structure is expected to be a consequence of the underlying positive geometry definition of (the loop integrand of) $F_n$. In that picture, the independent set of maximal residues of the integrand (obtained by replacing the Minkowski space integration by contour integrals) corresponds to the set of leading singularities.
Knowing what types of transcendental functions and which leading singularities may appear in a given observable is crucial information for bootstrap approaches. 

Recent studies revealed a number of surprising properties of $F_n$ that we review presently. First, a qualitative relationship is that $F_n$ and Yang-Mills scattering amplitudes depend on the same number of variables and, in fact, on the same kinematic data.
This connection can be understood if one considers that the dual conformal symmetry of $F_n$ can be used, without loss of generality, to send the location of the Lagrangian insertion point to infinity.
In that frame, $F_n$ depends only on the $n$ null momenta, as a scattering amplitude in a theory without any special symmetries.

Second, there is a remarkable conjectured quantitative relationship.
The authors of ref. \cite{Chicherin:2022bov} 
conjectured, that the maximal transcendental part of $(L+1)$-loop pure Yang-Mills all-plus helicity amplitudes coincides with $F_n$ at $L$ loops. This conjecture is supported by all data computed so far, i.e. $n$-point one-loop and two-loop all-plus amplitudes, as well as the three-loop four-point all-plus amplitude.
This correspondence involves the agreement of complicated multivariable transcendental functions.

Third, $F_n$ was found to have a uniform sign when evaluated within the Amplituhedron kinematic region \cite{Chicherin:2024hes}. This is a highly non-trivial property that in general comes about as a cancellation several terms, and which involved both leading singularities and transcendental functions.
(Previous analogous observation were made for finite parts of scattering amplitudes \cite{Arkani-Hamed:2014dca,Dixon:2016apl}.) 
Finally, the recent paper \cite{Henn:2024qwe} suggested that positivity may just the tip of the iceberg, 
and that in fact a class of \emph{completely monotonous functions}
(for which the function and all derivatives have positivity properties) 
may be relevant to functions coming from Positive Geometry.

Fourth, another unexpected feature concerns the leading singularities of $F_n$. 
In reference \cite{Chicherin:2022bov} it was found that all tree-level and one-loop $n$-point leading singularities have a hidden momentum-space conformal symmetry, which is typical of scattering amplitudes. Moreover, all leading singularities found could be expressed in terms of a simple Grassmannian formula, inspired by similar formulas for scattering amplitudes \cite{Arkani-Hamed:2009ljj,Arkani-Hamed:2009kmp,Arkani-Hamed:2012zlh}.

In this paper, we focus on the problem of leading singularities and we will succeed to classify all of them and determine their form to all-loop orders and all multiplicities. This provides all-order information about $F_n$ and scattering amplitudes in the planar ${\cal N}=4$ sYM theory. This could be used as an input in bootstrap methods \cite{Dixon:2011pw,Dixon:2011nj,Caron-Huot:2019vjl} (see also reviews \cite{Caron-Huot:2020bkp,Henn:2020omi}). 

As a tool, we use the Amplituhedron description \cite{Arkani-Hamed:2013jha} of the loop integrand as the differential form on a certain positive geometry provided by the sign flip definition \cite{Arkani-Hamed:2017vfh}. 
Compared to the amplitudes themselves, an advantage is that, being finite, $F_n$ can directly be defined in four dimensions (while for scattering amplitudes subtle regularization effects might need to be considered). The finiteness of $F_n$ can be made manifest in the context of an Amplituhedron-inspired, geometric expansion \cite{Arkani-Hamed:2021iya}.
In general, the building blocks in that expansion do not correspond to traditional Feynman diagrams, but have the advantage of being infrared finite. 

Classifying all leading singularities is also the first step to relate the geometric information to the integrated quantity. This has been first pointed out in the context of amplitudes in $\mathcal{N}=4$ sYM in a series of papers
\cite{Dennen:2016mdk,Prlina:2017azl,Prlina:2017tvx}. The underlying principle is that the geometric structure of the Amplituhedron's boundary guides us through the Landau analysis, and allows to obtain the singular loci of the (integrated) amplitude in a much more efficient way compared to the standard Landau analysis. This is because the geometry encodes the information of the numerator of the full integrand, whose vanishing in turn prevents the appearance of certain singularities in the integrated result. Therefore, the procedure proposed in the references is to start with Landau diagrams associated to leading singularities, and to discard all the sub-leading singularities thereof that in some sense are not compatible with the Amplituhedron's geometry. For the case of amplitudes, leading singularities correspond to maximal codimension boundaries of the Amplituhedron, and are easy to classify \cite{Arkani-Hamed:2013kca}. This has to be contrasted with our setting, where leading singularities correspond to higher-dimensional boundaries of the geometry, because we keep the dependence on an unintegrated loop (i.e., the Lagrangian insertion). This feature drastically complicates the determination of all possible leading singularities, but we reach a conclusive answer for the Wilson loop of Lagrangian insertion $F_n$. This therefore opens the door for carrying out considerations analogous to those in the references, i.e. to use the geometric information to extract the singular loci of $F_n$, or more generally, to integrals associated to individual negative geometries. We return to this point in another paper.

Leveraging the sign-flip definition of the MHV Amplituhedron \cite{Arkani-Hamed:2017vfh}, we analyze which `cuts' (residues) of the loop integrand of $F_n$, and more generally of integrals coming from `negative' geometries \cite{Arkani-Hamed:2021iya} are accessible. Even though one can think as usual of LS as maximal residues of the corresponding complexified integrand, in our setting we favour the geometric interpretation of LS as residues of a canonical form on a boundary of a positive geometry. The geometry then dictates which boundary, i.e. which residue, is allowed and which one is not. In this way we prove that every $n$-point $L$-loop LS of a negative geometry can be expressed as a linear combination of canonical functions of one-loop $n$-point Amplituhedron geometries in $AB$ with some extra conditions. The negative geometry expansion provides a perfect framework for approaching the problem of leading singularities to all loops and all points. 

The paper is organized as follows. The next section \ref{sec:lightning_review} provides a lightning review of the key quantities and definitions relevant to our paper.
The main results, as well as an outline of their derivation is presented in section \ref{sec:mainderivation}. (Certain technical parts of the proof are delegated to four Appendices. 
We conclude and provide an outlook in section \ref{sec:outlook}.

\section{Lightning review}
\label{sec:lightning_review}

\subsection{Kinematic data for scattering amplitudes}

We endow $\mathbb{R}^4$ with the Minkowski inner product, and take $n$ massless particles $p_i^2 = 0$ for $i=1, \dots,n$. These satisfy momentum conservation, i.e. $\sum_i p_i = 0$. For planar scattering amplitudes, it is common to use spacetime coordinates in the dual momentum space, where region momenta $x_i \in \mathbb{R}^4$ are related to cyclically ordered external momenta $p_i \in \mathbb{R}^4$, via $p_i = x_i - x_{i-1}$. Then, momentum conservation translates in dual momentum space to the identification $x_{n+1} = x_1$, i.e. to having a null-polygon in the dual omentum space. Yet another useful set of coordinates are \textit{spinor-helicity} variables, in which we express the momenta $p^{\dot{\alpha} \alpha}_i = \tilde{\lambda}^{\dot{\alpha}}_i \lambda^{\alpha}_i$ in terms of $\tilde{\lambda}_i, \lambda_i \in \mathbb{C}^2$, and similarly for the dual momentum coordinates.

In this paper we favor \textit{momentum twistor} variables, as they constitute (a part of) the geometric space in which the Amplituhedron lives \cite{Arkani-Hamed:2013jha}. These variables originate from the twistor correspondence \cite{Penrose}, which relates points in Minkowski spacetime to lines in twistor space. The kinematic data for $n$ massless particles can then be efficiently encoded by $n$ complex vectors $Z_i \in \mathbb{C}^4$. They are related to the previous coordinates by $Z_i = (\lambda_i^\alpha , x_i^{\dot{\alpha} \alpha} \lambda_{i,\alpha})$. A remarkably beautiful example of an amplitude expressed in spinor-helicity variables is the \textit{Parke-Taylor} formula \cite{Parke:1986gb},
\begin{equation}\label{PT}
    {A}_n^{\rm tree} = \, \langle ij  \rangle^4 \, {\rm PT}_n \,,
\qquad 
    {\rm PT}_n = \frac{1}{\langle 12 \rangle \cdots \langle n-1 n \rangle \la n1 \ra} \,,
\end{equation}
for the tree-level (color-orderd) MHV amplitude of the scattering of $n$ gluons in both $\mathcal{N}=4$ sYM and QCD, where the gluons $i$ and $j$ have negative helicity.
Here $\langle a b \rangle := \lambda_a^\alpha \lambda_b^\beta \varepsilon_{\alpha \beta}$. 
Similarly, the $\overline{{\rm MHV}}$ amplitude is given by \eqref{PT} by replacing $\la ij \ra \rightarrow [ij]:= \tilde{\lambda}_i^\alpha \tilde{\lambda}_j^\beta \varepsilon_{\alpha \beta}$.

We now restrict to planar $\mathcal{N}=4$ sYM theory, which exhibits a dual conformal symmetry. In momentum twistor variables, this symmetry is implemented by the action of ${\rm SL}(4)$, and the $Z_i$'s transform in the corresponding fundamental representation.
The (up to scalar unique) ${\rm SL}(4)$-invariant multilinear form is then given by 
\begin{equation}\label{four_bracket}
    \langle ijkl \rangle := \det(Z_i Z_j Z_k Z_l ) = \epsilon_{IJKL} Z^{I}_i Z^{J}_j Z^{K}_k Z^{L}_l \,,
\end{equation}
where $\epsilon_{IJKL}$ is the four-dimensional Levi-Civita tensor and in the determinant's argument we stack together the vectors $Z_i$'s to form a $4 \times 4$ matrix. Moreover, scattering amplitudes in planar $\mathcal{N}=4$ sYM are invariant under the action of the \textit{little group} $(\mathbb{C}^{*})^{n}$, which acts on the $Z_i$'s by scaling. This implies that we can regard the $Z_i$'s as points in the projective space~$\mathbb{P}^{3}$.

\subsection{Loop integrands}

In order to describe integrands for MHV amplitude's  $L$-loop corrections \cite{Arkani-Hamed:2010pyv}, we need to introduce (dual momentum) \textit{loop variables} $y_\ell \in \mathbb{R}^4$ for $\ell = 1, \dots , L$. This is achieved by associating a line $AB_\ell$ in $\mathbb{P}^3$ to $y_\ell$, where the notation means that $A_\ell$ and $B_\ell$ are points in $ \mathbb{P}^3$ and $AB_\ell$ is the unique line passing through them. 
This formulation naturally yields to the \textit{Grassmannian} ${\rm Gr}(k,n)$, the space of $k$-dimensional subspaces of $\mathbb{C}^{n}$. In fact, we can identify $AB_\ell$ as elements of ${\rm Gr}(2,4)$, represented by a $2 \times 4$ complex matrix modulo ${\rm GL}(2)$ transformations. 

In the planar limit of $\mathcal{N}=4$ sYM, the perturbative expansion of the $n$-point MHV amplitude, normalized by its tree-level value of eq. (\ref{PT}), reads 
\begin{equation}
    A_n = 1+ \sum_{L \geq 1} (g^2)^{L} A_n^{(L)} \,,
\end{equation}
where the coupling is
\begin{equation}\label{coupling_const}
    g^2= \frac{g_{YM}^2 N_c }{16 \pi^2} \,.
\end{equation}
The $n$-point $L$-loop contribution can be written as the integral 
\begin{equation}\label{ampl_def_integrand}
A_n^{(L)} = \int_{AB_1} \dots \int_{AB_L} \, \mathbf{\Omega}_n^{(L)} \,.
\end{equation}
The integrand in eq. \eqref{ampl_def_integrand} is a differential form that can be written as
\begin{equation}\label{ampl_form}
     \mathbf{\Omega}_n^{(L)} = \prod_{\ell=1}^{L} \, \langle AB_\ell \, d^{2}A_\ell \rangle \, \langle AB_\ell \, d^{2}B_\ell \rangle \, \, \Omega_n^{(L)}  \,,
\end{equation}
where $\Omega_n^{(L)}$ is a rational function in twistor coordinates in the external momenta $Z_i$ and in the loop variables $AB_\ell$. The explicit expressions for arbitrary $n$ is rather complicated for higher $L$. The state-of-the art is the two-loop N$^k$MHV integrands (for arbitrary $k$) \cite{Bourjaily:2015jna} and three-loop MHV integrands \cite{Arkani-Hamed:2010zjl}.

As a concrete example, let us discuss the four-point one-loop MHV amplitude's integrand in $\mathcal{N}=4$ sYM, whose only contributing diagram is a massless box. The integrand in both dual momentum coordinates and momentum twistors reads
\begin{equation}\label{box_integrand}
    \Omega_4^{(1)} = \frac{(x_1-x_3)^2 (x_2 -x_4)^2}{(y-x_1)^2 (y-x_2)^2 (y-x_3)^2 (y-x_4)^2} = \frac{\langle 1234 \rangle ^ 2}{\langle AB 12 \rangle \langle AB 23 \rangle \langle AB 34 \rangle \langle AB 41 \rangle} \,,
\end{equation}
where the integration measure is given by
\begin{equation}
     d^4y= \langle AB \, d^{2}A \rangle \, \langle AB \, d^{2}B \rangle \,.
\end{equation}
Then, the four-point one-loop MHV amplitude $A^{(1)}_4$ is obtained by integrating \eqref{box_integrand} over the dual loop momentum $y$, or equivalently over the loop twistor $AB$. More details on the Minkowski integration contour in momentum twistors can be found in \cite{Hodges:2010kq}.

Using generalized unitarity \cite{Arkani-Hamed:2010pyv}, we can express the integrand of the $n$-point one-loop amplitude in eq. (\ref{ampl_def_integrand}) as
\begin{equation}
    \Omega_{n}^{(1)} = \sum_{i=2}^{n-3} \sum_{j=i{+}2}^{n-1} \frac{\la AB(i{-}1ii{+}1)\cap(j{-}1jj{+}1)\ra}{\la ABi{-}1i\ra\la ABii{+}1\ra\la ABj{-}1j\ra\la ABjj{+}1\ra\la ABn1\ra}  \,,
\end{equation}
where we use the short-hand notation  
\begin{equation}\label{plane_intersection}
\begin{aligned}
        \langle AB (a_1 b_1 c_1) \cap (a_2 b_2 c_2) \rangle &= \langle A a_1 b_1 c_1 \rangle \langle B a_2 b_2 c_2 \rangle - \langle B a_1 b_1 c_1 \rangle \langle A a_2 b_2 c_2 \rangle \\ &= \langle AB a_1 b_1 \rangle \langle c_1 a_2 b_2 c_2 \rangle + \langle AB b_1 c_1 \rangle \langle a_1 a_2 b_2 c_2 \rangle \\ & \quad  + \langle AB c_1 a_1 \rangle \langle b_1 a_2 b_2 c_2 \rangle \,,
\end{aligned}
\end{equation} 
for the intersection of two planes $(a_i \, b_i \, c_i)$ in $\mathbb{P}^3$.

It is also useful to introduce the integrand of the logarithm of the amplitude, denoted by $\widetilde{\mathbf{\Omega}}_n^{(L)}$,
\begin{equation}\label{log_ampl}
({\rm log}A_n)^{(L)} = \int_{AB_1} \dots \int_{AB_L} \, \widetilde{\mathbf{\Omega}}_n^{(L)} \,.
\end{equation}
See e.g. ref. \cite{Henn:2019swt} for detailed examples at $n=4$.

\subsection{Amplituhedron and Kermits}
\label{subsec:Amplituhedron and Kermits}

The remarkable discovery from \cite{Arkani-Hamed:2013jha} is that the integrand $\mathbf{\Omega}_n^{(L)}$ is the \textit{canonical form} of a geometric object $\mathcal{A}^{(L)}_n$, called the \textit{Amplituhedron}. In the following we refer to $\Omega_n^{(L)}$ in \eqref{ampl_form} as the \textit{canonical function}. More precisely, eq. \eqref{ampl_form} is the unique logarithmic top-form with specific little group weights and prescribed pole locations dictated by the geometry of $\mathcal{A}^{(L)}_n$. 
This encapsulates two claims: that $\mathcal{A}^{(L)}_n$ is in fact a positive geometry according to the definition in \cite{Arkani_Hamed_2017}, and that its canonical form \eqref{ampl_form} computes the full integrand in planar $\mathcal{N}=4$ sYM. The first statement has been rigorously proven only for the case of $L=1$ in \cite{Ranestad:2024svp}, while the second one is supported by several computations consistent with the BCFW recursion, as e.g. for $n=4$ and $L=2,3$ in \cite{Arkani-Hamed:2013kca}. The BCFW recursion has been recently formalized i the context of cluster algebras and tiles for the $m=4$ Amplituhedron at tree level \cite{Even-Zohar:2023del}.

\subsubsection{The Amplituhedron} 
Let us now introduce the geometry of the Amplituhedron, following \cite{Arkani_Hamed_2018}. We parameterize the external kinematics by the momentum twistors $Z_i$. More precisely, the momenta of the $n \geq 4$ massless particles are encoded in a real $4 \times n$ matrix $Z$, whose columns are the $Z_i$'s. Then, $Z$ represents a point in the configuration space ${\rm Gr}(4,n)/(\mathbb{C}^{*})^{n}$. 

The first condition on the geometry is the \textit{external positivity}, which requires $Z$ to be an element in the positive Grassmannian, i.e.
\begin{equation}\label{external_positivity}
    \langle ijkl    \rangle > 0 \,, \qquad \forall \, i<j<k<l \,. 
\end{equation}
Then, the \textit{$n$-point $L$-loop {\rm MHV} Amplituhedron} $\mathcal{A}_{n}^{(L)}$, where we omit the other parameters\footnote{More generally, the (loop) Amplituhedron depends on additional nonnegative integer parameters $m$ and $k$, other than $n$ and $L$. The parameter $m$ is related to the dimension of the momentum twistors, and the physically relevant case is that of $m=4$. The parameter $k$ is then related to the helicity sector; e.g. for MHV amplitudes the relevant parameter is $k=0$. There is an isomorphism between the Amplituhedron for parameters $L=1$, $m=4$, $k=0$ and $L=0$, $m=2$, $k=2$. This is relevant when we later relate to the mathematical literature on the $m=2$ Amplituhedron.} $m=4$ and $k=0$, is a semialgebraic set in the product of Grassmannians ${\rm Gr}_{\mathbb{R}}(2,4)^{L}$, parameterised by $L$ loop variables $AB_\ell \in {\rm Gr}_{\mathbb{R}}(2,4)$, and carved out by two type of inequalities. First, the \textit{loop positivity} conditions for every $\ell = 1, \dots , L$:
\begin{equation}\label{loop_positivty}
\begin{aligned}
    \langle AB_\ell \, 12 \rangle > 0 \,, \dots, \langle AB_\ell \, n-1n \rangle > 0 \,, \ \langle AB_\ell \, 1 n \rangle > 0 \,, 
\end{aligned}
\end{equation}
and the sequence $(\langle AB_\ell \, 12 \rangle , \dots, \langle AB_\ell \, 1n \rangle)$ has exactly two sign flips, ignoring zeros.
Second, the \textit{mutual positivity} conditions between loops, which read
\begin{equation}\label{mutual_positivity}
    \langle AB_\ell \, AB_{\ell'} \rangle > 0 \,, \quad \forall \, \ell \neq \ell' \,.
\end{equation}

\begin{figure}[t]
\centering
\begin{tikzpicture}[scale = 0.5]
        
    \begin{scope}

    \node at (-7.5,0) {$[a_1 b_1 c_1;a_2 b_2 c_2] = $};
    
    \draw[thick] (0,0) circle(3cm);

    \coordinate (i) at (-100:3cm);
    \coordinate (j) at (185:3cm);
    \coordinate (k) at (110:3cm);
    \coordinate (l) at (70:3cm);
    \coordinate (n) at (-60:3cm);
    \coordinate (m) at (0:3cm);
    
    \node at (i) [below ] {$a_1$};
    \node at (j) [left] {$b_1$};
    \node at (k) [above ] {$c_1$};
    \node at (l) [above ] {$a_2$};
    \node at (m) [right ] {$b_2$};
    \node at (n) [below ] {$c_2$};

    \fill[black, opacity=0.3] (i) -- (j) -- (k) -- cycle;
    
    \fill[black, opacity=0.3]  (l) -- (n) -- (m) -- cycle;

    \draw[line width=0.4mm, red, dashed] (i) -- (j);
    \draw[line width=0.4mm, red, dashed] (l) -- (m);
    \draw[line width=0.4mm, red, dashed] (j) -- (k);
    \draw[line width=0.4mm, red, dashed] (m) -- (n);
    \draw[line width=0.4mm, red, thick] (i) -- (k);
    \draw[line width=0.4mm, red, thick] (l) -- (n);

     \foreach \p in {i,j,k,l,m,n} {
        \fill (\p) circle(3pt);
    }
    \end{scope}


    \begin{scope}[xshift = 14 cm]

    \node at (-6,0) {$[abcd] = $};
    
    \draw[thick] (0,0) circle(3cm);

    \coordinate (i) at (-85:3cm);
    \coordinate (j) at (185:3cm);
    \coordinate (k) at (110:3cm);
    \coordinate (l) at (0:3cm);
    
    \node at (i) [below ] {$a$};
    \node at (j) [left] {$b$};
    \node at (k) [above ] {$c$};
    \node at (l) [right ] {$d$};

    \fill[black, opacity=0.3] (i) -- (j) -- (k) -- (l) -- cycle;

    \draw[line width=0.4mm, red, dashed] (i) -- (j);
    \draw[line width=0.4mm, red, dashed] (j) -- (k);
    \draw[line width=0.4mm, red, dashed] (k) -- (l);
    \draw[line width=0.4mm, red, dashed] (i) -- (l);

     \foreach \p in {i,j,k,l} {
        \fill (\p) circle(3pt);
    }
    \end{scope}

\end{tikzpicture}
\caption{Pictorial representation of bicolored subdivisions of type $(2,n)$ associated to Kermit forms at $n$ points, where the circle represents an $n$-gon. }
\label{Figure_Yangian invariants}
\end{figure}
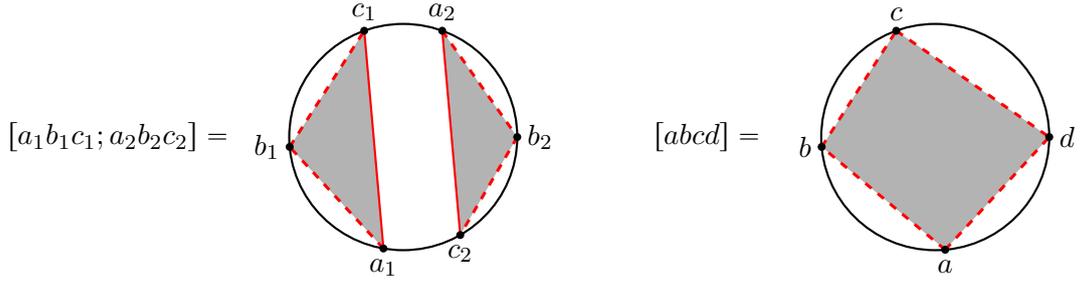
\subsubsection{Kermits} In order to compute the canonical form \eqref{ampl_form} one can tile the Amplituhedron into smaller pieces, which simple canonical form. If, for instance, the geometry was a polytope, one would triangulate it by simplices, whose form is straightforward to compute \cite{Arkani_Hamed_2018}. In the context of Amplituhedron, the objects analogous to simplices are called \textit{Kermits}, for $L=1$.
Then, the canonical function can be expressed as \cite{Arkani_Hamed_2017}
\begin{equation}\label{omega_in_Kermits_wr1}
    \Omega_n^{(1)} = \sum_{j=i{+}2}^{n{-}1} \sum_{i=2}^{n{-}3} \, [1ii{+}1;1jj{+}1] \,,
\end{equation}
where the terms in the sum are called \textit{Kermit forms}, and are more generally given by
\begin{equation}\label{six_invariant}
    [a_1 b_1 c_1 ; a_2 b_2 c_2] = \frac{\langle AB  (a_1 b_1 c_1) \cap (a_2 b_2 c_2) \rangle ^{2}}{\langle AB a_1 b_1 \rangle \langle AB b_1 c_1 \rangle \langle AB a_1 c_1 \rangle \langle AB a_2 b_2 \rangle \langle AB b_2 c_2 \rangle \langle AB a_2 c_2 \rangle } \,,
\end{equation}
with the special case
\begin{equation}\label{four_invariant}
    [abcd] := [a b c ; b c d] = -\frac{\langle abcd \rangle ^{2}}{\langle AB ab \rangle \langle AB bc \rangle \langle AB cd \rangle \langle AB da \rangle } \,.
\end{equation}

Kermits are in bijection with a class of combinatorial objects called \textit{bicolored subdivisions of type} $(2,n)$, which are pairs of non-overlapping black triangles inside an $n$-gon, labeled by their vertices $\{a_i,b_i,c_i\}$ lying on the $n$-gon, for $i=1,2$ \cite{Parisi:2021oql}. Such triangles either share at most one vertex, or they share one face, and therefore form a quadrilateral $\{a, b, c, d\}$. 
In the former case, the associated Kermit has six codimension-one boundaries and its form is given by eq. \eqref{six_invariant}, while in the latter case it has four codimension-one boundaries and its form is given in \eqref{four_invariant}. This combinatorial characterization of Kermits and of their forms appears in Figure \ref{Figure_Yangian invariants}. Then, \eqref{omega_in_Kermits_wr1} is a special case of 
\begin{equation}\label{omega_in_kermits}
    \Omega_n^{(1)} =  \sum_{\Delta_1, \, \Delta_2 \,\subset \, T} [\Delta_1; \Delta_2] \,, 
\end{equation}
where the notation means that we sum over all non-overlapping triangles $\Delta_1 = \{a_1,b_1,c_1\}$ and $\Delta_2 = \{a_2,b_2,c_2\}$ with arcs in a triangulation $T$ of the $n$-gon.
The fact that \eqref{omega_in_kermits} is true for any triangulation $T$ of the $n$-gon, reflects the fact that there are many ways of tiling $\mathcal{A}^{(1)}_n$ into Kermits. The triangulation involving all arcs passing through vertex 1 yields \eqref{omega_in_Kermits_wr1}.

\subsection{The integrand of the Wilson loop with Lagrangian insertion}

Let us introduce the Wilson loop with single Lagrangian insertion in planar ${\cal N}=4$ sYM. More precisely, the observable we consider is
\begin{equation}\label{WL}
    F_n(x_1, \dots , x_n;y) = \frac{1}{\pi^2} \frac{\langle W_n(x_1, \dots , x_n) \, \mathcal{L}(y) \rangle}{\langle W_n(x_1, \dots , x_n) \rangle} \,,
\end{equation}
where $W_n$ is the $n$-sided null Wilson loop in the fundamental representation of the gauge group, with $n$ cusps at the points $x_i$ (which are dual to $Z_i Z_{i+1}$). The Lagrangian $\mathcal{L}$ is inserted at an arbitrary point $y$, which is dual to $AB$. The observable \eqref{WL} is finite, since the cusp divergences of the Wilson loop (dual to infrared divergences of scattering amplitudes) are canceled by taking the ratio in \eqref{WL}. We write the perturbative expansion of eq. \eqref{WL} as
\begin{equation}
    F_n = \sum_{L \geq 0} (g^2)^{L+1} F_n^{(L)} \,,
\end{equation}
in the coupling \eqref{coupling_const}. 

From the duality between amplitudes and Wilson loops (as interpreted at the level of loop integrands, see e.g. \cite{Eden:2010ce}) it follows that,
\begin{equation}\label{WL_int}
  F_n^{(L)} = \int_{AB_1} \dots \int_{AB_{L}} \, \widetilde{\mathbf{\Omega}}_n^{(L+1)} \,,
\end{equation}
where 
$\widetilde{\mathbf{\Omega}}_n^{(L)}$ is the integrand of the logarithm of the amplitude, cf. (\ref{log_ampl}),
and 
$AB:= (AB)_{L+1}$ is not being integrated over.
Therefore, $F_n^{(L)}$ can be thought of as the integrand of the $(L+1)$-loop contribution to (the logarithm of) the scattering amplitude \eqref{log_ampl}, where all but one integrations have been carried out. In particular, the Born-level contribution is given by the rational function
\begin{equation}
    F^{(0)}_n = \widetilde{\Omega}^{(1)}_n  \,,
\end{equation}
and moreover $\widetilde{\Omega}^{(1)}_n=\Omega^{(1)}_n$, with $\Omega^{(1)}_n$ given in  eq. \eqref{omega_in_kermits}.

In \cite{Arkani-Hamed:2021iya} it was observed that an Amplituhedron-like description exists also for the integrand $\widetilde{\mathbf{\Omega}}_n^{(L)}$ in eq. \ref{WL_int} and \eqref{log_ampl}. In fact, $\widetilde{\mathbf{\Omega}}_n^{(L)}$, which is a rational function in the kinematic variables, can be decomposed by an expansion in so called \textit{negative geometries}~\cite{Arkani-Hamed:2021iya}.

\subsection{Negative geometry expansion} Let us now briefly review the geometric procedure to obtain the integrand $\widetilde{\mathbf{\Omega}}_n^{(L)}$ in eq. \eqref{log_ampl} and \eqref{WL_int}, discovered in \cite{Arkani-Hamed:2021iya}. In the following we denote by $\widetilde{\Omega}_n^{(L)}$ the rational function obtained from $\widetilde{\mathbf{\Omega}}_n^{(L)}$ by stripping-off the measure, analogous to eq. \eqref{ampl_form}. 

Let us start with the simple observation that the product of two one-loop Amplituhedra $\mathcal{A}^{(1)}_n$ can be triangulated into two pieces, one of which is the two-loop Amplituhedron $\mathcal{A}^{(2)}_n$, and the other one has an analogous definition as $\mathcal{A}^{(2)}_n$, but where the mutual positivity condition $\langle AB_1 AB_2 \rangle > 0 $ is replaced by a negativity condition $\langle AB_1 AB_2 \rangle < 0 $. The latter geometry is an example of what we call a \textit{negative geometry}. We write this relation schematically as
\begin{equation}
\label{eq:linkrel}
	\begin{tabular}{cc}
	 \includegraphics[scale=.86]{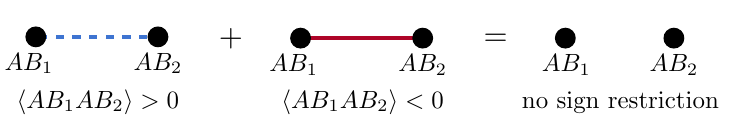} 
	\end{tabular}
\end{equation}
We represent the $L$-loop Amplituhedron $\mathcal{A}^{(L)}_n$ by a complete graph with $L$ nodes and dashed edges. Then, by repeatedly applying this relation, one can write
\begin{equation}\label{omega_in_neg_geom}
    \Omega^{(L)}_n = \sum_{\Gamma} (-1)^{E(\Gamma)} \, \Omega_{n,\Gamma}^{(L)} \,.
\end{equation}
Here the summation runs over all graphs $\Gamma$ with $L$ nodes and any number of plain edges, with $E(\Gamma)$ denoting the number of edges. The rational function $\Omega_{n,\Gamma}^{(L)}$ is the canonical function of the geometry defined as the $L$-fold product of one-loop Amplituhedra 
$\mathcal{A}^{(1)}_n$, with the extra conditions $\langle AB_\ell AB_{\ell'} \rangle < 0 $ for each edge $\ell - \ell'$ in $\Gamma$. 

In order to obtain $\widetilde{\Omega}_n^{(L)}$, which is the integrand for the logarithm of the amplitude \eqref{log_ampl}, we take the logarithm of \eqref{omega_in_neg_geom}. By the exponential formula for generating functions, the summation in \eqref{omega_in_neg_geom} is reduced to \textit{connected graphs},
\begin{equation}\label{omega_tilde_in_neg_geom}
    \widetilde{\Omega}^{(L)}_n = \sum_{\Gamma \, \text{connected}} (-1)^{E(\Gamma)} \, \Omega_{n,\Gamma}^{(L)} \,.
\end{equation}
In the definition of $F_n^{(L)}$, one loop variable is left unintegrated. We denote the latter by $AB$, and it is a marked node in $\Gamma$.
The summation can then be taken over connected graphs with a marked node, where each graph is be normalized by a symmetry factor associated to its automorphism group.

As an example, let us look at $n=4$. Then $\widetilde{\Omega}^{(1)}_4$ is equal to the canonical function of the four-point one-loop Amplituhedron \eqref{box_integrand}.
At two loops, we have 
\begin{equation}\label{n4_L2_int}
    \widetilde{\Omega}^{(2)}_4 =\frac{{-}\langle 1234 \rangle^3 (\langle AB 13 \rangle \langle CD 24 \rangle  + \langle AB 24 \rangle  \langle CD 13 \rangle ) }{\langle AB 12\rangle \langle AB 23\rangle \langle AB 34\rangle \langle AB 14\rangle \langle ABCD \rangle \langle CD 12\rangle \langle CD 23\rangle \langle CD 34\rangle \langle CD 14\rangle} \,,
\end{equation}
where $CD$ is the second loop variable.

\subsection{Leading singularities}

In the following we will be interested in the \textit{leading singularities} (LS) of the Wilson loop with Lagrangian insertion \eqref{WL_int}. More generally, we consider integrals in momentum twistor space of the same form as \eqref{WL_int}, involving a rational integrand. For such an integral, we define a LS to be a maximal residue obtained by integrating the integrand along any contour homotopic to a topological torus enclosing $4L$ poles in the loop variables $AB_1$, \dots, $AB_{L}$. 

From the geometric point of view, a LS of \eqref{WL_int} is the residue of the canonical function $\widetilde{\Omega}_{n}^{(L+1)}$ on any four-dimensional boundary obtained by localizing all the variables $AB_1, \dots, AB_{L}$ and keeping $AB$ unconstrained. 

The following remark is important for what follows. Generally, there can be many homologically inequivalent contours of the form described above. However, some can yield the same leading singularity. Correspondingly, different ordered residues of $\widetilde{\Omega}_{n}^{(L+1)}$ can yields the same result.
\begin{definition}\label{definition of LS}
We distinguish between:
\begin{itemize}
    \item \textit{leading singularity value}, or simply leading singularity, i.e. the algebraic function in $Z_i$ and $AB$ defined as a (maximal) residue, and
    \item \textit{leading singularity configuration}, which refers to a solution for the loop lines $AB_1$, \dots, $AB_{L}$ where the integrand has a pole of (maximal) order $4L$.
\end{itemize}
\end{definition}

We associate an integral to each negative geometry in the expansion \eqref{omega_tilde_in_neg_geom}. More precisely, to a given connected graph $\Gamma$ with ${L+1}$ nodes and a marked node, corresponding to the loop variable $AB=AB_0$, we associate the integral
\begin{equation}\label{F_integral}
    \mathcal{F}_{n,\Gamma}^{(L)} := \int_{AB_1}  \dots \int_{AB_{L}} \widetilde{\mathbf{\Omega}}_n^{(L+1)} \,,
\end{equation}
where $\widetilde{\mathbf{\Omega}}_n^{(L+1)}$ is the canonical form of the $n$-point negative geometry associated to $\Gamma$. In the following, we will also consider LS of individual negative geometries \eqref{F_integral}.

Let us give simple examples where one can see the distinction between LS value and LS configuration. We consider $F_4^{(1)}$, and show that two distinct LS configurations in $CD$ yield the same residue in \eqref{n4_L2_int}, i.e. the same LS value. For that,
performing an analysis as in ref. \cite{Arkani-Hamed:2010pyv}, it is possible to compute all maximal residues in $CD$ of \eqref{n4_L2_int}, i.e. where $CD$ is fully localized.
For example, consider the configurations $CD=13$ and $CD=24$, which correspond to zero-dimensional boundaries of the one-loop Amplituhedron in $CD$.
For the first case, we compute the contour integral
\begin{equation}
    \oint_{CD=13} \frac{\la 1234 \ra   \langle AB 13 \rangle \langle CD 24 \ra }{ \langle ABCD \rangle \langle CD 12\rangle \langle CD 23\rangle \langle CD 34\rangle \langle CD 14\rangle} = \pm \, 1 \,,
\end{equation}
where the sign depends on the orientation of the integration cycle. 
(An explanation on how to take residues of rational functions in momentum twistor variables can be found in Appendix A in \cite{Arkani-Hamed:2010pyv}.)
Taking the same contour integral on \eqref{n4_L2_int}, one obtains
\begin{equation}
    \oint_{CD=13} \widetilde{\Omega}^{(2)}_4 = \pm \, \Omega^{(1)}_4 \,.
\end{equation}
The same computation goes through if instead we localize  $CD=24$. This shows that the two one-loop LS configurations where $CD=13$ and $CD=24$ yield the same (up to sign) LS value $\Omega^{(1)}_4$.

\subsection{The structure of the Wilson loop in perturbation theory}

One of the motivations of this work is to leverage the geometric definition of the loop integrand to make predictions about the integrated answer.
\begin{tcolorbox}[colback=white]
On general grounds, we expect the Wilson loop to have the following decomposition,
\begin{equation}\label{F_decomposition}
    F_n^{(L)} = \sum_s \Omega_{n,s} \, f_{n,s}^{(L)} \,,
\end{equation}
where the functions $f_{n,s}^{(L)}$ are pure of transcendental degree $2L$, and $\Omega_{n,s}$ are algebraic prefactors. Note that the basis of the latter may depend on the loop order $L$.
\end{tcolorbox}

For integrals that evaluate to multi-polylogarithmic functions, a decomposition of the form of eq. \eqref{F_decomposition} follows from the structure of the differential equation they satisfy \cite{Henn:2013nsa}. Closely related to this, from the point of view of the integrand, such a decomposition follows from the possibility of finding a basis of local integrands with unit leading singularities \cite{Arkani-Hamed:2010pyv}, against which the full integrand can be decomposed. 
This suggests that the set of leading singularities of the integrand of $F_n$ corresponds to the set of algebraic prefactors $\Omega_{n,s}$.

Eq. (\ref{F_decomposition}) is also the starting point for many bootstrap approaches, where knowledge about the class of functions $f$, as well as the prefactors $\Omega$, is crucial. 
The main goal of the present paper is to determine a basis of LS valid at any loop order.

Let us illustrate explicit examples of the decomposition \eqref{F_decomposition} and of LS at $n=4,5$. 
\begin{eg}
At $n=4$ and $L=1$, eq. \eqref{n4_L2_int} can be decomposed into $\Omega_4^{(1)}$ multiplied by an integrand involving $CD$ and therefore $\Omega_4^{(1)}$ is the only LS. Hence,
\begin{equation}
    F^{(1)}_4 = \Omega_4^{(1)} \, f^{(1)}_4 \,.
\end{equation} 
The transcendental function $f^{(1)}_4$ is computed to be
\begin{equation}
    f^{(1)}_4 = \int_{CD} \widetilde{\Omega}_4^{(2)} =  \log^2(z) + \pi^2  \,, \quad \text{with} \quad  z = \frac{\langle AB 12 \rangle \langle AB 34 \rangle }{\langle AB 14 \rangle \langle AB 23 \rangle } \,.
\end{equation}
In fact, we will prove that at all loop orders there exists only one LS, given by $\Omega_4^{(1)} = [1234]$. This matches with the known explicit result for $F_4^{(L)}$ up to $L=3$ \cite{Alday:2012hy,Alday:2013ip,Henn:2019swt}.
\end{eg}

\begin{eg}
  The pentagonal Wilson loop $F_5^{(L)}$ has been computed up to $L=2$ in ref.~\cite{Chicherin:2022zxo}. At one loop, there are five LS given by
\begin{equation}\label{R5,13}
    \Omega_{5}(13) = [1234]+ [123;145] \,,
\end{equation}
together with its four cyclic shifts. At two loops, there are six linearly independent LS. One can complete the set above to a basis by adding the tree level contribution
\begin{equation}
    \Omega_5^{(1)} = [1234] + [123;145] + [1345] \,.
\end{equation}
\end{eg}
In the examples above, the LS can be expressed as linear combinations of Kermit forms. This is a general feature that has been conjectured in \cite{Chicherin:2022bov}, which we prove in this paper.

\begin{table}[t!]
    \centering
    \begin{tabular}{ |c|c|c|c|c|c|c|c|c|c|c| } 
     \hline
      & $n$ & 5 & 6 & 7 & 8 & 9 & 10 & 11 & 12  \\ 
      \hline
       $L=1$ & $\tfrac{n(n-3)}{2}$ & 5 & 9 & 14 & 20 & 27 & 35 & 44 & 54 \\
      $L \geq 2$ & $\frac{(n-1)(n-2)^{2}(n-3)}{12}$ & 6 & 20 & 50 & 105 & 196 & 336 & 540 & 825 \\
     \hline
    \end{tabular}
    \caption{The number of linearly independent leading singularities of $F_n^{(L)}$, as conjectured in ref. \cite{Chicherin:2022bov}, and proven in the present work.}
    \label{count}
\end{table}

\section{Main results and their derivation }
\label{sec:mainderivation}

The main results of this paper are summarized in Subsection \ref{subsec:main_results}. The other subsections are dedicated to the the explaining how our proof works: Subsection \ref{subsec: taxonomy of LS} concerns the classification of leading singularities into different families, Subsection \ref{subsec: short outline of the proof} gives a short summary of the proof's structure and Subsection \ref{subsec: exclusion of incompatible LS} is about the Exclusion Result, Proposition \ref{proposition exclusion result}, one of the main ingredients in the proof.

\subsection{Main results}
\label{subsec:main_results}

The main result of this paper is the proof of the following conjecture from ref.~\cite{Chicherin:2022bov}.
\begin{tcolorbox}[colback=black!5!white]
\begin{claim}\label{claim 1}
    All leading singularities of $F_n^{(L)}$ for $n \geq 4$ and $L \geq 1$ can be expressed as linear combinations of Kermit forms \eqref{six_invariant}. 
\end{claim}
\end{tcolorbox}
From Claim \ref{claim 1} we deduce the following important corollaries.
\begin{itemize}
    \item The linear space generated by all LS at fixed $n$ coincides with the space generated by all Kermit forms from $L \geq 2$. In particular, the dimension of this space saturates at $L=2$. The count is presented in Table \ref{count}; the proof of this last part is in Appendix~\ref{subapp:LS count}.
    \item In the frame where $AB$ is mapped to the \textit{infinity twistor} $I_\infty := \epsilon^{\dot{\alpha}\dot{\beta}}$, i.e. for $AB \rightarrow I_\infty$, all LS of $F_n^{(L)}$ multiplied by the Parke-Taylor factor \eqref{PT} are conformally invariant, see Proposition \ref{proposition conf inv} in Subsection~\ref{subsec: conformal invariance}.
\end{itemize}

Claim \ref{claim 1} itself is a consequence of the following Claim \ref{claim 2}, which we state here with the intent of connecting to the mathematics literature on the $m=2$ Amplituhedron \cite{Parisi:2021oql,Lukowski:2019sxw}.

\begin{tcolorbox}[colback=black!5!white]
\begin{claim}\label{claim 2}
Every LS of $F_n^{(L)}$ for $n \geq 4$ and $L \geq 1$ is the sum of canonical functions of a one-loop Amplituhedra $\mathcal{A}^{(1)}_n$ parametrised by $AB$, with additional constraints on the signs of twistor brackets $\langle AB i_r j_r \rangle $, for a collection of non-crossing arcs $\{(i_rj_r)\}$ of an $n$-gon, see Figure \ref{LS_summary}.
\end{claim}
\end{tcolorbox}

\begin{figure}[t!]
\centering
\begin{tikzpicture}[scale = 0.5]
    \draw[thick] (0,0) circle(3cm);

    \coordinate (i) at (-130:3cm);
    \coordinate (j) at (110:3cm);
    \coordinate (k) at (60:3cm);
    \coordinate (l) at (-50:3cm);
    
    \coordinate (a) at (100:3cm);
    \coordinate (b) at (90:3cm);
    \coordinate (c) at (75:3cm);
    \coordinate (d) at (-60:3cm);
    \coordinate (e) at (-80:3cm);
    \coordinate (f) at (-90:3cm);
    \coordinate (g) at (-110:3cm);

    \node at (i) [below left] {$i$};
    \node at (j) [above left] {$j$};
    \node at (k) [above right] {$k$};
    \node at (l) [below right] {$l$};

    \draw[line width=0.4mm, red, thick] (i) -- (j);
    \draw[line width=0.4mm, red, dashed] (k) -- (l);
    
    \draw[line width=0.4mm, red,  dashed] (a) -- (f);
    \draw[line width=0.4mm, red,  thick] (c) -- (d);

     \foreach \p in {i,j,k,l,a,c,d,f} {
        \fill (\p) circle(2.5pt);
    }

\end{tikzpicture}

\centering
\caption{Representation of a one-loop Amplituhedron geometry with extra conditions ${\langle AB ij \rangle < 0}$, if the arc $(ij)$ of the $n$-gon (represented by a circle) is plain, and $\langle AB kl \rangle > 0$ if dashed. Every LS of the Wilson loop with Lagrangian insertion is the sum of canonical functions of geometries like this one. We evaluate the latter in Section \ref{subsec:from part triang to Kermits} 
as a combination of Kermit forms.}
\label{LS_summary}
\end{figure}
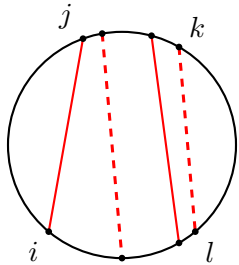

\subsection{Taxonomy of leading singularity configurations}
\label{subsec: taxonomy of LS}

In this paper we study leading singularities of $F_n^{(L)}$ from a geometric viewpoint. Thanks to the Amplituhedron-like construction given by the negative geometry expansion \eqref{omega_tilde_in_neg_geom}, every LS configuration corresponds uniquely to a configuration of $L+1$ lines $AB$, $AB_\ell$ and $n$ points $Z_i$ in three-dimensional projective space $\mathbb{P}^3$, with $\ell =1, \dots,L$ and $i=1, \dots, n$. In our setting, the data $AB$ and $Z_i$ are fixed, and constitute the external kinematics. The $Z_i$ satisfy external positivity \eqref{external_positivity}, and $AB$ lies in the one-loop Amplituhedron $\mathcal{A}^{(1)}_n$. Different LS configurations are obtained by fully localizing each loop line $AB_\ell$, i.e. by fixing each of its four degrees of freedom, by imposing boundary conditions related to $AB$ and the $Z_i$'s according to the loop-Amplituhedron's definition. In general, the location of a loop line $AB_\ell$ in a LS configuration depends on $AB$ and the $Z_i$'s. In particular, a bracket of the form $\langle AB_\ell \, AB_{\ell'} \rangle$ is a function of $AB$, where $AB$ lies in the one-loop Amplituhedron $\mathcal{A}^{(1)}_n$.

In this subsection we organize all LS configurations of relevant to individual negative geometries \eqref{F_integral} into different cases, as follows. This will be useful for our proof.
%
\begin{definition}\label{definition 1}
We divide all LS configurations of eq. \eqref{F_integral} into the three families: 
    \begin{enumerate}
    \item \textit{Simple configurations}: every loop line $AB_\ell$, for $\ell = 1, \dots, L$, localizes to ${AB_\ell = i_\ell j_\ell}$, with the indices $i_\ell < j_\ell$ taking values in $\{1,\dots,n\}$, see Figure \ref{fig:simple_LS}).
    \item \textit{Non-simple configurations}: the configuration is defined as not being simple, i.e. at least one loop line if different from $ij$. We further subdivide this case according to the relative configuration of the loop lines $AB_\ell$, and the distinguished loop line $AB$.\\
    a) If $\la AB \, AB_\ell \ra \neq 0 $ for every $\ell = 1, \dots, L$, we call such configurations
    {\it \LABELNONSIMPLE}. \\
    b) If instead at least one of the loop lines intersects the distinguished line, i.e. $\la AB_\ell \, AB \ra = 0$ for at least one $\ell$, then we call such configuration {\it \LABELREVERSE}.
\end{enumerate}
\end{definition}

Note that simple LS configurations are analogous to maximal codimension boundaries of the 
loop Amplituhedron, see \cite{Arkani-Hamed:2013kca}. In particular, we can represent each such configuration where $AB_\ell = i_\ell j_\ell$ by a collection of $L$ arcs $(i_\ell j_\ell)$ on an $n$-gon. Note that the arcs are allowed to coincide. On the other hand, {\LABELNONSIMPLE} LS are peculiar of negative geometries, 
i.e. they essentially arise because of replacing mutual positivity conditions with mutual negativity conditions between loops. Lastly, {\LABELREVERSE} configurations, arise because of the presence of the unintegrated loop $AB$, which has a distinguished role compared to all other loop lines.

\begin{definition}\label{definition 2}
    We partition LS configurations into the following two classes: 
\begin{enumerate}[label=(\alph*)]
    \item \textit{Compatible configurations}: every pair of loop lines in the configurations satisfies ${\langle AB_\ell \, AB_{\ell'} \rangle \geq 0}$, for every $AB \in \mathcal{A}^{(1)}_n$.
    \item \textit{Incompatible configurations}: all others, i.e. for which at least one pair of loops satisfy ${\langle AB_\ell \, AB_{\ell'} \rangle < 0}$, for at least one point $AB \in \mathcal{A}^{(1)}_n$.
\end{enumerate}
\end{definition}
This terminology is motivated by the fact that incompatible LS configurations can be discarded when considering LS of the Wilson loop $F_n^{(L)}$, as will be sown in Proposition \ref{proposition exclusion result}. Our notion of compatibility is also tightly related to that of cluster algebras of type $A_n$, which is essentially encoded in the non-crossing of arcs on an $n$-gon. In fact, a simple LS configuration is compatible, if and only if it is represented by a collection of non-crossing arcs, i.e. a \textit{partial triangulation}, of the $n$-gon, see Figure \ref{fig:simple_LS}. This follows immediately from external positivity \eqref{external_positivity}, and is summarized as follows.
\begin{tcolorbox}[colback=black!5!white]
   Every simple compatible LS configuration is associated to a partial triangulation of the $n$-gon, as Fig. \ref{fig:simple_LS}(a).
\end{tcolorbox}

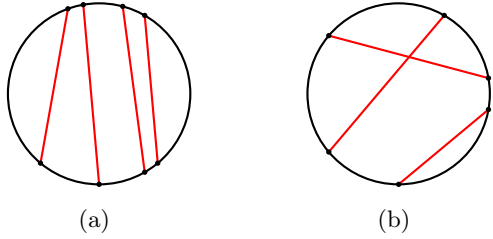
\begin{figure}[t]
\centering
\subfigure[]{
\begin{tikzpicture}[scale = 0.4]
    \draw[thick] (0,0) circle(3cm);

    \coordinate (i) at (-130:3cm);
    \coordinate (j) at (110:3cm);
    \coordinate (k) at (60:3cm);
    \coordinate (l) at (-50:3cm);
    
    \coordinate (a) at (100:3cm);
    \coordinate (b) at (90:3cm);
    \coordinate (c) at (75:3cm);
    \coordinate (d) at (-60:3cm);
    \coordinate (e) at (-80:3cm);
    \coordinate (f) at (-90:3cm);
    \coordinate (g) at (-110:3cm);

    \draw[line width=0.4mm, red, thick] (i) -- (j);
    \draw[line width=0.4mm, red, thick] (k) -- (l);
    
    \draw[line width=0.4mm, red,  thick] (a) -- (f);
    \draw[line width=0.4mm, red,  thick] (c) -- (d);

     \foreach \p in {i,j,k,l,a,c,d,f} {
        \fill (\p) circle(2.5pt);
    }

\end{tikzpicture}
}
\hspace{1cm}
\subfigure[]{
\begin{tikzpicture}[scale = 0.4]
    \begin{scope}[xshift = 8.5cm]
    \draw[thick] circle(3cm);


    \coordinate (x) at (-140:3cm);
    \coordinate (y) at (60:3cm);
    \coordinate (w) at (140:3cm);
    \coordinate (z) at (10:3cm);
    \coordinate (r) at (-10:3cm);
    \coordinate (s) at (-90:3cm);

    \draw[line width=0.4mm, red,  thick] (x) -- (y);
    \draw[line width=0.4mm, red,  thick] (w) -- (z);
    \draw[line width=0.4mm, red,  thick] (r) -- (s);

    \foreach \p in {x,y,w,z,r,s} {
        \fill (\p) circle(2.5pt);
    }

    \end{scope}
\end{tikzpicture}
}

\centering
\caption{Representations of two simple LS configurations, in which loop lines correspond to arcs of the $n$-gon. By definition \ref{definition 2}, (a) is a compatible configuration, as all the arcs are non-crossing, while (b) is an incompatible configuration.}
\label{fig:simple_LS}
\end{figure}

\subsection{Short outline of the proof}
\label{subsec: short outline of the proof}

Now that we subdivided all leading singularity configurations in different categories, we are ready to present an outline of the proof of Claim \ref{claim 1} and Claim \ref{claim 2}.
The proof is achieved via the following three steps.

{\it{\underline{Step 1: The exclusion result.}}} The first ingredient is the following, which guarantees that when considering LS values of $F_n^{(L)}$, all incompatible LS configurations are irrelevant.

\begin{tcolorbox}[colback=black!5!white]\begin{proposition}[Reduction Result]\label{proposition exclusion result}
   Incompatible leading singularity configurations of $F_n^{(L)}$ either evaluate to zero, or reduce to a sum of LS values of lower-loop compatible LS configurations. 
\end{proposition}
\end{tcolorbox}
 
It is therefore sufficient to restrict ourselves to considering compatible LS configurations. Moreover, we can discard all {\LABELNONSIMPLE} LS configurations, thanks to the following Lemma, whose proof is found  in Appendix~\ref{app: proof lemma 1}.
\begin{tcolorbox}[colback=black!5!white]
\begin{proposition}[Classification result]\label{lemma 1}
    All {\LABELNONSIMPLE} leading singularity configurations are incompatible.
\end{proposition}
    
\end{tcolorbox}

{\it{\underline{Step 2: Non-simple intersecting compatible LS.}}} 
Thanks to the result of step 1, we can restrict our attention to compatible LS configurations only, which can be either simple or {\LABELREVERSE}. 
The former are easy to understand: they are associated to partial triangulations of the $n$-gon, see Figure \ref{fig:simple_LS}. 
The second main ingredient provides a better understanding of
{\LABELREVERSE} compatible LS, via the following Proposition.
\begin{tcolorbox}[colback=black!5!white]
\begin{proposition}[Signed Partial Triangulations]\label{proposition reverse compatible are triang}
    All {\LABELREVERSE}, compatible leading singularities correspond to unions of \textit{signed partial triangulations} of an $n$-gon.
\end{proposition}
\end{tcolorbox} 
The proof of this proposition is given in Appendix \ref{app: proof prop 2}. 

{\it{\underline{Step 3: From signed triangulations to Kermits.}}} 
As a result of steps 1 and 2, we identified every relevant LS configuration of $F_n^{(L)}$ with a signed partial triangulation.
Lastly, we associate to each such configuration a one-loop Amplituhedron geometry, with extra conditions \eqref{geom of traing}. 
The LS value is given by the canonical function of this geometry. 
The third main ingredient is the following, which shows that every such geometry can be tiled into Kermits.

\begin{tcolorbox}[colback=black!5!white]
\begin{proposition}[From Triangulations to Kermits]\label{proposition from triang to kermits}
    For every signed partial triangulation $T$ of the $n$-gon, the geometric space \eqref{geom of traing} can be tiled into Kermits.
\end{proposition} 
\end{tcolorbox}
This follows from the complete understanding of tiles and tilings for the $m=2$ Amplituhedron \cite{Parisi:2021oql}. Together, this shows that every LS value of $F_n^{(L)}$ can be expressed as a sum of Kermit forms, which is the content of Claim \ref{claim 1}.

Lastly, we prove a conjecture from ref. \cite{Chicherin:2022bov} about conformal properties of leading singularities of $F_n^{(L)}$ in the frame where the Lagrangian insertion point is mapped to infiity.
\begin{tcolorbox}[colback=black!5!white]
\begin{proposition}[Conformal Invariance]\label{proposition conf inv}
In the frame where $AB \rightarrow I_\infty$, all leading singularities of $F_n^{(L)}$, or equivalently, all Kermit forms \eqref{six_invariant}, multiplied by the Parke-Taylor factor \eqref{PT} are conformally invariant.
\end{proposition}
\end{tcolorbox}

\subsection{Reduction of incompatible leading singularities}
\label{subsec: exclusion of incompatible LS}

In this section we prove Proposition \ref{proposition exclusion result} (Reduction Result). The idea is that even though individual negative geometries appearing in eq. (\ref{omega_tilde_in_neg_geom}) may have arbitrarily complicated LS, most of them either cancel or simplify in the sum.
Proposition \ref{proposition exclusion result}
implies in particular that we can discard all incompatible LS when considering the full Wilson loop. Hence, only compatible (simple and {\LABELREVERSE}) LS configurations are relevant for $F_n^{(L)}$.

\begin{proof}
Let $\mathcal{L} = \{AB_1, \dots, AB_L\}$ be an incompatible LS configuration of some negative geometry \eqref{F_integral}. Let us consider the set $\mathcal{G}_L$ of connected graphs with $L+1$ labeled nodes, with labels $0,1,\dots,L$ corresponding to the loop lines $AB,AB_1, \dots,AB_L$. Note that \eqref{omega_tilde_in_neg_geom} is in fact a sum over all $\Gamma \in \mathcal{G}_L$. By assumption, there exist $AB_\ell, AB_{\ell'}\in \mathcal{L} $ such that
\begin{equation}\label{mutual_negativity}
    \la AB_\ell \, AB_{\ell'} \ra < 0 \,,
\end{equation}
for at least one point $AB \in \mathcal{A}^{(1)}_n$.
The integral $\mathcal{F}_\Gamma$ in \eqref{F_integral} associated to any graph $\Gamma \in \mathcal{G}_L$ that contains the edge $e:=\ell-\ell'$, yields the same LS value as the integral of $\Gamma \setminus \{e\}$\footnote{In general, multi-loop Amplituhedra have also non-unit residues,
see ref. \cite{Dian_2023}. Nevertheless, the presence of the edge $e$ does not affect the 
normalization of the LS, since by assumption the localization of $AB_{\ell}$ and $AB_{\ell'}$ lies away from the boundary $\langle AB_{\ell} \, AB_{\ell'} \rangle = 0$.}. However, the latter is well-defined only if $\Gamma \setminus \{e\}$ is still connected. Note that the LS value may be equal to zero. This happens if for example $\Gamma$ does not contain an edge $a-b$ corresponding to a multiloop cut $\la AB_a \, AB_b \ra = 0$ needed to build the configuration~$\mathcal{L}$. Nevertheless, if the LS vanishes for $\mathcal{F}_\Gamma$, so it does for $\mathcal{F}_{\Gamma \setminus \{e\}}$. Also the converse is true, if the LS is supported by $\mathcal{F}_\Gamma$, so it is by $\mathcal{F}_{\Gamma \setminus \{e\}}$, since by \eqref{mutual_negativity} the cut associated to $e$ is not involved in the configuration $\mathcal{L}$.

The LS value of $F_n^{(L)}$ corresponding to $\mathcal{L}$ involves the sum \eqref{omega_tilde_in_neg_geom} over all LS values of $\mathcal{F}_\Gamma$ for $\Gamma \in \mathcal{G}_L$, each multiplied by a factor $(-1)^{E(\Gamma)}$, with $E(\Gamma)$ being the number of edges of $\Gamma$. By what we just discussed, the LS value cancels between $\Gamma$ and $\Gamma \setminus \{e\}$, if the latter is a connected graph. Therefore, the sum reduces to all graphs $\Gamma$ such that $\Gamma \setminus \{e\}$ is disconnected. Consider any such $\Gamma \in \mathcal{G}_L$. Then $\Gamma$ consists of two connected subgraphs $\Gamma_1$ and $\Gamma_2$, connected by the single edge $e$. Accordingly, let us partition $\mathcal{L} = \mathcal{L}_1 \cup \mathcal{L}_2$ into sets of loop lines associated to nodes of $\Gamma_1$ and $\Gamma_2$, respectively. Assume without loss of generality that the node labeled by $0$ belongs to $\Gamma_1$, i.e. that $AB$ is in $\mathcal{L}_1$. We claim that the LS value of $\mathcal{F}_\Gamma$ associated to $\mathcal{L}$, is the same as that of $\mathcal{F}_{\Gamma_1}$ associated to $\mathcal{L}_1$. This is equivalent to all loop lines in $\mathcal{L}_2$ being independent of $AB$. This is clear, from the fact that no line in $\mathcal{L}_2$ intersects $AB$, nor it intersects a line intersecting $AB$, because $\Gamma_1$ is connected to $\Gamma_2$ only through $e=\ell -\ell'$ and \eqref{mutual_negativity} holds true. By repeatedly applying this argument, we can remove from $\mathcal{L}$ all pairs of loops satisfying \eqref{mutual_negativity}. We can therefore reduce the incompatible configuration $\mathcal{L}$ to a compatible configuration $\mathcal{L}'$, such that the LS value of $\mathcal{F}_\Gamma$ associated to $\mathcal{L}$, is the same as that of $\mathcal{F}_{\Gamma'}$ associated to $\mathcal{L}'$. By applying this argument to all graphs $\Gamma$ such that $\Gamma \setminus \{e\}$ is disconnected, we prove the claim.

Technically, we proved the result only at a point $AB \in \mathcal{A}^{(1)}_n$ for which \eqref{mutual_negativity} holds true. However, since LS of $F_n^{(L)}$ are analytic in the twistor coordinates of $AB$ and $Z_i$'s (in fact, they are rational), and by continuity of (\ref{mutual_negativity}) in $AB$, our result extends to an open set inside the one-loop Amplituhedron in $AB$. It follows that our result holds true everywhere in $\mathcal{A}^{(1)}_n$.

\end{proof} 

\begin{eg}
The simplest example for which the exclusion result applies happens at $L=2$. Consider $n=5$ and $AB_1=13$ and $AB_2 = 24$.
This LS configuration is simple; moreover it is incompatible, since $\langle 1324 \rangle <0$.
The LS value, which we discuss in subsection \ref{subsec:examplesofleadingsingularitiesvalues},
cancels between the diagrams in the negative geometry expansion \eqref{omega_tilde_in_neg_geom} for $F_5^{(2)}$, which is given by
\begin{equation} \label{eq:2-loop}
\begin{tabular}{cc}
 \includegraphics[width=8cm]{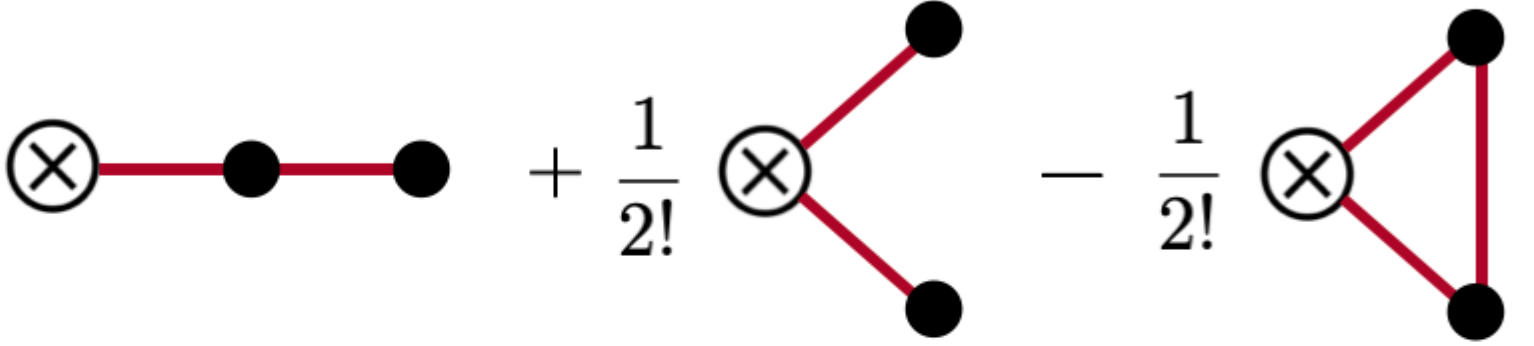}
\end{tabular}
\end{equation}
More precisely, the LS value in eq. \eqref{LSvaluesimpleincompatibleexample} cancels between the ladder with marked point in the middle, and the triangle, as they appear with the same coefficient but with opposite sign. The LS value of this configuration receives contributions only from the ladder with marked point at the extreme, i.e. the first term in eq. \eqref{eq:2-loop}. There are two labeled ladders, and the LS value is equal to $\Omega_5(13) + \Omega_5(24)$, see eq. \eqref{omega_in_Kermits_L1}. This LS is therefore equal to the sum of two LS associated to simple compatible configurations, as expected from Proposition \ref{proposition exclusion result}. 

\end{eg}

\begin{eg}\label{exp: L2 reverse incompatible}
\begin{figure}[t]
\centering
\begin{tikzpicture}[scale = 0.6]

    \coordinate (k) at (1,-2);
    \coordinate (P) at (0.17,0.05);
    \coordinate (Q) at (-1.12,-0.22);
    \coordinate (R) at (-1.25,2.18);
    \coordinate (S) at (-0.72,2.33);
    \coordinate (j) at (-2,2);
    \coordinate (j_plus) at (0,2.5);
    \coordinate (i) at (-1,-3);
    \coordinate (C) at (-1,3);
    \coordinate (D) at (1.4,-3);
    \coordinate (A) at (-2.5,-0.5);
    \coordinate (B) at (2.5,0.5);
    \coordinate (E) at (-1.3,3);
    \coordinate (F) at (-0.95,-3.8);


    \draw[dashed, gray] (j) -- (j_plus);


    \draw[red, thick] (C) -- (D);
    \node[right, red] at (C) {$AB_2$};

    \draw[red, thick] (E) -- (F);
    \node[left, red] at (E) {$AB_1$};

    \draw[teal, thick] (A) -- (B);
    \node[above, teal] at (B) {$AB$};

    \fill (i) circle (2pt);
    \fill (P) circle (2pt);
    \fill (Q) circle (2pt);
    \fill (R) circle (2pt);
    \fill (S) circle (2pt);
    \fill (j) circle (2pt);
    \fill (j_plus) circle (2pt);
    \fill (k) circle (2pt);

    \node[right] at (i) {$i$};
    \node[left] at (j) {$j$};
    \node[right] at (j_plus) {$j+1$};
     \node[right] at (k) {$k$};


    \begin{scope}[xshift = 8cm]
    \node at (4.5,0) {$\bigcup \quad \dots$};
    
    \draw[thick] (0,0) circle(2cm);

    \coordinate (i) at (-100:2cm);
    \coordinate (j) at (130:2cm);
    \coordinate (k) at (100:2cm);
    \coordinate (l) at (-20:2cm);

    \node at (i) [below ] {$i$};
    \node at (j) [above left] {$j$};
    \node at (k) [above ] {$j+1$};
    \node at (l) [below right] {$k$};

    \draw[line width=0.4mm, red, dashed] (i) -- (j);
    \draw[line width=0.4mm, red, thick] (i) -- (k);
    \draw[line width=0.4mm, red, thick] (l) -- (j);
    \draw[line width=0.4mm, red, dashed] (l) -- (k);
    \draw[line width=0.4mm, red, thick] (l) -- (i);

     \foreach \p in {i,j,k,l} {
        \fill (\p) circle(2.5pt);
    }
    \end{scope}

\end{tikzpicture}
\caption{An example of a {\LABELREVERSE} incompatible LS configuration at $L=2$, where the bracket $\la AB_1 \, AB_2 \ra $, which is a function of $AB$, does not have a fixed sign for $AB \in \mathcal{A}^{(1)}_n$.}
\label{fig:L2 reverse incompatible}
\end{figure}
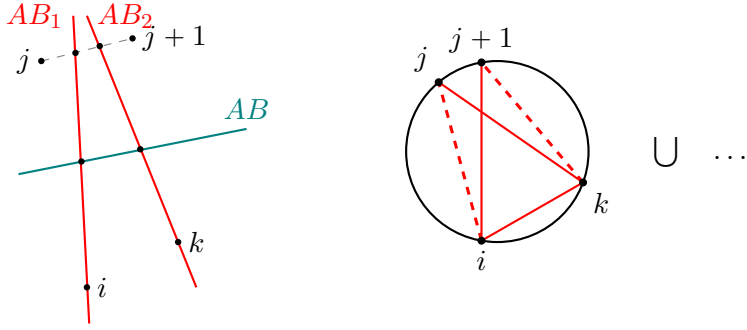

We now give an example that illustrates the relevance of the last paragraph in the proof of Proposition \ref{proposition exclusion result}. Generally, the twistor coordinate $\la AB_1 \, AB_2 \ra$ between two loops $AB_1$ and $AB_2$ in a LS configuration, is a function of $AB$. According to Definition \ref{definition 2}, if such a bracket takes negative values for $AB \in \mathcal{A}^{(1)}_n$, then the LS configuration is incompatible. The simplest example of this is at $L=2$, given by the configuration in Figure \ref{fig:L2 reverse incompatible}. The loops are localized as $ AB_1 = i(j + \alpha (j+1))$ and $AB_2=(j+ \beta \, (j+1))k$, where $\alpha,\beta$ depend on $AB$ as in eq. \eqref{alpha}. Then,
\begin{equation}\label{CDEF}
    \langle CDEF \rangle = (\beta - \alpha) \, \langle ijj+1 k \rangle 
\end{equation}
has no definite sign for $AB \in \mathcal{A}^{(1)}_n$. The last paragraph in the proof of Proposition \ref{proposition exclusion result} shows that the LS value associated to this configuration either vanishes or it reduces to that of compatible LS configurations.

\end{eg}

\section{Leading singularity values}
\label{sec: leading singularity values}

In this section we pass from LS configurations to LS values. More precisely, for every LS configuration relevant for $F_n^{(L)}$, i.e. for every simple and {\LABELREVERSE} compatible configuration, we explain how to associate to it a rational function in twistor variables in the $Z_i$'s and $AB$. In Section \ref{subsec:from part triang to Kermits} we associate a modified one-loop Amplituhedron space \eqref{geom of traing} to every simple and {\LABELREVERSE} compatible LS configuration.
We then show that any such space can be triangulated into Kermits, from which it follows that the associated LS value can be expressed as a linear combination of Kermit forms. In Section \ref{subsec: formulae for simple compatible LS} we give explicit formulae of simple compatible LS in terms of Kermit forms. Conversely, in Section \ref{subsec: Kermit forms in terms of simple compatible LS} we express all Kermit forms in terms of simple compatible LS, which implies in particular that the linear span of the former is the same as that of the latter.

\subsection{From partial triangulations to Kermit forms}
\label{subsec:from part triang to Kermits}

We first associate to each chord diagram a positive geometry leaving in the one-loop Amplituhedron in $AB$.

\begin{definition}
We define a \textit{signed partial triangulation} of an $n$-gon to be a subset $T$ of arcs of the $n$-gon, partitioned into two disjoint sets $T=T^{+} \cup T^{-}$, where $T^{+}$ contains dashed and $T^{-}$ plain arcs, as in Figure \ref{LS_summary}. For an arc $(ij) \in T$, we define $\text{sgn}_T(ij)$ to be $+1$ if $(ij) \in T^{+}$ and $-1$ if $(ij) \in T^{-}$. 
To a signed partial triangulation $T$, we associate a modified one-loop Amplituhedron geometry by
\begin{equation}\label{geom of traing}
    \mathcal{A}_n(T):=\{AB \in \mathcal{A}_{n}^{(1)} \ : \ \text{sgn}(ij) \, \langle AB ij \rangle > 0 \ , \quad \forall \, (ij) \in T\} \,.
\end{equation}
\end{definition}
Proposition \ref{proposition from triang to kermits} says that the space in \eqref{geom of traing} can be tiled into Kermits. In particular, the canonical function $\Omega_n(T)$ of $\mathcal{A}_n(T)$ can be written as a sum of Kermit forms \eqref{six_invariant}. Note that we allow for $\mathcal{A}_n(T) = \emptyset$, in which case $\Omega_n(T) = 0$. We now prove Proposition \ref{proposition from triang to kermits}, for which we rely on results about the $m=2$ Amplituhedron and its tiles \cite{Parisi:2021oql}.

\begin{proof}
Let $T$ be the signed partial triangulation of an $n$-gon. We complete $T$ to any full triangulation $\overline{T} = T \cup T'$, such that $T \cap T' = \emptyset$. Geometrically, $\mathcal{A}_n(T)$ can be written as the disjoint union
\begin{equation}
    \mathcal{A}_n(T) = \bigsqcup_{\overline{T}'} \mathcal{A}_n(\overline{T}') \,,
\end{equation}
where $\overline{T}'$ ranges over all possible signed triangulations obtained from $\overline{T} $ by assigning a sign to each arc in $T'$. It follows by \cite{Parisi:2021oql} that  $\mathcal{A}_n(\overline{T}')$ is nonempty if and only if there exists a bicolored triangulation $\mathcal{T} $ of type $(2,n)$ (depending on $\overline{T}'$) compatible with $\overline{T}'$. This means that each arc $(ij) \in \mathcal{T}$ is compatible with $\overline{T}$ and $(-1)^{\text{area}(ij)} = \text{sgn}_{\overline{T}'}(ij)$, where $\text{area}(ij)$ is the number of black triangles in $\mathcal{T} $ that lie on the left of the arc $(ij)$ when walking from $i$ to $j$.
If this is the case, $\mathcal{A}_n(\overline{T}') $ is in fact a Kermit, completely determined by $\mathcal{T}$. 
\end{proof}

\subsection{Formulae for simple compatible LS}
\label{subsec: formulae for simple compatible LS}

\begin{figure}[t]
\centering
\subfigure[]{
\begin{tikzpicture}[scale = 0.5]

    \draw[thick] (0,0) circle(3cm);

    \coordinate (i) at (-130:3cm);
    \coordinate (j) at (110:3cm);
    \coordinate (k) at (60:3cm);
    \coordinate (l) at (-50:3cm);
    
    \coordinate (a) at (100:3cm);
    \coordinate (b) at (90:3cm);
    \coordinate (c) at (75:3cm);
    \coordinate (d) at (-60:3cm);
    \coordinate (e) at (-80:3cm);
    \coordinate (f) at (-90:3cm);
    \coordinate (g) at (-110:3cm);

    \node at (i) [below left] {$i$};
    \node at (j) [above left] {$j$};
    \node at (k) [above right] {$k$};
    \node at (l) [below right] {$l$};

    \fill[black, opacity=0.3] (i) arc (230:110:3cm) -- (j) -- cycle;
    \fill[black, opacity=0.3] (l) arc (-50:60:3cm) -- (k) -- cycle;

    \draw[line width=0.4mm, red, thick] (i) -- (j);
    \draw[line width=0.4mm, red, thick] (k) -- (l);
    
    \draw[line width=0.4mm, red,  thick] (a) -- (f);
    \draw[line width=0.4mm, red,  thick] (a) -- (g);
    \draw[line width=0.4mm, red,  thick] (b) -- (e);
    \draw[line width=0.4mm, red,  thick] (c) -- (d);

     \foreach \p in {i,j,k,l,a,b,c,d,e,f,g} {
        \fill (\p) circle(2.5pt);
    }

    \node at (170:2.2cm) {$P_1$};
    \node at (10:2.3cm) {$P_2$};

\end{tikzpicture}
}
\hspace{1cm}
\subfigure[]{
\begin{tikzpicture}[scale = 0.5]
    \begin{scope}[xshift = 8.5cm]
    \draw[thick] circle(3cm);


    \coordinate (x) at (-140:3cm);
    \coordinate (y) at (130:3cm);
    \coordinate (w) at (100:3cm);
    \coordinate (z) at (30:3cm);
    \coordinate (r) at (20:3cm);
    \coordinate (s) at (-60:3cm);

    \draw[line width=0.4mm, red,  thick] (x) -- (y);
    \draw[line width=0.4mm, red,  thick] (w) -- (z);
    \draw[line width=0.4mm, red,  thick] (r) -- (s);

    \foreach \p in {x,y,w,z,r,s} {
        \fill (\p) circle(2.5pt);
    }

    \end{scope}

\end{tikzpicture}
}
\centering
\caption{(a) depicts the only type of nonzero LS with $|T| \geq 3$ arcs: the LS depends on the extremal arcs $(ij)$ and $(kl)$ bounding two polygons $P_1 $ and ${P_2}$ only. Eq. (\ref{LS_simple_formulae}) expresses the LS as a sum over Kermit forms, associated to (any) triangulations of $P_1$ and $P_2$, with one triangle in each. In particular, the LS depends only on the arcs $(ij),(kl)$.
(b) shows the other possible configuration for $|T| \geq 3$, for which the associated LS vanishes.}
\label{LS_configurations}
\end{figure}
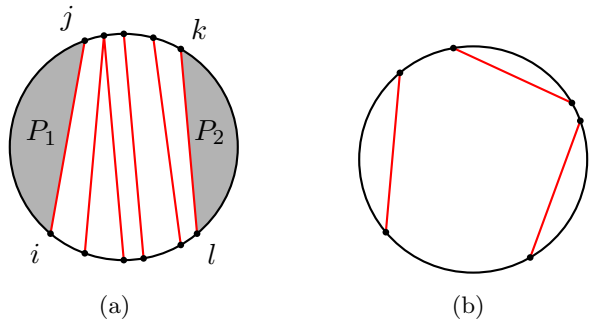

We give explicit formulae for simple compatible LS in terms of Kermit forms, and show that the former saturate at $L=2$. For that, let $T=\{(i_\alpha,j_\alpha)\}_{\alpha}$ be a partial triangulation of an $n$-gon associated to a simple compatible LS configuration, with $|T|$ arcs. We can regard $T$ as a signed partial triangulation by setting $T^{-}=T$ and $T^+= \emptyset$. We show that $\Omega_n(T) \neq 0 $, if and only if $|T| \leq 2$ or there exist two distinct arcs $(ij),(kl) \in T$ with $j \leq k$ such that for every other arc $(ab) \in T$ we have that either $l \leq a \leq i$ or $l \leq b \leq i$ but not both. The geometric meaning of this condition is explained in Figure \ref{LS_configurations}, where the for the picture on the left we have $\Omega_n(T) \neq 0 $ while for that on the right $\Omega_n(T) = 0 $.

Moreover, let us complete $T$ to a full triangulation $\overline{T}$ containing $T$, and denote by $P_1 = \{ i,i+1, \dots,j \}$ and $P_2= \{ k,k+1,\dots,l \}$ the polygons bounded by $(ij)$ and $(kl)$, see Figure \ref{LS_configurations}. Then, the following formula holds true,
\begin{align}\label{LS_simple_formulae}
    \Omega_n(ij,kl) &= \sum_{\substack{\Delta_1, \Delta_2 \subset \overline{T} \\ \Delta_1 \subset P_1 \,, \, \Delta_2 \subset P_2 }} [\Delta_1;\Delta_2] \,,
\end{align}
where the sum is over non-overlapping triangles $\Delta_1 = \{a_1,b_1,c_1\}$ and $\Delta_2 = \{a_2,b_2,c_2\}$ with arcs in $\overline{T}$. We allow $(ij)=(kl)$, for which we write simply $\Omega_n(ij)$ and \eqref{LS_simple_formulae} still makes sense. We also denote by $\Omega_n := \Omega_n(\emptyset ) = \Omega_n^{(1)}$, see \eqref{omega_in_kermits}.

It follows that for $|T| \geq 2$ the leading singularity $\Omega_n(T)$ depends only on the arcs $(ij),(kl)$ specified as above. Therefore all LS saturate at $|T|=2$, i.e. at loop level $L=2$.

\begin{proof}
Let $|T|=2$ and $T=\{(ij),(kl)\}$ be a partial triangulation. Complete $T$ to a full triangulation $\overline{T}$. It is clear that all bicolored triangulations of type $(2,n)$ compatible with some term $\overline{T}'$ in the expansion contain one colored triangle in each polygon $P_1$ and ${P_2}$, therefore (\ref{LS_simple_formulae}) for $|T|=2$ follows. 

Let $|T| \geq 3$. With the assumptions for $\Omega_n(T) \neq 0$, the same argument applies. If on the other hand the assumption is not satisfied, then there exist three arcs $(ij),(kl)$ and $(ab)$ such that w.l.o.g. $j \leq a$ and $b \leq k$. It is then straightforward to see that there does not exist any compatible bicolored triangulation of type $(2,n)$, see Figure \ref{LS_configurations}. 
  
\end{proof}

As an example, all simple compatible LS at $L=1$ are given by
\begin{equation}\label{omega_in_Kermits_L1}
    \Omega_n(ij) = \sum_{\substack{\Delta_1, \Delta_2 \subset \overline{T} \\ \Delta_1 \subset P_1 \,, \, \Delta_2 \subset P_2 }} [\Delta_1;\Delta_2] \,,
\end{equation}
where $P_1=\{i,i+1,\dots,j\}$, $P_2 = \{j,j+1, \dots , i\}$. These are indeed LS values of $F^{(1)}_n$.

\subsection{Kermit forms in terms of simple compatible LS}
\label{subsec: Kermit forms in terms of simple compatible LS}

We now express all Kermit forms in terms of simple compatible LS. This implies in particular that the linear space spanned by Kermit forms agrees with that of simple compatible LS. In particular, the linear space of all LS values of $F^{(L)}_n$ agrees with that of Kermits for $L \geq 2$.

We can express Kermit forms in terms of \eqref{LS_simple_formulae} as follows,
\begin{equation}\label{six_kermit_in_LS}
\begin{aligned}
    [ijk;abc] &= \Omega_n(ik,ac) - \Omega_n(ik,ab) - \Omega_n(ik,bc) - \Omega_n(ij,ac) - \Omega_n(jk,ac) \\
    & \quad + \Omega_n(ij,jk) + \Omega_n(ij,ab) + \Omega_n(ij,bc) + \Omega_n(jk,ab) + \Omega_n(jk,bc) + \Omega_n(ab,bc)   \,,
\end{aligned}
\end{equation}
for every $i<j<k \leq a < b < c \leq i$, and
\begin{equation}\label{four_kermit_in_LS}
\begin{aligned}
     [abcd] &= \Omega_n(ac)- \Omega_n(ab,ac) - \Omega_n(bc,ac) - \Omega_n(ac,cd) - \Omega_n(ac,ad) \\
    & \quad + \Omega_n(ab,bc) + \Omega_n(ab,cd) + \Omega_n(ab,ad) + \Omega_n(bc,cd) + \Omega_n(bc,ad) + \Omega_n(cd,ac)  \,,
\end{aligned}
\end{equation}
for every $a<b<c<d$. These formulae follow from an inclusion/exclusion principle analogous to that of the proof of Proposition \ref{proposition from triang to kermits}.

In the last two sections we showed how to evaluate LS values of $F_n^{(L)}$, and how to pass from signed partial triangulations to rational functions. The LS values can be explicitely expressed in terms of Kermit forms \eqref{six_invariant}, or in terms of simple compatible LS \eqref{LS_simple_formulae}.

\subsection{Conformal invariance of leading singularities at infinity}
\label{subsec: conformal invariance}

In this section we study the conformal properties of leading singularity values of $F_n^{(L)}$. In~\cite{Chicherin:2022bov} the authors argued that LS of $F_n^{(L)}$ are conformally invariant in the frame where the unintegrated loop is mapped to the infinity twistor, $AB \rightarrow I_{\infty}$. This surprising claim was supported by a Grassmann integral formula in the same paper, together with an argument from integrability. Here instead we present an independent proof of this fact, which only makes use of the fact the fact that the linear span of LS of $F_n^{(L)}$ agrees with that of Kermits \eqref{six_invariant}, or equivalently with that os simple compatible LS \eqref{LS_simple_formulae}.

More precisely, we study the action of the conformal generators in the frame $AB \rightarrow I_{\infty}$ given by
\begin{equation}\label{conf_gen}
    \mathbb{K}_{\alpha \dot{\alpha}} = \sum_{i=1}^{n} \frac{\partial^2}{\partial \lambda^{\alpha}_i \partial \tilde{\lambda}^{\dot{\alpha}}_i} \,.
\end{equation}
Let us look at the structure of the Kermit forms in this frame. Starting from \eqref{six_invariant}, substituting the external twistors $Z_i$ by $(\lambda_i, x_i \lambda_i)$ and evaluating the four-brackets as determinants, we compute
\begin{equation}\label{six_infty}
    [a_1 b_1 c_1 ; a_2 b_2 c_2] \Big|_{AB \rightarrow I_\infty} = \frac{\bigl( \langle a_1 b_1 \rangle \langle c_1a_2b_2c_2 \rangle + \langle b_1 c_1 \rangle \langle a_1a_2b_2c_2 \rangle + \langle c_1 a_1 \rangle \langle b_1a_2b_2c_2 \rangle \bigl)^2}{\langle a_1 b_1 \rangle \langle b_1 c_1 \rangle \langle a_1 c_1 \rangle \langle a_2 b_2 \rangle \langle b_2 c_2 \rangle \langle a_2 c_2 \rangle} \,.
\end{equation}
Analogously, eq. (\ref{four_invariant}) becomes
\begin{equation}\label{4_invariant_infty}
        [abcd] \Big|_{AB \rightarrow I_\infty} =  \frac{\langle abcd \rangle ^{2}}{\langle  ab \rangle \langle  bc \rangle \langle  cd \rangle \langle  da \rangle } \,.
\end{equation}
For these calculations we used the identity
\begin{equation}\label{four_bracket_general}
\begin{aligned}
    \langle ijkl \rangle & =  \langle i|x_i x_j | j\rangle \langle k l \rangle - \langle i|x_i x_k | k \rangle \langle j l \rangle + \langle i|x_i x_l | l \rangle \langle j k \rangle \\
    & \quad + \langle j |x_j x_k | k \rangle \langle i l \rangle - \langle j |x_j x_l | l \rangle \langle ik \rangle + \langle k|x_c x_l | l \rangle \langle ij \rangle \,,
\end{aligned}
\end{equation}
with the short-hand notation $\langle i | x y | j \rangle := \lambda_i^{\alpha} x _{\alpha \dot{\alpha}} y^{\dot{\alpha} \beta} \lambda_{j \, \beta}$, and index contractions are implemented with the Levi-Civita tensor. Also, 
\begin{equation}\label{identities}
    \begin{aligned}
        &\langle AB \, ij \rangle \Big|_{AB \rightarrow I_\infty} = \langle i j \rangle \ , \\
        &\ang{AB \, (i{-}1 i i{+}1)\cap (j{-}1 jj{+}1)} \Big|_{AB \rightarrow I_\infty} = \sq{ij}\ang{i-1 i}\ang{ii+1}\ang{j-1j}\ang{jj+1}  \,.
    \end{aligned}
\end{equation}
For the following class of Kermit forms, conformal invariance easy to prove,
\begin{equation}\label{special_conf_inv}
     [i-1ii+1;j-1jj+1] \Big|_{AB \rightarrow I_\infty} = \frac{ \sq{ij}^2 \ang{i-1\, i}\ang{i\,i+1}\ang{j-1\,j}\ang{j\,j+1}  }{ \ang{i-1 i+1}  \ang{j-1 j+1} } \,.
\end{equation}
After multiplication with the Parke-Taylor factor \eqref{PT}, they depend on $\tilde{\lambda}_i, \tilde{\lambda}_j$ and $\lambda_k$ for $k \neq i,j$ and are therefore annihilated by (\ref{conf_gen}). Examples of such Kermit forms are $\Omega_5(13,15)=[123;145]$ and $\Omega_6(13,46)=[123;456]$. 
Proposition \ref{proposition conf inv}, which we prove presently, generalizes the above examples, and establishes conformal invariance of all Kermit forms.

\begin{proof}[Proof of Proposition \ref{proposition conf inv} (Conformal Invariance).]
Recall that by Claim \ref{claim 1}, together with the results from subsections \ref{subsec: formulae for simple compatible LS} and \ref{subsec: Kermit forms in terms of simple compatible LS}, it is equivalent to show conformal invariance of Kermits forms, of all LS of $F^{(L)}_n$, or of only simple compatible LS \eqref{LS_simple_formulae}. Here we focus on the latter, and show that they are annihilated by the conformal generators \eqref{conf_gen}, by induction over the number of points $n \geq 4$.
The base case is $n=4$, where by \eqref{omega_in_kermits} we have that $\Omega_4^{(1)} = [1234]$. In the frame $AB \rightarrow I_{\infty}$, this is a special case of \eqref{special_conf_inv},
\begin{equation}\label{1234}
   {\rm PT}_n \cdot [1234] \Big|_{AB \mapsto I_\infty} = \frac{[24]^2}{\la 13 \ra^2} \,. 
\end{equation}
Eq. \eqref{1234} is clearly annihilated by the conformal generators \eqref{conf_gen}.
In particular, for any $n>4$, any cyclic shift of $\Omega_4^{(1)}$, which takes the form of $[i-1ii+1i+2]$, is conformally invariant. This follows from the fact that \eqref{conf_gen} commute with cyclic shifts of the external particle's labels. 

We now perform a series of auxiliary computations that help building the full recursive proof.
Let us start with $n>4$ and write eq. \eqref{omega_in_Kermits_wr1} as
\begin{equation}\label{MHV_recursion}
\begin{aligned}
    \Omega_n &= \sum_{1<a<b<m}[1aa+1;1bb+1]  + \sum_{1<a<m \leq b <n} [1 a a+1;1bb+1] + \sum_{m \leq a<b<n } [1 a a+1;1bb+1] \\
    &= \Omega_m + \widetilde{\Omega}_{n-m+2} + \Omega_n(1m)  \,,
\end{aligned}
\end{equation}
which holds for every $m$ with $ 4 \leq m \leq n$ (note that by our definitions we have that $\Omega_n(1n)=0$, and we may define $\Omega_r = 0$ for $r<4$). The function $\widetilde{\Omega}_{n-m+2}$ is the same as $\Omega_{n-m+2}$, but with cyclically shifted indices $1 \rightarrow m, \, 2 \rightarrow m+1, \, n-m+2 \rightarrow 1$.

Then, let us consider a simple compatible LS \eqref{LS_simple_formulae} associated to a single arc. For $n=4$ we have $\Omega_4(13)=\Omega_4(24)=\Omega_4$, which we have shown to be conformally invariant. Let $n>4$ and consider w.l.o.g. $\Omega_n(1i)$ for $2<i<n$. We write \eqref{LS_simple_formulae} for $\overline{T} = \{(1j)\}_{j=3,\dots,n-2}$ as
\begin{equation}\label{LS_E_1_recursion}
\begin{aligned}
    \Omega_n(1i) &= \sum_{\substack{1<a<i \\ i \leq b < n}} [1 a a+1;1 b b+1] = \sum_{\substack{1<a<i \\ i \leq b < m}} [1 a a+1;1 b b+1] + \sum_{\substack{1<a<i \\ m \leq b < n-1 }} [1 a a+1;1 b b+1] \\
    &= \Omega_{m}(1i) + \Omega_n(1i,1m) \,,
\end{aligned}
\end{equation}
for any $m$ with $i \leq m <n$.

Lastly, we consider a simple compatible LS \eqref{LS_simple_formulae} associated to two arcs for $n \geq 5$. For $n=5$ we have up to cyclic shift only one LS of this type, given by $\Omega_5(13,14) = [123;145]$. This has the form of \eqref{special_conf_inv} and hence it is conformally invariant. For $n > 5$ consider two arcs $(1i), (jk)$ with $i \leq j < k \leq n$. Note that if $k<n$, then
\begin{equation}\label{trivial_recursion}
    \Omega_n(1i,jk) = \Omega_{n-1}(1i,jk) \,,
\end{equation}
so in this case conformal invariance follows directly by induction. So let us assume that $k=n$. Similarly, if $i+1<j$ we can use the cyclic symmetry and shift all the indices by $i+1$ to obtain a recursion as (\ref{trivial_recursion}) for the cyclically shifted configuration. Then, since (\ref{conf_gen}) commutes with cyclic shifts, (\ref{trivial_recursion}) still holds. Therefore, we additionally assume that $j \in \{i,i+1\}$ and rewrite (\ref{LS_simple_formulae}) 
with the inclusion-exclusion principle as 
\begin{equation}\label{LS_E_2_recursion}
\begin{aligned}
    \Omega_n(1i,jn) &= \Omega_n(1i) + \Omega_n(jn) -\Omega_n + \Omega_i + \widetilde{\Omega}_{n-j+1} + [1ijn] \ , \quad \text{for} \ j \in \{i,i+1 \} \,,
\end{aligned}
\end{equation}
where $\widetilde{\Omega}_{n-j+1}$ is equal to $\Omega_{n-j+1}$ with all indices cyclically shifted as $1 \rightarrow j, \, 2 \rightarrow j+1, \dots , \, n-j+1 \rightarrow n$. Note that $[1iin]=0$.

We are now ready to prove the claim by induction over $n$. The induction basis $n=4$ for \eqref{MHV_recursion} and $n=5$ for \eqref{LS_E_1_recursion} have been already discussed. Consider \eqref{LS_E_2_recursion}: for $j=i$ the induction basis is $n=5$ and it has been already discussed, while for $j=i+1$ the basis is $n=6$ and $\Omega_6(13,46)=[123;456]$ has the form of (\ref{special_conf_inv}), and hence it is conformally invariant. By combining \eqref{MHV_recursion}, \eqref{LS_E_1_recursion}, \eqref{trivial_recursion} and \eqref{LS_E_2_recursion} we can use the induction hypothesis to reduce the problem to showing conformal invariance of $[1ii+1n]$. For that, from \eqref{4_invariant_infty} we compute
\begin{equation}
    [1ii+1n] \Big|_{AB \rightarrow I_\infty} =  \frac{\langle 1ii+1n \rangle ^{2}}{\langle  1i \rangle \langle  ii+1 \rangle \langle  i+1n \rangle \langle  n1 \rangle } = \frac{ \ang{ii+1} \, \ang{1|x_{1,i+1}x_{i+1,n}|n}^2}{\langle  1i \rangle  \langle  i+1n \rangle \ang{n1}  } \,,
\end{equation}
where in the second equality we used the identity $\ang{ij-1jk}= \ang{j-1j} \, \ang{i|x_{ij}x_{jk}|k}$, which is a special case of (\ref{four_bracket_general}).
Therefore, we can write
\begin{equation}\label{missing_conf}
\begin{aligned}
    {\rm PT}_n \cdot \, &[1ii+1n] \Big|_{AB \rightarrow I_\infty}   
    = \frac{C^{(1)}_{\alpha \beta} \ C^{(2), \alpha \beta}}{\ang{n1}^2} \,,
\end{aligned}
\end{equation}
where
\begin{align}
C^{(1)}_{\alpha \beta} =  \frac{\bigl(\langle 1 | x_{1,i+1} \bigl)_{\alpha } \, \bigl(\langle 1 | x_{1,i+1}\bigl)_{\beta }}{\ang{12}\ang{23} \dots \ang{i-1 i} \ang{i1}} \,, \quad 
C^{(2), \alpha \beta} = \frac{\bigl(x_{i+1,n} |n\rangle \bigl)^{\alpha} \, \bigl(x_{i+1,n} |n\rangle \bigl)^{\beta}}{\ang{i+1i+2} \dots \ang{n-1n}\ang{ni+1}} \,.
\end{align}
Note that particles $\{1,2\}$ appear only in a holomorphic way. Also, the anti-holomorphic dependence of $C^{(1)}_{\alpha \beta}$ is only on the particles $\{2, \dots,i\}$, while that of $C^{(2),\alpha \beta}$ only on $\{i+1, \dots,n-1\}$. Therefore, the action of (\ref{conf_gen}) factorizes, and by the same computations presented below eq. (3.9) in \cite{Henn:2019mvc}, $C^{(1)}_{\alpha \beta}$ and $C^{(2),\alpha \beta}$ are individually conformally invariant. This shows that (\ref{missing_conf}) itself is conformal, which concludes the proof. 
   
\end{proof}

\subsection{Examples of different types of leading singularities}
\label{subsec:examplesofleadingsingularitiesvalues}

\begin{figure}[t]
\centering
\begin{minipage}[b][4cm][c]{0.3\textwidth}
\centering
    \subfigure[]{
\begin{tikzpicture}[scale = 1]
    \draw[thick] (90:1) 
        \foreach \x in {162,234,306,18} 
            { -- (\x:1) } -- cycle;
    \coordinate (D) at (18:1);
    \coordinate (C) at (90:1);
    \coordinate (B) at (162:1);
    \coordinate (A) at (234:1);
    \coordinate (E) at (306:1);
    \draw[red, thick] (A) -- (C);
    \node[below] at (A) {1};
    \node[left] at (B) {2};
    \node[above] at (C) {3};
    \node[right] at (D) {4};
    \node[below] at (E) {5};
 \end{tikzpicture}
}
\end{minipage}
\begin{minipage}[b][4cm][c]{0.3\textwidth}
\centering
        \subfigure[]{
\begin{tikzpicture}[scale = 1]
    \draw[thick] (90:1) 
    \foreach \x in {162,234,306,18} 
        { -- (\x:1) } -- cycle;
    \coordinate (D) at (18:1);
    \coordinate (C) at (90:1);
    \coordinate (B) at (162:1);
    \coordinate (A) at (234:1);
    \coordinate (E) at (306:1);
    \fill[black, opacity=0.3] (A) -- (B) -- (C) -- (D) -- cycle;
    \draw[red, dashed] (A) -- (D);
 \end{tikzpicture}
}
\end{minipage}
\begin{minipage}[b][4cm][c]{0.3\textwidth}
\centering
        \subfigure[]{
\begin{tikzpicture}[scale = 1]
    \draw[thick] (90:1) 
    \foreach \x in {162,234,306,18} 
        { -- (\x:1) } -- cycle;
    \coordinate (D) at (18:1);
    \coordinate (C) at (90:1);
    \coordinate (B) at (162:1);
    \coordinate (A) at (234:1);
    \coordinate (E) at (306:1);
    \fill[black, opacity=0.3] (A) -- (B) -- (C)  -- cycle;
    \fill[black, opacity=0.3] (A) -- (D) -- (E) -- cycle;
    \draw[red, thick] (A) -- (D);
    \draw[red, thick] (A) -- (C);
 \end{tikzpicture}
}
\end{minipage}
  \caption{(a): A simple, compatible LS configuration at $n=5$ and $L=1$. It can be decomposed in terms of the two Kermit geometries depicted in (b),(c).  The corresponding LS value, denoted by $\Omega_{5}(13)$, is given in eq. \eqref{omega_n5_13}. 
}
\label{fig:5pt simple 13}
\end{figure}
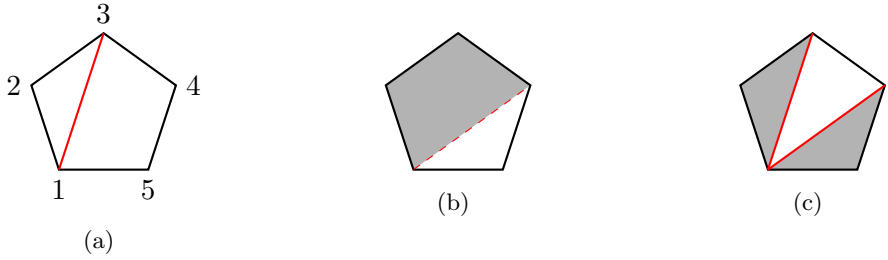

Here we provide explicit examples for each family of LS configurations defined in subsection~\ref{subsec: taxonomy of LS}, as well as the associated LS value. We also give examples of incompatible LS values, and provide evidence that in general they cannot be expressed in terms of Kermit forms, and that they are not conformally invariant in the frame $AB \mapsto I_\infty$. Finally, we illustrate that incompatible LS configurations can be very complicated, and involve loop lines localized to non-rational locations with respect to the external data $Z_i$ and $AB$.

\subsubsection{Examples of simple leading singularities}
\label{subsubsec: a simple LS}

For a concrete example of a simple LS, let $n=5$ and $L=1$. There is only one one-loop negative geometry, namely the ladder.
It turns out that there is only one non-zero LS configuration in this case, which corresponds to $AB_1= 13$. 
The geometry is then a one-loop Amplituhedron $\mathcal{A}^{(1)}_5$ in $AB$, with an extra condition $\la AB \, 13 \ra < 0$, cf. Fig.~\ref{fig:5pt simple 13}(a). Thanks to eq. \eqref{omega_in_Kermits_L1}, this can be expressed in terms of the
two Kermits shown in cf. Fig.~\ref{fig:5pt simple 13}(b),(c).
Hence the corresponding form is given by
\begin{equation}\label{omega_n5_13}
    \Omega_5(13) = [1234] + [123;145] \,.
\end{equation}
This is one of the LS appearing in $F^{(1)}_5$.

For an example of a simple, incompatible LS we have to move to $L=2$. Consider localizing $AB_1=13$ and $AB_2=24$.
The associated geometry is the one-loop Amplituhedron with extra conditions ${\la AB 13 \ra < 0}$ and ${\la AB 24 \ra < 0}$.
Its canonical form is given by 
\begin{equation}\label{LSvaluesimpleincompatibleexample}
{\tikz[baseline=.1ex]{
  \begin{scope}[xshift = 5 cm]
    \draw[thick] (90:1) 
        \foreach \x in {162,234,306,18} 
            { -- (\x:1) } -- cycle;
    \coordinate (D) at (18:1);
    \coordinate (C) at (90:1);
    \coordinate (B) at (162:1);
    \coordinate (A) at (234:1);
    \coordinate (E) at (306:1);
    \draw[red, thick] (A) -- (C);
    \draw[red, thick] (B) -- (D);
    \node[below] at (A) {1};
    \node[left] at (B) {2};
    \node[above] at (C) {3};
    \node[right] at (D) {4};
    \node[below] at (E) {5};
    \end{scope}
}}=
\frac{\ang{AB \,(123)\cap (451)}\ang{AB \,(234)\cap (451)}}{\ang{AB14}\ang{AB23}\ang{AB13}\ang{AB15}\ang{AB45}\ang{AB24}} + [1234] \,.
\end{equation}
Let us point out the following undesirable features of this LS value:
\begin{itemize}
    \item it cannot be expressed in terms of Kermit forms only; 
    \item in the frame $AB \rightarrow I_\infty$, it is not conformally invariant. Indeed, in this frame the first term in eq. (\ref{LSvaluesimpleincompatibleexample}) becomes 
\begin{equation}
    \frac{[25][35]}{\langle 13 \rangle \langle 1 4 \rangle \langle 2 4 \rangle} \,.
\end{equation}
One can show that this is gives a non-zero result when acting with the conformal generators of eq. (\ref{conf_gen}).
\end{itemize}
These features are characteristic of incompatible LS. However, thanks to the exclusion result given in Proposition \ref{proposition exclusion result}, these undesirable LS vanish for the full Wilson loop.

\subsubsection{An example of a {\LABELNONSIMPLE} leading singularity }

\begin{figure}[t]
    \centering
\begin{tikzpicture}[scale = 0.5]
    
    \begin{scope}

    \node at (0,-5) {$(a)$};

    \node at (5,-5) {\includegraphics[scale=0.5]{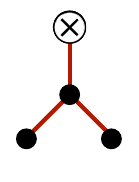}};

    \coordinate (k) at (1,-2);
    \coordinate (i) at (2.4,1.3);
    \coordinate (i_plus) at (0,2.5);
    \coordinate (l) at (-2,-0.77);
    \coordinate (C) at (-1,2);
    \coordinate (D) at (1.5,-3);
    \coordinate (E) at (-2.5,-1);
    \coordinate (F) at (2.8,1.5);
    \coordinate (G) at (-2.5,1);
    \coordinate (H) at (2.8,0.3);
    \coordinate (A) at (3,-2.5);
    \coordinate (B) at (5,-0.5);
    \coordinate (R) at (-0.09,0.14);
    \coordinate (S) at (0.7,-1.38);
    \coordinate (T) at (-0.35,0.7);
    \coordinate (U) at (-0.8,1.55);
    \coordinate (j) at (2.4,0.35);
    \coordinate (m) at (-2,0.92);



    \draw[red, thick] (C) -- (D);
    \node[above, red] at (C) {$AB_3$};

    \draw[red, thick] (E) -- (F);
    \node[left, red] at (E) {$AB_1$};

    \draw[red, thick] (G) -- (H);
    \node[left, red] at (G) {$AB_2$};

    \draw[teal, thick] (A) -- (B);
    \node[right, teal] at (B) {$AB$};

    \fill (i) circle (2pt);
    \fill (j) circle (2pt);
    \fill (l) circle (2pt);
    \fill (R) circle (2pt);
    \fill (T) circle (2pt);
    \fill (j) circle (2pt);
    \fill (m) circle (2pt);
    \fill (k) circle (2pt);

    \node[below left] at (k) {1};
    \node[above] at (i) {4};
    \node[below] at (l) {2};
    \node[below] at (j) {5};
    \node[below] at (m) {3};

    \end{scope}


    \begin{scope}[xshift = 15 cm]

    \node at (0,-5) {$(b)$};

    \node at (5,-5) {\includegraphics[scale=0.5]{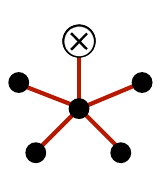}};

    \coordinate (k) at (1,-2);
    \coordinate (i) at (2.4,1.3);
    \coordinate (i_plus) at (0,2.5);
    \coordinate (l) at (-2,-0.77);
    \coordinate (C) at (-1,2);
    \coordinate (D) at (0.5,-4);
    \coordinate (E) at (-2.5,-1);
    \coordinate (F) at (2.8,1.5);
    \coordinate (G) at (-2.5,1);
    \coordinate (H) at (2.8,0.3);
    \coordinate (A) at (4,-2);
    \coordinate (B) at (6,-0);
    \coordinate (R) at (-0.5,-0.05);
    \coordinate (S) at (0.7,-1.38);
    \coordinate (T) at (-0.66,0.75);
    \coordinate (U) at (-0.8,1.55);
    \coordinate (j) at (2.4,0.35);
    \coordinate (m) at (-2,0.92);
    \coordinate (A3) at (-2.5,-4);
    \coordinate (B3) at (2.8,-1.5);
    \coordinate (A4) at (-2.5,-2);
    \coordinate (B4) at (2.8,-2.7);
    \coordinate (i1) at (-2,-3.8);
    \coordinate (i2) at (-2,-2.1);
    \coordinate (i3) at (2.3,-1.7);
    \coordinate (i4) at (2.3,-2.65);
    \coordinate (j1) at (0.1,-2.35);
    \coordinate (j2) at (0.2,-2.8);



    \draw[red, thick] (C) -- (D);
    \node[above, red] at (C) {$AB_5$};

    \draw[red, thick] (E) -- (F);
    \node[above left, red] at (E) {$AB_3$};

    \draw[red, thick] (G) -- (H);
    \node[left , red] at (G) {$AB_4$};

    \draw[teal, thick] (A) -- (B);
    \node[right, teal] at (B) {$AB$};

    \draw[red, thick] (A3) -- (B3);
    \node[left, red] at (A3) {$AB_1$};

    \draw[red, thick] (A4) -- (B4);
    \node[below left, red] at (A4) {$AB_2$};

    \fill (i) circle (2pt);
    \fill (j) circle (2pt);
    \fill (l) circle (2pt);
    \fill (R) circle (2pt);
    \fill (T) circle (2pt);
    \fill (j) circle (2pt);
    \fill (m) circle (2pt);
    \fill (i1) circle (2pt);
    \fill (i2) circle (2pt);
    \fill (i3) circle (2pt);
    \fill (i4) circle (2pt);
    \fill (j1) circle (2pt);
    \fill (j2) circle (2pt);

    \node[above] at (i) {7};
    \node[above] at (l) {5};
    \node[below] at (j) {8};
    \node[above] at (m) {6};
    \node[below] at (i1) {1};
    \node[below] at (i2) {2};
    \node[above] at (i3) {3};
    \node[below] at (i4) {4};

    \end{scope}

\end{tikzpicture}
    \caption{Two {\LABELNONSIMPLE} LS configurations, (a) at $L=3$ and $n \geq 5$, and (b) at $L=5$ and $n\geq 8$. Both are incompatible, as are all {\LABELNONSIMPLE} LS configurations. These are LS of the negative geometry topologies of an $L=3$ and $L=5$ star, respectively, which is also depicted.}
    \label{fig:non-simple LS}
\end{figure}
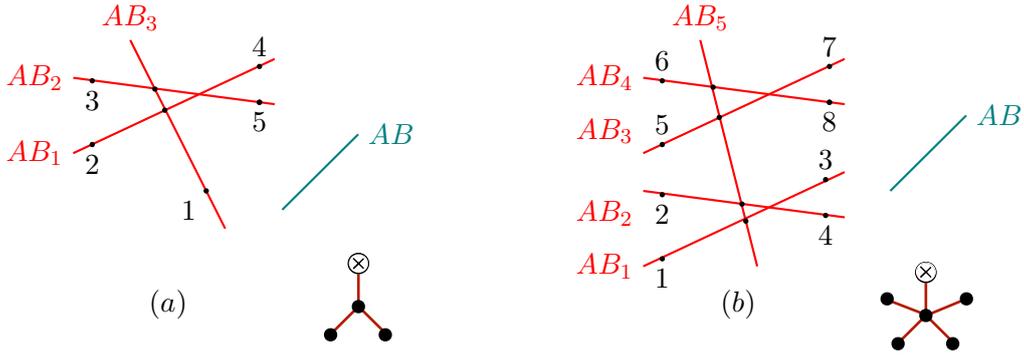

Consider the three-loop negative geometry whose underlying graph is a star with three nodes.
Moreover, we set $n=5$.
For the LS configuration we want to consider, let us localize $AB_1=24$, $AB_2=35$ and fix $AB_3$ by imposing
\begin{equation}\label{non-simple LS cuts}
  \la AB_3 \, AB_1\ra= \la AB_3 \, AB_2 \ra = \la AB_3 \, 51\ra = \la AB_3 \, 12\ra = 0 \,. 
\end{equation}
The Schubert problem \eqref{non-simple LS cuts} for $AB_3$ has two solutions (see e.g. \cite{Arkani-Hamed:2010pyv}), but only one lies inside the one-loop Amplituhedron, namely
\begin{equation}\label{loop_nonsimple nonintersecting}
    AB_3=(124){\cap}(135) \,.
\end{equation}
The resulting geometry in the unintegrated loop $AB$ is then a one-loop Amplituhedron with extra condition $\la AB \, AB_3 \ra < 0$. This space 
is a positive geometry, whose canonical function is given by
\begin{equation}\label{LS_star}
    \frac{\langle 1235 \rangle \langle 1245 \rangle \langle AB \, (123) \cap (451) \rangle }{\langle AB 12 \rangle \langle AB 23 \rangle \langle AB 45 \rangle \langle AB 51 \rangle  \langle AB (124) \cap (135) \rangle} \,.
\end{equation}
In other words, eq. (\ref{LS_star}) 
is the LS value associated to this configuration.
Note that the LS configuration is incompatible, since $\la AB_1 \, AB_2 \ra = \la 2435 \ra < 0$. In fact, the LS value given in eq. (\ref{LS_star}) has the same undesirable features as the incompatible simple LS of eq. (\ref{LSvaluesimpleincompatibleexample}). 

We remark that {\LABELNONSIMPLE} LS of individual negative geometries may be very complicated. In fact, there are LS configurations that involve loop lines that are localized to even non-rational locations. For instance, consider the negative geometry associated to the $L=5$ star for $n \geq 8$. 
Localizing $AB_1=13$, $AB_2=24$, $AB_3=57$, $AB_4=68$, and fixing $AB_5$ to be one of the two lines intersecting all other four $AB_\ell$, cf. Figure~\ref{fig:non-simple LS}(b).
We checked numerically that one solution for $AB_5$ lies in $\mathcal{A}^{(1)}_n$, and therefore the associated LS value is non-vanishing. Since the lines $AB_\ell$'s for $\ell = 1, \dots,4$ do not intersect each other, the parametrization of the solution for $AB_5$ is non-rational. This follows by the fact that solving this type of generic Schubert problem requires solving a non-trivial quadratic equation. 
The value of this LS is the canonical function of $\mathcal{A}^{(1)}_n$ in $AB$, with extra condition $\langle AB \, AB_5 \rangle < 0$.
At increasing loop order, one can iterate this procedure and e.g. use the solution just found for $AB_5$ to localize other loop lines. Therefore, {\LABELNONSIMPLE} LS configurations may become arbitrarily complicated.

\subsubsection{An example of a {\LABELREVERSE} leading singularity}
\label{subsubsec:example of reverse LS}

\begin{figure}[t]
\centering
\begin{tikzpicture}[scale = 0.5]

    \node at (0,-6) {$(a)$};

    \coordinate (i) at (1,-2);
    \coordinate (i_plus) at (0,0);
    \coordinate (j) at (-2,2);
    \coordinate (j_plus) at (0,2.5);
    \coordinate (C) at (-1.5,3);
    \coordinate (D) at (1.5,-3);
    \coordinate (A) at (-2.5,-0.5);
    \coordinate (B) at (2.5,0.5);
    \coordinate (R) at (-1.1,2.23);


    \draw[dashed, gray] (j) -- (j_plus);

    \draw[red, thick] (C) -- (D);
    \node[above, red] at (C) {$AB_1$};

    \draw[teal, thick] (A) -- (B);
    \node[above, teal] at (B) {$AB$};

    \fill (i) circle (2pt);
    \fill (i_plus) circle (2pt);
    \fill (j) circle (2pt);
    \fill (j_plus) circle (2pt);
    \fill (R) circle (2pt);

    \node[right] at (i) {$i$};
    \node[above left] at (j) {$j$};
    \node[above right] at (j_plus) {$j+1$};


    \begin{scope}[xshift = 9cm]
    \node at (0,-3.5) {$\langle ABij \rangle > 0 $ ,};
    \node at (0,-4.5) {$\langle ABij+1 \rangle < 0 $ ,};
    \node at (4,-6) {$(b)$};
    
    \draw[thick] (0,0) circle(2cm);

    \coordinate (i) at (-100:2cm);
    \coordinate (j) at (130:2cm);
    \coordinate (k) at (100:2cm);
    
    \node at (i) [below ] {$i$};
    \node at (j) [above left] {$j$};
    \node at (k) [above ] {$j+1$};


    \draw[line width=0.4mm, red, dashed] (i) -- (j);
    \draw[line width=0.4mm, red, thick] (i) -- (k);

     \foreach \p in {i,j,k} {
        \fill (\p) circle(2.5pt);
    }
    \end{scope}


    \begin{scope}[xshift = 16cm]
    \node at (0,-3.5) {$\langle ABij \rangle < 0 $ ,};
    \node at (0,-4.5) {$\langle ABij+1 \rangle > 0 $ ,};
    \node at (-3.5,0) {$\bigcup$};
    
    \draw[thick] (0,0) circle(2cm);

    \coordinate (i) at (-100:2cm);
    \coordinate (j) at (130:2cm);
    \coordinate (k) at (100:2cm);
    
    \node at (i) [below ] {$i$};
    \node at (j) [above left] {$j$};
    \node at (k) [above ] {$j+1$};


    \draw[line width=0.4mm, red, thick] (i) -- (j);
    \draw[line width=0.4mm, red, dashed] (i) -- (k);

     \foreach \p in {i,j,k} {
        \fill (\p) circle(2.5pt);
    }
    \end{scope}

\end{tikzpicture}
\caption{A {\LABELREVERSE} compatible LS configuration at $L=1$. Picture $(a)$ is the line configuration in $\mathbb{P}^3$, while $(b)$ represents the geometric space associated to the LS value, as a disjoint union of one-loop Amplituhedron geometries in $AB$ with extra sign constraints. The LS value is then the sum of the respective canonical functions, and it is given as a special case of the first row in Table~\ref{tab:all L2 config}. }
\label{reverse LS_L=1}
\end{figure}
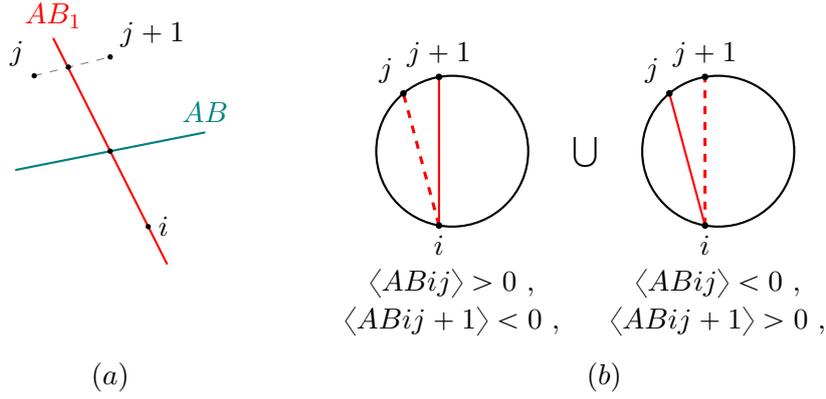

A {\LABELREVERSE} LS configuration involves at least one loop line that intersects $AB$. Let us consider the one-loop ladder geometry. 
In order for $AB_1$ to be fully localized, in addition to satisfying $\langle AB_1 AB \rangle =0$, it must lie on a one-dimensional boundary of $\mathcal{A}^{(1)}_n$. As shown in ref.~\cite{Ranestad:2024svp}, such boundaries are of the form
\begin{equation}
    AB_1 = i \bigl(j + \alpha \, (j+1) \bigl)\,,  \quad \text{for }  \alpha > 0 \,,
\end{equation}
where we can take $i<j$ without loss of generality. In particular, the line $AB_1$ lies in the plane $(ijj+1)$.
Requiring that it intersects $AB$ means that it must pass through the point $AB \cap (ijj+1)$, which fixes
\begin{equation}\label{alpha}
    \alpha =  \frac{\langle AB ij \rangle}{\langle AB j+1i \rangle} > 0 \,.
\end{equation}
We conclude from this that the geometry associated to this LS can be
decomposed into
two one-loop Amplituhedron geometries in $AB$ with extra constraints $\langle AB i j \rangle , \, \langle AB j+1 i \rangle > 0$ for one piece, and $\langle AB i j \rangle , \, \langle AB j+1 i \rangle < 0$ for the other, cf.
Figure \ref{reverse LS_L=1}(b) and (c). The LS value is a special case of that in the first row of Table \ref{tab:all L2 config}, for $k=j+1$.

The first example of a {\LABELREVERSE}, incompatible LS configuration occurs at $L=2$. We can consider the product of two $L=1$ configurations just described, with a specific ordering of the indices as $i<k<j<l<i$, where we take the indices modulo $n$. It turns out that the associated geometry in $AB$ can be triangulated into pieces associated to signed partial triangulations. The forms associated to these is analogous to that of eq. \eqref{LSvaluesimpleincompatibleexample}, and present in general the same undesirable properties: it cannot be expressed in terms of Kermits an it is not conformally invariant in the frame $AB \rightarrow I_\infty$. 

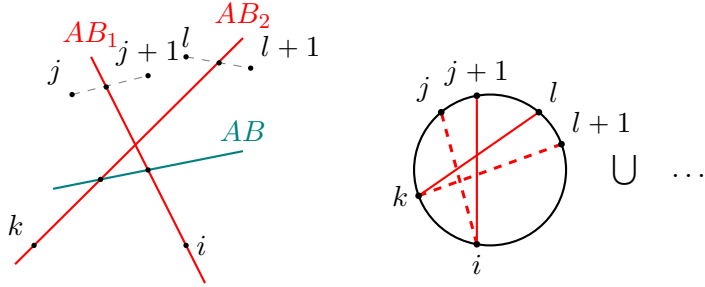
\begin{figure}[t]
\centering
\begin{tikzpicture}[scale = 0.5]

    \coordinate (i) at (1,-2);
    \coordinate (P) at (0,0);
    \coordinate (Q) at (-1.25,-0.25);
    \coordinate (R) at (-1.1,2.2);
    \coordinate (S) at (1.87,2.84);
    \coordinate (j) at (-2,2);
    \coordinate (j_plus) at (0,2.5);
    \coordinate (k) at (-3,-2);
    \coordinate (l) at (1,3);
    \coordinate (l_plus) at (2.7,2.7);
    \coordinate (C) at (-1.5,3);
    \coordinate (D) at (1.5,-3);
    \coordinate (A) at (-2.5,-0.5);
    \coordinate (B) at (2.5,0.5);
    \coordinate (E) at (2.5,3.5);
    \coordinate (F) at (-3.5,-2.5);


    \draw[dashed, gray] (j) -- (j_plus);

    \draw[dashed, gray] (l) -- (l_plus);

    \draw[red, thick] (C) -- (D);
    \node[above, red] at (C) {$AB_1$};

    \draw[red, thick] (E) -- (F);
    \node[above, red] at (E) {$AB_2$};

    \draw[teal, thick] (A) -- (B);
    \node[above, teal] at (B) {$AB$};

    \fill (i) circle (2pt);
    \fill (P) circle (2pt);
    \fill (Q) circle (2pt);
    \fill (R) circle (2pt);
    \fill (S) circle (2pt);
    \fill (j) circle (2pt);
    \fill (j_plus) circle (2pt);
    \fill (k) circle (2pt);
    \fill (l) circle (2pt);
    \fill (l_plus) circle (2pt);

    \node[right] at (i) {$i$};
    \node[above left] at (j) {$j$};
    \node[above] at (j_plus) {$j+1$};
    \node[above left] at (k) {$k$};
    \node[above] at (l) {$l$};
    \node[above right] at (l_plus) {$l+1$};


    \begin{scope}[xshift = 9cm]
    \node at (4.5,0) {$\bigcup \quad \dots$};
    
    \draw[thick] (0,0) circle(2cm);

    \coordinate (i) at (-100:2cm);
    \coordinate (j) at (130:2cm);
    \coordinate (k) at (100:2cm);
    \coordinate (l) at (50:2cm);
    \coordinate (m) at (20:2cm);
    \coordinate (n) at (200:2cm);
    
    \node at (i) [below ] {$i$};
    \node at (j) [above left] {$j$};
    \node at (k) [above ] {$j+1$};
    \node at (l) [above right] {$l$};
    \node at (m) [above right ] {$l+1$};
    \node at (n) [left] {$k$};

    \draw[line width=0.4mm, red, dashed] (i) -- (j);
    \draw[line width=0.4mm, red, thick] (i) -- (k);
    \draw[line width=0.4mm, red, thick] (n) -- (l);
    \draw[line width=0.4mm, red, dashed] (n) -- (m);

     \foreach \p in {i,j,k,l,m,n} {
        \fill (\p) circle(2.5pt);
    }
    \end{scope}

\end{tikzpicture}
\caption{A {\LABELREVERSE} incompatible LS configuration at $L=2$. }
\label{reverse_incomatible LS}
\end{figure}

\section{Outlook}
\label{sec:outlook}

In this paper, we used the negative geometry expansion for the logarithm of the amplitude, and calculated the leading singularities of the Wilson loop with a Lagrangian insertion $F_n^{(L)}$ to all multiplicities at all loop orders. Leading singularities are very important functions, they are defined as special contour integrals of the loop integrand, and appear as rational prefactors multiplied by pure transcendental functions in $F_n^{(L)}$. In the amplitude case, they are subject to further special properties -- they are Yangian invariant and are associated with $2n{-}4$ dimensional cells in the positive Grassmannian $G_{>0}(k,n)$. They also manifest fascinating cluster algebra properties which correlate with the analogous structures in the symbol \cite{Lukowski:2019sxw,Gurdogan:2020tip}. The leading singularities of the Wilson loop $F_n^{(L)}$ are also special: they are written using very simple building blocks -- kermits, which enjoy a hidden conformal symmetry property. This was conjectured in \cite{Chicherin:2022bov} and our paper proves the statement to all $n$ and $L$. This is in contrast with additional leading singularities for individual negative geometries, which are more complicated and in general are not conformal. 

There are further constraints that leading singularities can impose on the structure of the full function $F_n^{(L)}$. Because of the presence of spurious poles, leading singularities are individually not positive or continuously monotonic when evaluated inside the full Amplituhedron region. As already seen in ref. \cite{Henn:2024qwe} in other cases any complete monotonicity properties must follow from an intriguing conspiracy between leading singularities $\Omega_{n,s}$ and functions $f_{n,s}^{(L)}$ in eq. (\ref{F_decomposition}). To understand more concretely the link between leading singularities $\Omega_{n,s}$ and transcendental functions $f_{n,s}^{(L)}$ from this perspective is a fascinating question for further exploration.

Another interesting question is to go beyond the MHV degree, to NMHV and higher. Here the map between amplitudes logarithm, Wilson loop with a Lagrangian insertion and the negative geometries is not completely clear. The starting point could be the definition of the negative geometry where each loop lives in the NMHV one-loop Amplituhedron and the mutual conditions are now $\la YABCD\ra<0$. The canonical form for this geometry corresponds to a supersymmetric function which produces IR finite results when integrated, and could be compared to the supersymmetric Wilson loop calculations. 

We can also search for the positive geometry for the integrand of pure functions that multiply the leading singularities in $F_n^{(L)}$. This would be an important input in the symbol bootstrap for these functions and can dramatically simplify the size of the ansatz and provide additional constraints. It would be also an important step towards the positive geometry picture for functions which are not necessarily amplitudes but other objects of interest. Some preliminary steps along these lines will be made in ref. \cite{progress}.

\section*{Acknowledgments}
EM thanks Dmitry Chicherin, Matteo Parisi and Lauren Williams for useful discussions. 
This work received funding from the European Union (ERC, UNIVERSE PLUS, 101118787). Views and opinions expressed are however those of the authors only and do not necessarily reflect those of the European Union or the European Research Council Executive Agency. Neither the European Union nor the granting authority can be held responsible for them. J.~T. is supported by the U.S. Department of Energy, grant No. SC0009999 and the funds of the University of California. 
This research was supported in part by grant NSF PHY-2309135 to the Kavli Institute for Theoretical Physics (KITP).

\appendix

\addtocontents{toc}{\protect\setcounter{tocdepth}{1}}

\section{All two-loop leading singularity configurations}
\label{app: all two-loop LS config}

\begin{table}[t]
    \centering
    \begin{tabular}{c|c|c}
    Line configuration & Signed partial triangulations & Value \\
    \hline\hline
        \includegraphics[scale = 0.65]{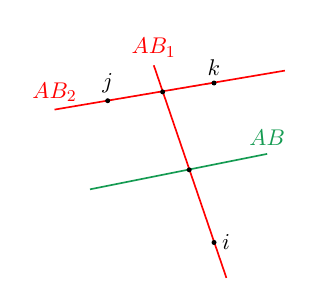} & \includegraphics[scale = 0.65]{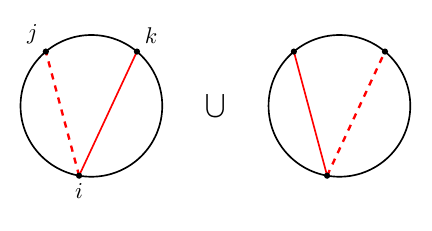} & \includegraphics[scale=0.65]{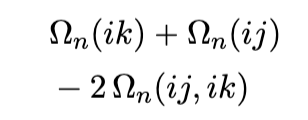}
        \\
        \hline
        \includegraphics[scale = 0.65]{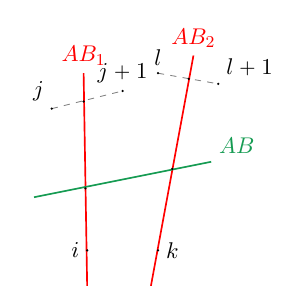} & \includegraphics[scale = 0.65]{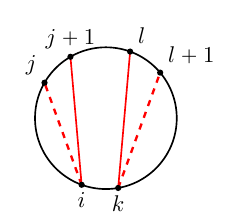} & \includegraphics[scale=0.65]{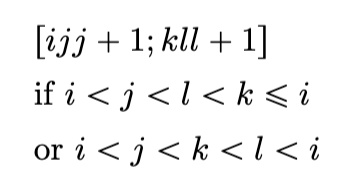} \\
        \hline
         \includegraphics[scale = 0.65]{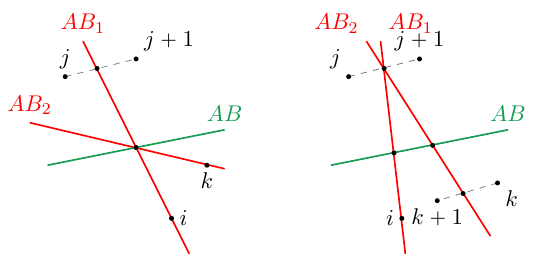} & \includegraphics[scale = 0.65]{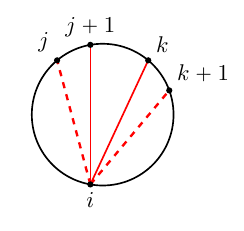} & \includegraphics[scale=0.65]{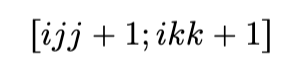}  \\
    \end{tabular}
    \caption{Leading singularity configurations of $F^{(2)}_n$. The line configurations are associated to signed partial triangulations, which in turn can be evaluated to simple compatible LS \eqref{LS_simple_formulae} and Kermit forms \eqref{six_invariant}. }
    \label{tab:all L2 config}
\end{table}

In this section we list all LS configurations of two-loop negative geometries, which appear in eq. \eqref{eq:2-loop}. 
For this result, we performed a case by case analysis starting from the knowledge of the boundaries of the one-loop Amplituhedron \cite{Ranestad:2024svp}, and imposed all possible multi-loop cut conditions, similarly to what we did in Subsubsection \ref{subsubsec:example of reverse LS} for the case in Fig. \ref{reverse LS_L=1}.

The result is the following. In addition to the configurations in Table~\ref{tab:all L2 config} we have all two-loop simple configurations, where $AB_1 = ij$ and $AB=kl$, and the configuration as in Figure \ref{reverse LS_L=1}, with an additional loop line $AB_2 = kl$.
Among all these, by Proposition \ref{proposition exclusion result} only compatible ones are relevant for the Wilson loop $F^{(2)}_n$. This in particular excludes the simple configurations with $(ij)$ and $(kl)$ crossing, and similarly for the other one, if $(kl)$ crosses $(ij)$ or $(ij)$. The only incompatible configurations in Table~\ref{tab:all L2 config} are those in the second row for indices such that the triangles $\{i,j,j+1\}$ and $\{k,l,l+1\}$ overlap in the $n$-gon, i.e. if $i<j<l<k\leq i$ or $i<j<k<l<i$. We evaluate the compatible LS in Table~\ref{tab:all L2 config}. The LS values are expressed either in terms of simple compatible LS \eqref{LS_simple_formulae}, or in terms of Kermit forms \eqref{six_invariant}.

We expect this classification result to be relevant for bootstrap approaches, either for individual two-loop negative geometries or for the full Wilson loop $F^{(2)}_n$.  Specifically, the first step in this method would be to determine of the integral's alphabet. 
The idea is that the Landau analysis is guided by the knowledge of the LS configurations, thanks to the geometric information provided by the Amplituhedron. We will report on progress along these lines in ref. \cite{progress}.

\section{Proof of Proposition \ref{lemma 1} (Classification Result)}
\label{app: proof lemma 1}

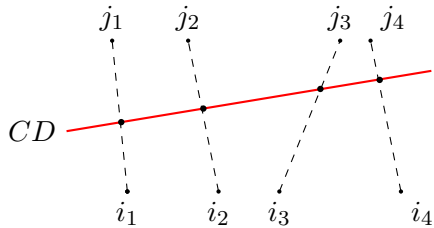
\begin{figure}[t]
    \centering
\begin{tikzpicture}[scale = 0.4]
        \begin{scope}    
    \coordinate (i1) at (-4,-3);
    \coordinate (i2) at (-1,-3);
    \coordinate (i3) at (1,-3);
    \coordinate (i4) at (5,-3);
    \coordinate (j1) at (-4.5,2);
    \coordinate (j2) at (-2,2);
    \coordinate (j3) at (3,2);
    \coordinate (j4) at (4,2);
    \coordinate (C) at (-6,-1);
    \coordinate (D) at (6,1);
    \coordinate (P) at (-4.2,-0.7);
    \coordinate (Q) at (-1.5,-0.25);
    \coordinate (R) at (2.35,0.4);
    \coordinate (S) at (4.3,0.7);

    \draw[red, thick] (C) -- (D);
    \node[left] at (C) {$CD$};

    \draw[dashed] (i1) -- (j1);
    \draw[dashed] (i2) -- (j2);
    \draw[dashed] (i3) -- (j3);
    \draw[dashed] (i4) -- (j4);
   
    \fill (P) circle (3pt);
    \fill (Q) circle (3pt);
    \fill (R) circle (3pt);
    \fill (S) circle (3pt);
    \fill (i1) circle (2pt);
    \fill (i1) circle (2pt);
    \fill (i1) circle (2pt);
    \fill (i2) circle (2pt);
    \fill (i3) circle (2pt);
    \fill (i4) circle (2pt);
    \fill (j1) circle (2pt);
    \fill (j2) circle (2pt);
    \fill (j3) circle (2pt);
    \fill (j4) circle (2pt);

    \node[below] at (i1) {$i_1$};
    \node[below] at (i2) {$i_2$};
    \node[below] at (i3) {$i_3$};
    \node[below right] at (i4) {$i_4$};
    \node[above] at (j1) {$j_1$};
    \node[above] at (j2) {$j_2$};
    \node[above] at (j3) {$j_3$};
    \node[above right] at (j4) {$j_4$};

    \end{scope}

\end{tikzpicture}
    \caption{The line $CD$ intersects other four lines, each of which is associated to an arc $(i_rj_r)$ of the $n$-gon.}
    \label{fig:CD int four lines}
\end{figure}

\subsection{Auxiliary analysis at $L=1$\label{subsec:aux}}

In a preliminary step, we consider an intersection of lines as shown in Fig.~\ref{fig:CD int four lines}.
In other words, we are interested in solutions to the following equations,
\begin{align}
\langle CD \, i_1 j_1 \rangle = 
\langle CD \, i_2 j_2 \rangle = 
\langle CD \, i_3 j_3 \rangle = 
\langle CD \, i_4 j_4 \rangle = 0 \,,
\end{align}
with $CD$ representing a loop line that is inside the one-loop Amplituhedron $\mathcal{A}^{(1)}_n$ and $i_r,j_r \in \{1,\dots,n\}$.
Without loss of generality, we take $i_1<i_2<i_3<i_4$, and $i_r < j_r$. We allow for the case $j_r = i_r +1$ for some $r$.

Consider the first two conditions that state that $CD$ intersects $(i_1j_1)$ and $(i_2j_2)$. We parametrize this solution as $C = i_1 + \alpha_1 \, j_1 $ and $D = i_2 + \alpha_2 \, j_2$, with real $\alpha_1,\alpha_2$. We are looking for a full-dimension (i.e., two-dimensional) solution. In this case, loop positivity \eqref{loop_positivty} forces an ordering on the indices, such that 
$i_1<j_1 \leq i_2<j_2$, which corresponds to $(i_1j_1),(i_2j_2)$ not crossing, and moreover it implies $\alpha_1, \alpha_2 > 0$. 

Intersecting $CD$ further with $(i_3j_3)$, i.e. setting $\langle CD \, i_3j_3 \rangle=0$, yields
    \begin{equation}\label{triple_inters}
        \langle i_1i_2 i_3j_3 \rangle + \alpha_1 \, \langle j_1i_2 i_3j_3 \rangle + \alpha_2 \, \langle i_1j_2 i_3j_3 \rangle + \alpha_1 \, \alpha_2 \, \langle j_1j_2 i_3j_3 \rangle  = 0   \,.
\end{equation}
For generic (distinct) indices $i_r,j_r$, 
we find that a solution to \eqref{triple_inters} exists only if $(i_3j_3)$ crosses $(i_1j_1)$ or $(i_2j_2)$, but not both. 
In addition, there are solutions to \eqref{triple_inters} in the case where some indices coincide.
Finally, we consider in addition the condition $\langle CD \, i_4 j_4 \rangle = 0$. 
We find that the allowed solutions fall into three classes:
\begin{enumerate}
    \item For general indices, i.e. for four disjoint arcs, it turns out that the Amplituhedron conditions on $CD$ imply that exactly two pairs of arcs cross, for a solution to exist, as in Fig. \ref{fig:non-simple LS}(b) with $AB_5=CD$.
     \item When two indices coincide, i.e. two arcs meet at a vertex $i$, then $CD$ passes through $i$ and and the other two arcs must cross in order for a solution to exist, as in Fig. \ref{fig:non-simple LS}(a) with $AB_3=CD$.
     \item If two pairs of indices coincide, i.e. if two pairs of arcs meet at vertices $i$ and $j$ respectively, then $CD = ij$, which is the usual Amplituhedron case.
\end{enumerate}

\subsection{Proof at all $L$}

We are now ready to prove Proposition \ref{lemma 1} (Classification Result).
By definition, {\LABELNONSIMPLE} LS do not involve any intersection condition with $AB$, i.e. ${\langle AB \, AB_\ell \rangle \neq 0}$ for every $\ell=1, \dots L$. Therefore, such configurations correspond to zero-dimensional boundaries of a (negative) loop-Amplituhedron geometry, in which we can ignore $AB$. 

Among these zero-dimensional boundaries there are configurations corresponding to simple LS, where all loop lines localize to arcs of the $n$-gon. 
Consider now a {\LABELNONSIMPLE} LS, i.e. for which there exists some loop line, which we denote by $CD$, that is not localized to an arc of the $n$-gon. We show presently that such a configuration must be incompatible. 
Since all zero-dimensional boundaries of the one-loop Amplituhedron $\mathcal{A}^{(1)}_n$ are of the form of arcs on an $n$-gon, it follows that $CD$ must intersect at least two other loop lines. Assume in fact that $CD$ intersects only one other loop line. By induction on the number of loops we can assume the that this is equal to an arc $kl$. Then, the only solution for $CD$ is $CD = i \, (kl) \cap (ijj+1) $, which violates the one-loop Amplituhedron conditions unless the indices degenerate such that $CD $ is equal to an arc of the $n$-gon. This is the case for e.g. $k=j$, in which case $CD=ij$.

We can therefore assume that $CD$ intersects other two loop lines. We can assume by induction on the number of loops that these are localized to arcs $(i_1j_1),(i_2j_2)$ of the $n$-gon. 
The auxiliary analysis of subsection \ref{subsec:aux} then shows that $CD$ localizes to an arc of the $n$-gon, unless $(i_1j_1)$ crosses $(i_2j_2)$, i.e. when $\langle i_1 j_1 i_2 j_2 \rangle < 0$. 
This proves that every {\LABELNONSIMPLE} LS is incompatible.

\section{Proof of Proposition \ref{proposition reverse compatible are triang} (Signed Partial Triangulations)}
\label{app: proof prop 2}

Every LS configuration corresponds to a boundary of a negative geometry. Such boundary corresponds to a collection of $L$ lines $AB_\ell$ in $\mathbb{P}^3$, completely localized by imposing intersection conditions among themselves, with $AB$, or with the fixed lines $(ii+1) := (Z_i Z_{i+1})$. These intersection conditions are expressed in equations as
\begin{equation}\label{cut_conditions}
    \la AB_\ell \, AB_{\ell'} \ra = 0 \,, \quad \la AB_\ell \, AB \ra = 0  \,, \quad \la AB_\ell \, ii+1 \ra = 0 \,,
\end{equation}
for $\ell \in \{1, \dots ,L\}$ and $i \in \{1, \dots,n\}$.
Every line $AB_\ell$ in $\mathbb{P}^3$ has four degrees of freedom.

\subsection{Classification of intersection points}
\label{subapp: ordering intersection points}

Note that any line is uniquely determined by two distinct points on it. Points obtained as intersections of $AB_\ell$, $(ii{+}1)$ and $AB$ must be determined by projective invariant equations depending on only $(ii{+}1)$ and $AB$. We therefore define the following.
\begin{definition}\label{def_int_pt}
    By \textit{intersection point} in an $L$-loop LS configuration, we mean all intersection points between the lines $AB_\ell$, $(ii+1)$ and $AB$.
\end{definition}
Generally, there exist intersection points which cannot be represented as iterated intersections of lines and planes. We can understand this as follows. A solution to a generic Schubert problem in $\rm{Gr}(2,4)$, as e.g. all lines intersecting four generic (pairwise non-intersecting) fixed lines in $\mathbb{P}^3$, requires solving a quadratic equation. Solutions to such problems are therefore not rational. On the other hand, all points represented by iterated intersections have rational parametrization. An example of a non-rational intersection point in a LS configuration is shown in Fig. \ref{fig:non-simple LS}(b). 

We now introduce a notion for measuring the complexity of rational intersection points. Note for each of these, there can be several equivalent expressions representing it as an intersection of different lines and planes. Nevertheless, we can define the following canonical quantity.
\begin{definition}\label{def_length}
    The \textit{length} of a rational intersection point in LS configuration is the minimal number of intersection symbols needed to represent the point as intersections of lines and planes in $\mathbb{P}^3$, constructed iteratively out of $i:=Z_i$ and $AB$. 
\end{definition}
For instance, the length of $i \in \{1,\dots,n\} $ is zero. The intersection points of length one are given by
\begin{equation}\label{simplest_invariants}
       P[jk,i] := (jk) \cap (iAB) \,, \quad    AB \cap (ijk) = AB \cap (iP[jk,i])  \,,
\end{equation}
as well as 
\begin{equation}\label{nonMHV_points}
    (ijk) \cap (ab)  \,,
\end{equation}
where all indices are in $\{1,\dots,n\}$. An example involving points as in eq. \eqref{simplest_invariants} is given in the first configuration in Table \ref{tab:all L2 config} in Appendix \ref{app: all two-loop LS config}. 

At length two, one can considers intersection points built out of length one points. For instance, we have
\begin{equation}\label{length_2_invariants}
    AB \cap (P[jk;i]cd)  \,, \quad (P[jk,i]c) \cap (dAB) \,,
\end{equation}
as well as
\begin{equation}\label{nonMHV_length2_inv}
     (P[jk,i]ab) \cap (cd)  \,.
\end{equation}

One can further iterate this process and generate an infinite number of intersection points of increasing length, starting from lines and planes involving points of lower length. 
In the following, we are only interested in the recursive structure of the allowed points in LS configurations of $F^{(L)}_n$, and not in their complete classification. Note that $A$ and $B$ cannot appear separately, i.e. $AB$ appears either as a line or as part of a plane. 
\begin{tcolorbox}[colback=white]
   The intersection points relevant for compatible LS configurations are the following:
   \begin{equation}\label{iteration}
   (PQ) \cap ( ABR) = P \, \langle AB QR \rangle - Q \, \langle AB PR \rangle \,,
   \end{equation}
   and 
   \begin{equation}\label{iteration_on_AB}
   \begin{aligned}
       AB \cap (PQR) &= P \, \la AB QR \ra + Q \, \la AB RP \ra + R \, \la AB PQ \ra  \\
       &= (PQ) \cap ( ABR) + \la AB PQ \ra \, R \,.
   \end{aligned}
   \end{equation}
The intersection points $P,Q,R$ are constructed iteratively via eq. \eqref{iteration} and \eqref{iteration_on_AB} starting from points among $ \{1,\dots,n\}$. 
\end{tcolorbox}
The second equality in eq. 
\eqref{iteration_on_AB} follows from the identity
\begin{equation}\label{identities_lines_planes}
\begin{aligned}
     (PQR) \cap (STU) &= \bigl((PQ)\cap (STU),(QR)\cap (STU) \bigr) \,.
\end{aligned}
\end{equation} 
Other than the intersection points in eqs. \eqref{iteration} and \eqref{iteration_on_AB}, there are those generalizing the points in eq. \eqref{nonMHV_points}, 
\begin{equation}\label{forbidden_points}
    (PQ) \cap (RST) \,.
\end{equation}
The intersection points $P,Q,R,S,T$ are constructed inductively out of eqs. \eqref{iteration}, \eqref{iteration_on_AB} and \eqref{forbidden_points}. In Subsection \ref{subapp: disappearance of non-MHV int pts} we show that the points in eq. \eqref{forbidden_points} do not appear in LS configurations for $F^{(L)}_n$.

\subsection{The allowed intersection points}
\label{subapp: disappearance of non-MHV int pts}

We control the {\LABELREVERSE} compatible LS configurations since they do not contain intersection points which are not rational or those as in eq. \eqref{forbidden_points}. 
\begin{tcolorbox}[colback=white]
    All intersection points in a {\LABELREVERSE} compatible LS configuration are rational and of the form as in eqs. \eqref{iteration} and \eqref{iteration_on_AB}.   
\end{tcolorbox}

\begin{proof} 
Let $\mathcal{L}$ be a {\LABELREVERSE} compatible LS.
We first argue that no loop line in $\mathcal{L}$ can intersect three other loop lines different from $AB$ in three distinct points. This implies in particular that non-rational points are not allowed. In fact, a quadratic equation appears in a Schubert problem only if the intersection conditions involve four generic loop lines, and a solution line to such a problem would therefore contain four intersection points. 

For that, let $PQ$, $RS$, $TU$ be lines built out of intersection points. Let $CD$ be a loop line intersecting all three of them. Let $l \geq 0$ denote the sum of the lengths of $P,Q,R,S,T,U$. If $l = 0$, then all points belong to $\{1, \dots,n\}$ and by the proof in Appendix~\ref{app: proof lemma 1}, the only compatible configuration allowed by the Amplituhedron conditions is when e.g. $C=P=R$ and $D \in TU$. In particular, the intersection points on $CD$ are only two. Let $l \geq 1$ and assume that the length of $P$ is greater or equal than one, so that we can write $P=P_1 + \alpha \, P_2$ for $P_i$ of length less than that of $P$. This follows from the twistor identities as in eqs. \eqref{iteration} and \eqref{iteration_on_AB}. Then, replacing $P$ by $P_1$ and by $P_2$, and using induction on the length, we conclude e.g. that either $P_1 = R$ and $P_2 = S$, or $P_2=T$. In the former case we are done, while in the latter one checks that the intersection points are also two and not more.

Consider now an intersection point as in eq. \eqref{forbidden_points}. We argue that such a point is not allowed in a {\LABELREVERSE} compatible LS configuration, by induction over the sum $l \geq 0$ of the lengths of $P,Q,R,S,T$. For $l=0$, we have $P,Q,R,S,T \in \{1,\dots,n\}$. By the proof in Appendix \ref{app: proof lemma 1}, such points are forbidden for distinct indices in compatible LS configurations. Assume that $l \geq 1$ and that the length of $P$ is greater than zero. Write $P=P_1 + \alpha \, P_2$ with $P_i$ intersection points of length less than that of $P$. Then,
\begin{equation}
    (PQ) \cap (RST) =  (P_1Q) \cap (RST) + \alpha \, (P_2Q) \cap (RST) \,.
\end{equation} 
By induction hypothesis, we must have that $(P_iQ) \cap (RST) \in \{R,S,T\}$. Assume e.g. that $P_1=R$ and $P_2=S$, then $(PQ) \cap (RST)$ is equal to the intersection point of the lines $PQ$ with $RS$, which we now argue to be either empty or not allowed for distinct $P,Q,R,S$.

We show that for distinct $P,Q,R,S$, an allowed intersection point between $PQ$ and $RS$ happens if $S=PQ \cap (RAB)$, in which case the intersection point is given by $AB \cap (PQR)$. Otherwise, the intersection of $PQ$ with $RS$ is non-empty only on a codimension one locus in $AB$ intersecting the interior of $\mathcal{A}^{(1)}_n$. In the latter case, the bracket $\la PQ \,  RS \ra$ changes sign for varying $AB$ inside $\mathcal{A}^{(1)}_n$, and in particular the configuration is incompatible. Assume that the length of $S$ is greater than zero, and write $S=S_1+ \alpha \, S_2$ with $S_i$ invariant points. Then, we can write $S=(S_1 S_2) \cap (PQR)$ and $T= (PQ) \cap (RS_1S_2)$. 
Since these intersection points have total length strictly smaller than those we started with, we can use induction to argue that either $\{S_1,S_2\} \cap \{P,Q\} \neq \emptyset$, which contradicts the fact that $P,Q,R,S$ are distinct, or that $S_1S_2=AB$. It follows that the the intersection of $PQ$ with $RS$ is non-empty only in the following cases. First, if $Q=U + \alpha \, V$ and e.g. $S=U$, so that $\langle PQ \, RS \rangle = 0$ for $\alpha = 0$. Second, if $Q=U + \alpha \, V$ and $S= U + \beta \, V$, so that $\langle PQ \, RS \rangle = 0$ on the codimension one subvariety in $AB$ given by $\alpha = \beta$, which intersects the interior of $\mathcal{A}^{(1)}_n$. In both cases the intersection of $PQ$ with $RS$ is empty for generic values of $AB$. This concludes the proof.
\end{proof}

\subsection{Composite residues are irrelevant}
\label{subapp: Composite residues are irrelevant}

In order to fully localize all lines in an $L$-loop LS configuration, we generally need $4L$ equations among eq. \eqref{cut_conditions}. However, leading singularities are residues of a rational function, whose denominator involves factors among those in eq. \eqref{cut_conditions}. Here is a subtle point. Some of the brackets in \eqref{cut_conditions} may factorize into non-trivial irreducible polynomials, when evaluated on the support of some cuts. 

A prototypical example of such factorization happens on the support of $\la AB_\ell \, i{-}1 i \ra = 0$, on which the bracket $\la AB_\ell \, ii{+}1 \ra$ factorizes into two terms. The vanishing locus of each term has a clear geometric meaning: the line $AB_\ell$ can intersect both $(i{-}1i)$ and $(ii{+}1)$, by either passing through the point $i$, or by lying in the plane $(i{-}1ii{+}1)$. This example captures this factorization feature, which is very general for lines in $\mathbb{P}^3$: every line intersecting two lines meeting in a point, either passes through this point, or it lies in the plane spanned by the two lines.
From the point of view of residues, whenever a polynomial in the denominator factorizes, one can take two residues, one for each factor. Algebraically, this means imposing the vanishing of both factors. We call such residues \textit{composite}.

Having to consider all possible composite residues, i.e. all possible factorizations of the brackets in eq. \eqref{cut_conditions}, would result in a very complicated task. However, the following result guarantees that we can ignore composite residues.
\begin{tcolorbox}[colback=white]
Every LS of a negative geometry that involves composite residues, either vanishes or it has the same value as a LS not involving any composite residues. Geometrically, every four-dimensional boundary of an $L$-loop negative geometry for which $AB$ lies in the interior of $\mathcal{A}^{(1)}_n$, is an irreducible component of a boundary obtained by imposing exactly $4L$ equations among those in eq. \eqref{cut_conditions}.
\end{tcolorbox}
The important corollary of this statement, is that we can always assume that every condition in eq. \eqref{cut_conditions} fixes only one degree of freedom, and in particular, we can ignore factorizations of twistor brackets in eq. \eqref{cut_conditions} when considering LS configurations.

\begin{proof}
    As a warm-up case, let us consider the case of $L=1$. The only possible composite residue in this case arises from the factorization of $\la AB_1 \, ii+1 \ra $ on the support of $\la AB_1 \, i{-}1i \ra = 0$. As we discussed previously, imposing the simultaneous vanishing of both factors forces $AB_1$ to pass through $i$ and lie in the plane $(i{-}1ii{+}1)$. In this configuration, the line $AB_1$ has one degree of freedom. Therefore, in order for $AB_1$ to be fully localized, we need to impose another one-loop Amplituhedron condition $\la AB_1 \, jj{+}1 \ra = 0$. The solution for $AB_1$ is then
    $AB_1 = i \, (i{-}1ii{+}1) \cap (jj{+}1)$, which violates the one-loop Amplituhedron conditions unless $j=i{+}1$ or $j{+}1 = i{-}1$, in which case $AB_1 = ii{+}1$ or $AB_1 = i{-}1i$, respectively. Note that the latter cases can be obtained by imposing only conditions of the form $\la AB_1 \, kk+1 \ra = 0$, instead of considering a composite singularity. For instance, the solution $AB_1 = ii+1$ is obtained by from the four cut conditions
    \begin{equation}\label{L1_four_cuts}
        \la AB_1 \, i{-}2i{-}1 \ra  = \la AB_1 \, i{-}1i \ra  = \la AB_1 \, ii{+}1 \ra  = \la AB_1 \, i{+}1i{+}2 \ra  = 0 \,.
    \end{equation}
    We have therefore proven that for $L=1$, every LS involving composite residues localizes the loop line $AB_1$ either outside the one-loop Amplituhedron or to a configuration associated to \eqref{L1_four_cuts}, which is free of composite residues.

    We now present the argument for $L>1$, which is completely analogous. Consider a LS configuration that involves a composite residue, i.e. a factorization of a multiloop cut $\la AB_1 \, AB_2 \ra  $. As we explained previously, such factorization arises geometrically only if there exists a line $RT$, where $R$ and $T$ are intersection points, intersecting $AB_1$ at $R$. We can then write $AB_1 = RS$, for some intersection point $S$. Then, taking the composite residue along $\la AB_1 \, AB_2 \ra  $ corresponds to localizing $AB_2$ to pass through $R$ and lie in the plane $(RST)$. Since after this condition $AB_2$ has still one degree of freedom, $AB_2$ must solve at least another intersection condition $\la AB_2 \, PQ \ra $, where $P$ and $Q$ are intersection points. The solution for $AB_2$ now reads $AB_2 = R \,  (PQ) \cap (RST)$, see Fig. \ref{fig:composite}. This solution involves the intersection point $(PQ) \cap (RST)$, which by Subsection \ref{subapp: disappearance of non-MHV int pts} it is forbidden, i.e. the associated LS value vanishes, unless $P=T$, in which case $AB_2 = RP$. In the latter case, we can localize $AB_2 = RP$, instead of using composite residue by imposing the following cut conditions
    \begin{equation}\label{general_four_cuts}
          \la AB_2 \, RP \ra  =  \la AB_2 \, PQ \ra  = \la AB_2 \, RS \ra =  \la AB_2 \, RQ \ra = 0   \,. 
    \end{equation}
    The last condition in eq. \eqref{general_four_cuts} corresponds to the intersection condition between $AB_2$ and the line $RQ$. It could be that $RQ$ is not part of the initial line configuration, in which case we can consider the $L+1$ configuration where the first $L$ loop lines are localized as in the previous configuration and $AB_{L+1} = RQ$. This concludes the proof.

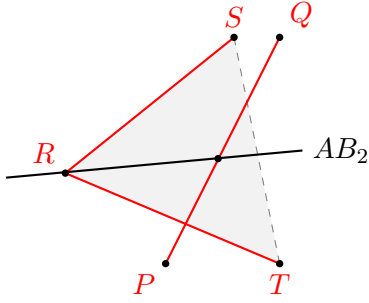
\begin{figure}[t]
\centering
\begin{tikzpicture}[scale = 0.6]

    \coordinate (P) at (0.5,-2);
    \coordinate (Q) at (3,3);
    \coordinate (R) at (-1.7,0);
    \coordinate (S) at (2,3);
    \coordinate (T) at (3,-2);
    \coordinate (A) at (-3,-0.1);
    \coordinate (B) at (3.5,0.5);
    \coordinate (C) at (1.65,0.32);

    \fill[gray!10] (R) -- (S) -- (T) -- cycle;
    
    \draw[dashed, gray] (S) -- (T);

    \draw[red, thick] (R) -- (S);
    \draw[red, thick] (R) -- (T);
    \draw[red, thick] (P) -- (Q);
    \draw[thick] (A) -- (B);
    \node[right] at (B) {$AB_2 $};

    \fill (P) circle (2.3pt);
    \fill (Q) circle (2.3pt);
    \fill (R) circle (2.3pt);
    \fill (S) circle (2.3pt);
    \fill (T) circle (2.3pt);
    \fill (C) circle (2.3pt);

    \node[above left, red] at (R) {$R$};
    \node[above, red] at (S) {$S$};
    \node[below, red] at (T) {$T$};
    \node[below left, red] at (P) {$P$};
    \node[above right, red] at (Q) {$Q$};

\end{tikzpicture}
\caption{The line $AB_2$ is fixed by a composite residue to pass through $R$ and lie in the plane $(RST)$, as well as by a further intersection condition with $PQ$. The resulting line is $AB_2 = R \, (PQ) \cap (RST)$.}
\label{fig:composite}
\end{figure}

\end{proof}

\subsection{The conditions on the LS geometry}
\label{subapp: conditions on AB geometry}

Imposing positivity or negativity of a twistor bracket $\langle CD  EF \rangle $ between two non-intersecting loop lines $CD$ and $EF$, yields non-trivial constraint on the geometry in $AB$, only if the bracket changes sign for $AB$ ranging in $\in \mathcal{A}^{(1)}_n$. However, if this is the case, the LS configuration is incompatible. Therefore, multiloop conditions are irrelevant when considering compatible configurations. We now determine all possible non-trivial conditions on the geometry in $AB$.

\begin{tcolorbox}[colback=white]
The LS value of any {\LABELREVERSE} compatible LS configuration is associated to a one-loop Amplituhedron geometry in $AB$ with extra conditions originating as follows.
\begin{itemize}
    \item {\it{\underline{Condition 1:}}} The mutual negativity condition $\langle AB \, CD \rangle<0$, for every loop line $CD$ not intersecting $AB$.
    \item {\it{\underline{Condition 2:}}} The one-loop Amplituhedron conditions for every other loop line~$CD$.
\end{itemize}
We show that these conditions take the form of imposing a fixed sign on products of twistor coordinates of the form $\langle AB \, ij \rangle$. 
\end{tcolorbox}

The proof is by induction on the sum of the length of $C$ and $D$, where $CD$ is represented by two intersection points. By Subsection \ref{subapp: disappearance of non-MHV int pts}, $C$ and $D$ are as in eqs. \eqref{iteration} and \eqref{iteration_on_AB}. Their structure, combined with the Plücker relation, directly yields the claim for condition~1. For condition~2, we use the definition of the Amplituhedron $\mathcal{A}_n^{(1)}$ as a projection of the positive Grassmannian ${\rm Gr}_{\geq 0}(2,n)$, and show that $CD$ being the image of a point in ${\rm Gr}_{\geq 0}(2,n)$ imposes the claimed constraints on~$AB$.

\begin{proof}

{\it{\underline{Condition 1:}}} 
Let $CD$ be a loop line such that $\la AB \, CD \ra \neq 0$. We argue by induction over the sum $l \geq 0$ of the lengths of $C$ and $D$, where we take $C$ and $D$ to be intersection points. If $\ell=0$, then $C,D \in \{1,\dots,n\}$ and the claim is obvious. Assume that $\ell \geq 1$ and that the length of $C$ is greater or equal than one. Since $CD$ does not intersect $AB$, we can write $C$ as in \eqref{iteration} with $P,Q,R$ intersection points that also do not lie on $AB$, and such that the sum of their length is less than $\ell $. We then compute
\begin{equation}
\begin{aligned}
    \langle AB \, CD \rangle =&\, 
    \langle  AB \, PD \rangle \langle  AB \, QR \rangle - \langle  AB \, QD \rangle \langle  AB \, PR \rangle \\
    =& \, \langle  AB \, PQ \rangle \langle  AB \, DR \rangle \,,
    \end{aligned}
\end{equation}
where the last equality is the Plücker relation.
By induction hypothesis, $\langle  AB \, PQ \rangle$ and $\langle  AB \, DR \rangle$ can be expressed as products of brackets of the form $\la AB \, ij \ra$, and consequently the same is true for $\langle AB \, CD \rangle$. \vspace{0.1in}

{\it{\underline{Condition 2:}}} 
We now consider the one-loop Amplituhedron conditions of a loop line $CD$. For that, we are going to use the alternative definition of the one-loop Amplituhedron. Namely, $\mathcal{A}^{(1)}_n$ is the image of the non-negative Grassmannian ${\rm Gr}_{\geq 0}(2,n)$ under the linear projection ${\rm Gr}(2,n) \rightarrow {\rm Gr}(2,4)$ induced by right multiplication by $Z$. The equivalence of this notion with the definition given in Subsection \ref{subsec:Amplituhedron and Kermits} has been proven in \cite{Parisi:2021oql}.
In particular, we represent $CD$ as $CD = M \cdot Z$, with a $2 \times n$ real matrix $M $. Then, $CD \in \mathcal{A}^{(1)}_n$ if and only if $M \in {\rm Gr}_{\geq 0}(2,n)$, i.e. if all Plücker coordinates of $M$ have the same sign. 

Let $l \geq 0$ denote the sum of the lengths of $C$ and $D$. If $l=0$, then $C < D$ belong to $\{1,\dots,n\}$ and $CD$ is a zero-dimensional boundary of the one-loop Amplituhedron geometry. In this case, $M$ has only non-zero columns 
\(\left(\begin{smallmatrix}
  1 \\
  0 
\end{smallmatrix}\right)\)
at $C$ and 
\(\left(\begin{smallmatrix}
  0 \\
  1 
\end{smallmatrix}\right)\)
at $D$, so it lies in ${\rm Gr}_{\geq 0}(2,n)$. 

Assume now that $l \geq 1$ and that the length of $C$ is greater or equal than one. If $C$ does not lie on $AB$, by eq. \eqref{iteration} we can write it as
\begin{equation}\label{C1}
    C=P + \alpha \, Q \, , \quad \text{with} \quad \alpha = \frac{\langle AB \, PR \rangle}{\langle AB \, RQ \rangle} \,.
\end{equation}
If instead $C$ lies on $AB$, by eq. \eqref{iteration_on_AB} we can write it as
\begin{equation}
    C=R + \beta \, S \, , \quad \text{with} \quad \beta = \frac{\langle AB \, RD \rangle}{\langle AB \, DS \rangle} 
\end{equation}
is fixed by $\langle AB \, CD\rangle =0$.
In both cases, it follows by induction on the length that $\alpha$ and $\beta$ can be expressed as products of brackets of the form $\langle AB \, ij \rangle$ and their inverses. We now show that in order for $CD$ to lie in the one-loop Amplituhedron, the sign of $\alpha$ or $\beta$ must be fixed, respectively. The following argument is the same for both cases, so let us consider the case where $C \notin AB$, i.e. eq. \eqref{C1}. We can then write $CD =(P + \alpha \, Q)D$ and take parametrisations $M_1$ of $PD$ and $M_2$ of $QD$, respectively, such that the first rows of $M_1$ and $M_2$ are the same. Then 
\begin{equation}\label{Plucker_sum}
    p_{ab}(M_1 + \alpha \, M_2) = p_{ab}(M_1) + \alpha \, p_{ab}(M_2) \,,
\end{equation}
where $p_{ab}(M)$ denotes the $(ab)$-minor of $M$. 
Moreover, the lines $PD$ and $QD$ must lie inside the one-loop Amplituhedron, 
which is equivalent to $p_{ab}(M_1)$ and $p_{ab}(M_2)$ having the same sign for every $1 \leq a<b \leq n$. By induction hypotheses, this requirement imposes fixed signs on products of brackets $\langle AB ij \rangle$. Moreover, we can fix the sign of $\alpha$ such that the sign of all $p_{ab}(M_1 + \alpha \, M_2)$ is the same for all $1 \leq a<b \leq n$. 
One can check that the converse is also true: if $\alpha$ in eq. \eqref{Plucker_sum} does not have a fixed sign, then all minors in \eqref{Plucker_sum} cannot have the same sign for every $AB \in \mathcal{A}^{(1)}_n$. We conclude that the sign of $\alpha$ must be fixed, which concludes the proof.
\end{proof}

\subsection{Compatibility implies non-crossing of arcs  }
\label{subapp: crossing diagrams are incompatible}

By Subsection \ref{subapp: conditions on AB geometry}, every {\LABELREVERSE} compatible LS is associated to a geometric space in $AB$, represented as a union of collections of signed arcs of the $n$-gon. We now prove the following.
\begin{tcolorbox}[colback=white]
    Every collection of signed arcs of the $n$-gon associated to a {\LABELREVERSE} compatible LS configuration involves only non-crossing arcs. In particular, it is a signed partial triangulation.
\end{tcolorbox}

The proof's idea is the following. Let $\mathcal{L}$ be a {\LABELREVERSE} compatible LS configuration, whose associated collection of signed arcs of the $n$-gon involves two arcs $(i_1j_1)$ and $(i_2j_2)$. We show that there exist two chains of loop lines in $\mathcal{L}$, starting from $i_t$ and ending at $j_t$ for $t=1,2$, respectively. Since $\mathcal{L}$ is compatible, every pair of lines, one in each chain, has non-negative mutual twistor bracket. We then show that this implies that $\la i_1j_1i_2j_2 \ra \geq 0$, which means that $(i_1j_1)$ and $(i_2j_2)$ do not cross.

\begin{proof}
We first prove the following auxiliary result. Let $P,Q,R,S,T$ be intersection points in $\mathcal{L}$. If $\langle PQST \rangle \geq 0$ and $\langle QRST \rangle \geq  0$ for every $AB \in \mathcal{A}^{(1)}_n$, then also $\langle PRST \rangle \geq 0$ for every $AB \in \mathcal{A}^{(1)}_n$. We prove this by induction over $l \geq  0$, the sum of the lengths of all six intersection points. 
If $l=0$, then $P,Q,R,S,T \in \{1,\dots,n\}$, and the claim follows directly from external positivity \eqref{external_positivity}. Let $l \geq 1$, and assume that the length of $P$ is greater or equal than one. We write $P=P_1 + \alpha \, P_2$ with $P_i$ invariant points of length less than that of $P$. Since $\langle PQST \rangle  \geq 0$ for every $AB \in \mathcal{A}^{(1)}_n$, we can assume that $\alpha>0$ for every $AB$ in the LS geometry associated to $\mathcal{L}$. If $\langle P_r QST \rangle <0$ for $r=1$ or $r=2$, for some $AB$, we could continuously reach the locus where $\alpha $ is equal to zero or infinity, simultaneously preserving the sign $\langle P_rQST \rangle <0$. This would contradict the fact that $\langle PQST \rangle \geq 0$ on the closure of the geometry in $AB$. It follows that $\langle P_rQST \rangle \geq 0$ for both $r=1,2$. By induction on the length, $\langle P_rRST \rangle > 0$ for $r=1,2$ and therefore $\langle PRST \rangle > 0$. The same argument works when considering $Q$, $S$ or $T$ instead of $P$ in the induction step.

We are now prove the main result. Assume that $\mathcal{L}$ involves a condition on the sign of $\langle AB i_1j_1 \rangle$ and $\langle AB i_1j_1 \rangle$. Then, there exist two chains of lines $(P_rP_{r+1})$ and $(Q_sQ_{s+1})$ in $\mathcal{L}$, of length $N$ and $M$ respectively, such that $P_0=i_1 $, $Q_0=i_2$, $P_{N}=j_1$ and $Q_M=j_2$. This follows from the fact that the bracket $\langle AB i_1j_1 \rangle$ can appear among the conditions on the geometry in $AB$, if and only if there exists an intersection point depending on both $i_1$ and $j_1$. This in turn can happen only if the two points are connected by a chain of lines in $\mathcal{L}$, which by Subsection \ref{subapp: disappearance of non-MHV int pts} must intersect each other in points $P_r$ of the form as in eqs. \eqref{iteration} and \eqref{iteration_on_AB}.
Since $\mathcal{L}$ is compatible, $\la P_{r-1} P_{r} Q_{s-1}Q_{s} \ra \geq 0$ for every $1 \leq r< \leq N$ and $1 \leq s \leq M$ and for any $AB \in \mathcal{A}_n^{(1)}$. We now use the auxiliary result above. In fact, from $\la P_0P_1Q_0Q_1 \ra \geq 0$ and $\la P_1P_2Q_0Q_1 \ra \geq 0$, it follows that $\la P_0P_2Q_0Q_1 \ra \geq 0$. Running this argument over $P_r$ for $r=1,\dots,N$ yields that $\la P_0P_NQ_0Q_1 \ra \geq 0$. Then, the analogous procedure on $Q_s$ for $s=1,\dots,M$ yields $\la P_0P_NQ_0Q_M \ra = \la i_1j_1i_2j_2 \ra  \geq 0$.

\end{proof}

\subsection{Example of a configuration involving higher length intersection points}

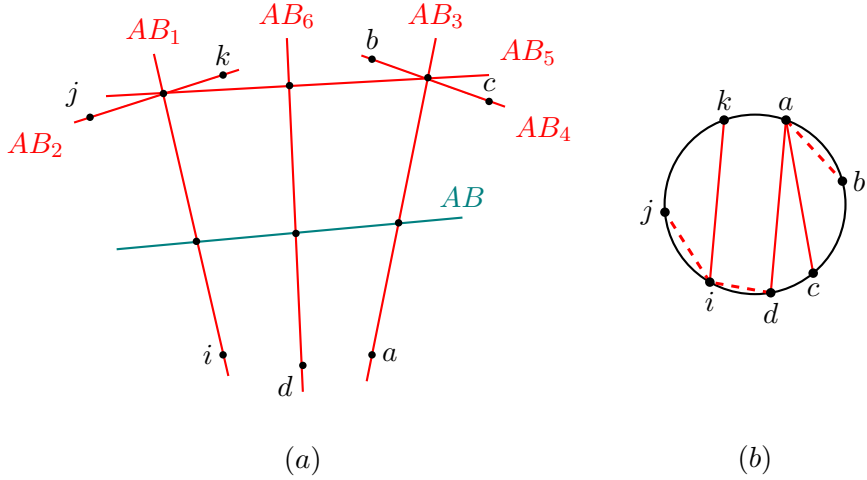
\begin{figure}[t]
\centering
\begin{tikzpicture}[scale = 0.7]

    \node at (0,-4) {$(a)$};

    \coordinate (i) at (-1.5,-2);
    \coordinate (J) at (-2,0.15);
    \coordinate (Y) at (1.8,0.5);
    \coordinate (R) at (-2.62,2.95);
    \coordinate (S) at (2.35,3.25);
    \coordinate (j) at (-4,2.5);
    \coordinate (k) at (-1.5,3.3);
    \coordinate (b) at (1.3,3.6);
    \coordinate (c) at (3.5,2.8);
    \coordinate (a) at (1.3,-2);
    \coordinate (d) at (0,-2.2);
    \coordinate (C) at (-2.8,3.7);
    \coordinate (D) at (-1.4,-2.4);
    \coordinate (A) at (-3.5,0);
    \coordinate (B) at (3,0.6);
    \coordinate (G) at (1.5,3.5);
    \coordinate (H) at (1.5,-2.5);
    \coordinate (E) at (-4.3,2.4);
    \coordinate (F) at (-1.2,3.4);
    \coordinate (L) at (1.1,3.67);
    \coordinate (M) at (3.8,2.7);
    \coordinate (G) at (2.5,4);
    \coordinate (H) at (1.2,-2.5);
    \coordinate (O) at (-3.7,2.9);
    \coordinate (N) at (3.5,3.3);
    \coordinate (P) at (0.0,-2.7);
    \coordinate (Q) at (-0.3,4);
    \coordinate (W) at (-0.13,0.3);
    \coordinate (Z) at (-0.25,3.1);


    \draw[dashed, gray] (j) -- (k);

    \draw[dashed, gray] (b) -- (c);


    \draw[red, thick] (C) -- (D);
    \node[above, red] at (C) {$AB_1$};

    \draw[red, thick] (G) -- (H);
    \node[above, red] at (G) {$AB_3$};

    \draw[red, thick] (E) -- (F);
    \node[below left, red] at (E) {$AB_2$};

    \draw[red, thick] (L) -- (M);
    \node[below right, red] at (M) {$AB_4$};

    \draw[red, thick] (N) -- (O);
    \node[above right, red] at (N) {$AB_5$};

    \draw[red, thick] (P) -- (Q);
    \node[above, red] at (Q) {$AB_6$};

    \draw[teal, thick] (A) -- (B);
    \node[above, teal] at (B) {$AB$};

    \fill (i) circle (2pt);
    \fill (J) circle (2pt);
    \fill (Y) circle (2pt);
    \fill (R) circle (2pt);
    \fill (S) circle (2pt);
    \fill (j) circle (2pt);
    \fill (k) circle (2pt);
    \fill (a) circle (2pt);
    \fill (b) circle (2pt);
    \fill (c) circle (2pt);
    \fill (d) circle (2pt);
    \fill (W) circle (2pt);
    \fill (Z) circle (2pt);

    \node[left] at (i) {$i$};
    \node[above left] at (j) {$j$};
    \node[above] at (k) {$k$};
    \node[right] at (a) {$a$};
    \node[above] at (b) {$b$};
    \node[above] at (c) {$c$};
    \node[below left] at (d) {$d$};


    \begin{scope}[scale = 0.85,xshift = 10 cm,yshift=1 cm]

    \node at (0,-5.7) {$(b)$};
    
    \draw[thick] (0,0) circle(2cm);

    \coordinate (i) at (-120:2cm);
    \coordinate (j) at (185:2cm);
    \coordinate (k) at (110:2cm);
    \coordinate (l) at (70:2cm);
    \coordinate (n) at (-50:2cm);
    \coordinate (m) at (15:2cm);
    \coordinate (o) at (-80:2cm);
    
    \node at (i) [below ] {$i$};
    \node at (j) [left] {$j$};
    \node at (k) [above ] {$k$};
    \node at (l) [above ] {$a$};
    \node at (m) [right ] {$b$};
    \node at (n) [below ] {$c$};
    \node at (o) [below ] {$d$};


    \draw[line width=0.4mm, red, dashed] (i) -- (j);
    \draw[line width=0.4mm, red, dashed] (l) -- (m);
    \draw[line width=0.4mm, red, thick] (i) -- (k);
    \draw[line width=0.4mm, red, thick] (l) -- (n);
    \draw[line width=0.4mm, red, dashed] (o) -- (i);
    \draw[line width=0.4mm, red, thick] (o) -- (l);

     \foreach \p in {i,j,k,l,m,n,o} {
        \fill (\p) circle(3pt);
    }

    \end{scope}

\end{tikzpicture}
\caption{A {\LABELREVERSE} LS configuration at $L=6$ involving intersection points up to length equal to three. The LS is compatible, if and only if the ordering of the indices is such that no arc in (b) crosses. If this is the case, the LS value is given by eq. \eqref{L6_LS_value}.}
\label{fig:L6}
\end{figure}

We give an example of a {\LABELREVERSE} LS configuration involving intersection points of length up to three. This example also illustrates the content of Subsection \ref{subapp: conditions on AB geometry} and Subsection \ref{subapp: crossing diagrams are incompatible}.

Let $L=6$ and localize the six loop lines as
\begin{equation}
\begin{aligned}
    &AB_1= i \, AB\cap(ijk) \ , \quad AB_2=jk \ , \quad  AB_3=a \, AB \cap(abc) \ , \\ 
    &AB_4=bc \ , \quad AB_5 = P[jk,i]P[bc,a] \ ,\quad AB_6= d \,  AB \cap (dAB_5) \ ,
\end{aligned}
\end{equation}
where all indices are in $\{1, \dots,n\}$ and the points $P[jk,i]$ were defined in eq. \eqref{simplest_invariants}. The line configuration, together with its associated signed collection of arcs, is depicted in Fig.~\ref{fig:L6}.
In particular, we can write $AB_6$ in terms of two intersection points $A_6=d$ and 
\begin{equation}\label{B6}
    B_6 = P[jk,i] + \gamma \, P[bc,a] =  i (j+ \alpha \, k) + \gamma \, (a (b + \beta \, c ))   \ ,
\end{equation}
with
\begin{equation}
    \alpha = \frac{\langle AB \, ij \rangle }{\langle AB \, ki \rangle} > 0  \ , \quad \beta = \frac{\langle AB \, ab \rangle}{\langle AB \, ca \rangle} >0 \ .
\end{equation}
Note that $B_6$ in eq. \eqref{B6} is an intersection point of length equal to three.
  The real constant $\gamma$ is fixed by imposing $\langle AB \, AB_6 \rangle = 0$, and it is equal to 
\begin{equation}\label{delta2}
    \gamma  =  \frac{\langle ABjk \rangle \langle ABca \rangle \langle ABdi \rangle}{\langle ABki \rangle \langle AB bc \rangle \langle ABad \rangle}  \ ,
\end{equation}
which can be verified using the Plücker relation. The one-loop Amplituhedron conditions on $AB_6$ forces an ordering on the indices: either as $j < k \leq d  \leq  b<c$, in which case $\gamma < 0$, or as $b<c \leq d \leq j$, in which case $\gamma > 0$.
Moreover, if present in the negative geometry, the mutual negativity condition between $AB$ and $AB_5$ reads
\begin{equation}\label{ABNO}
    \langle AB \, AB_5 \rangle = \frac{\langle AB jk \rangle\langle AB bc \rangle \langle AB ia \rangle}{\langle AB ki \rangle\langle AB ca \rangle} < 0 \ .
\end{equation}
Eqs. \eqref{delta2} and \eqref{ABNO} impose conditions on the geometry in $AB$, coming from the one-loop Amplituhedron conditions for $AB_6$ and the mutual negativity condition $\langle AB \, AB_5 \rangle < 0$, respectively. All these conditions are of the form of imposing a fixed sign on products of twistor brackets $\la AB ij \ra$, as claimed in Subsection \ref{subapp: conditions on AB geometry}.

The LS configuration is compatible if and only if the indices are ordered such that no arc in Fig. \ref{fig:L6}(b) crosses. If this is the case, the LS value is given by
\begin{equation}\label{L6_LS_value}
    \Omega_n(ik,ac) -  \Omega_n(ik,ab) - \Omega_n(ij,ac) +  \Omega_n(ij,ab) \ .
\end{equation}
Eq. \eqref{L6_LS_value} can be verified using a triangulation-based argument, similar to the proof of Proposition~\ref{proposition from triang to kermits}.

\section{Leading singularity count and basis}

\subsection{Leading singularity count and basis of Kermit forms}
\label{subapp:LS count}

By a simple combinatorial count the number of Kermits at fixed $n \geq 4$ is 
\begin{equation}
    3 \binom{n}{6} + 5 \binom{n}{5} + \binom{n}{4} \,,
\end{equation}
which is the ``$k=2$'' row in \cite[Table 1] {Lukowski:2019sxw}. 
However, the Kermit forms given in eq. \eqref{six_invariant}, as well as the simple compatible LS \eqref{LS_simple_formulae}, are not linearly independent for $n \geq 5$. 
The authors of ref. \cite{Chicherin:2022bov} conjectured the dimension of the linear space generated by LS of $F_n^{(L)}$ for $L \ge 2$ to be equal to 
\begin{equation}\label{number_lin_ind}
      \binom{n}{4} + \binom{n-1}{4} = \frac{(n-1)(n-2)^{2}(n-3)}{12} \,.
\end{equation}
In this section we prove that this count is indeed correct, and provide an explicit basis for the linear space generated by Kermit forms.
\begin{tcolorbox}[colback=white]
For every $n \geq 4$, the collection of Kermit forms given by
\begin{equation}\label{Kermit_basis}
\begin{aligned}
    &[abcd] \quad \quad \ \text{with} \quad 1 \leq a < b <c < d \leq n \ , \ \text{and}  \\
    &[1ab;1cd] \quad \text{with} \quad 1<a<b<c<d\leq n \ ,
\end{aligned}
\end{equation}
forms a basis of the linear space generated by all Kermit forms \eqref{six_invariant} and \eqref{four_invariant}. In particular, the dimension of this linear space is given by eq. \eqref{number_lin_ind}.
\end{tcolorbox}
\begin{proof}
{\it{\underline{Generating set: }}} 
We first show that the forms in eq. (\ref{Kermit_basis}) generate the full linear space of Kermit forms. 
For that, we consider the following linear relations between Kermit forms. For $n \geq 5$, we have the equality
\begin{equation}\label{linear_relations}
    \sum_{\Delta_1, \, \Delta_2 \, \subset \, T} [\Delta_1; \Delta_2] = \sum_{\Delta_1', \, \Delta_2' \, \subset \, T'} [\Delta_1'; \Delta_2'] \,.
\end{equation}
Here the sums run over all pairs of triangles in the triangulations $T, T'$ of a pair of polygons $P_1$ and $P_2$ inscribed in an $n$-gon, and $\Delta_i \subset P_i$. Eq. \eqref{linear_relations} is true for any non-overlapping pair of polygons $P_i$ inscribed in the $n$-gon; we allow the polygons to share some vertices, with the condition that either $P_1 = P_2$, or that the interiors of the polygons do not intersect. The proof of eq. \eqref{linear_relations} uses the same arguments introduced in Subsection \ref{subsec: formulae for simple compatible LS}. In particular, both sides of eq. \eqref{linear_relations} can be expressed as a linear combination of simple compatible LS \eqref{LS_simple_formulae} depending only on the vertex indices of $P_1$ and $P_2$.

Let us come back to the proof. We first show that all Kermit forms for which one index is equal to $1$ can be expressed as linear combinations of Kermit of the form as in eq. \eqref{Kermit_basis}. Consider first $[1bc;def]$ with $1<b<c\leq d<e<f \leq n $. Then, we consider the two polygons $P_1=\{1,b,c\}$ and $P_2=\{1,d,e,f\}$ and write eq. \eqref{linear_relations}, for the only two possible triangulations,  
\begin{equation}\label{reduction1}
    [1bc;1df] + [1bc;def] = [1bc;1de] + [1bc;1ef] \, .
\end{equation}
This shows that $[1bc;def]$ can be expressed only in terms of elements in eq. \eqref{Kermit_basis}. Consider now $[a1c;def]$ with $a<1<c\leq d<e<f \leq 1 $. Take the polygons $P_1=\{a,1,c,d\}$ and $P_2=\{d,e,f\}$ and write eq. \eqref{linear_relations}, for the only two possible triangulations,  
\begin{equation}\label{reduction2}
    [a1c;def] + [a1cd] = [1cd;def] + [a1d;def] \, .
\end{equation} 
We therefore need to show the claim for $[a1d;def]$. For that, take a single polygon $P_1=P-2=\{a,1,c,d,e\}$, with triangulations $T=\{ac,ce\}$ and $T'=\{1d,1e\}$. Then, eq. \eqref{linear_relations} reads
\begin{equation}\label{reductio22}
    [a1ce] + [acde]+ [a1c;ced] = [1cde] + [a1de] + [1cd;a1e] \, .
\end{equation} 
This concludes the proof for the case where one index is equal to $1$.

Let us now consider a general Kermit form $[abc;def]$ with $a<b<c\leq d<e<f \leq 1 $. If one or two indices are equal to $1$, then by above we are done. So let us assume that $1<a<b<c\leq d<e<f \leq n $. Thanks to the rotational symmetry of the problem in the indices, we can restrict to considering only two cases, specified by the position of $1$ relative to $a$. If $a<1<b$, then we set the polygons to be $P_1=\{a,1,b,c\}$ and $P_2=\{d,e,f\}$, and the possible triangulations are only two. Then, eq. \eqref{linear_relations} reads
\begin{equation}
\begin{aligned}
    [abc;def] + [a1b;def] = [1bc;def] + [a1c;def] \, ,
\end{aligned}
\end{equation}
which together with the previous analysis completes the proof for this case.
For the other case, in which $f < 1 < a$, we set the polygons to be $P_1=\{1,a,b,c\}$ and $P_2=\{1,d,e,f\}$, and the triangulations $T=\{(ac),(df)\}$, $T'=\{(1b),(1e)\}$. Then, eq. \eqref{linear_relations} becomes
\begin{equation}
\begin{aligned}
    &[abc;def] + [abc;1df] + [1ab;def] + [1ab;1df] = \\ &[1ab;1de] + [1ab;1ef] + [1bc;1de] + [1bc;1ef] \, ,
\end{aligned}
\end{equation}
which proves the claim also in this case. This concludes the proof that the set in eq. \eqref{Kermit_basis} generates the linear space of all Kermit forms. \vspace{0.1in}

{\it{\underline{Linear independence: }}} 
We are left to show that the set in eq. \eqref{Kermit_basis} is linearly independent. 
For that take any linear combination of the form 
\begin{equation}\label{lin_rel}
    \sum_{1 \leq a < b <c < d \leq n} \lambda(a,b,c,d) \, [abcd] + \sum_{1 < a < b <c < d \leq n} \mu(a,b,c,d) \, [1ab;1cd] = 0 \,,
\end{equation}
for complex coefficients $\lambda(a,b,c,d)$ and $\mu(a,b,c,d)$. We show that all coefficients must be equal to zero, by taking different residues of eq. \eqref{lin_rel} corresponding to higher codimension loci in $AB$. For indices $1<i<j<k<l \leq n$ consider the double residue along
\begin{equation}\label{residue}
    \la AB il \ra = \la AB jk \ra =  0 \,.
\end{equation}
Then, the only term in eq. \eqref{lin_rel} with non-zero residue is $ [ijkl]$. In fact, no other Kermit form among those in \eqref{Kermit_basis} involves both poles in eq. \eqref{residue}. It follows that $\lambda(a,b,c,d)$ vanish if $1<a<b<c<d \leq n$, and the only remaining Kermit forms involving four indices in eq. \eqref{lin_rel} are of the form $[1bcd]$. 
Consider now instead the residue of eq. \eqref{lin_rel} along
\begin{equation}\label{residue2}
    \la AB ij \ra = \la AB kl \ra =  0 \,,
\end{equation}
for $1<i<j<k<l \leq n$.
The only term with non-zero residue on eq. \eqref{residue2} is $ [1ij;1kl]$. In fact, no Kermit form $[1bcd]$ can have both poles in eq. \eqref{residue2}. This shows that all $\mu(a,b,c,d)$ vanish. The only remaining terms are of the form $[1bcd]$, which are individually isolated by taking e.g. the double residue along 
\begin{equation}\label{residue3}
    \la AB 1j \ra = \la AB kl \ra =  0 \,,
\end{equation}
for $1<j<k<l \leq n$. By the same argument as before, all coefficients $\lambda(1,b,c,d)$ vanish. This concludes the proof.
\end{proof}

\subsection{A leading singularity basis in terms of simple compatible LS}

From Claim \ref{claim 1} we know that for $L \geq 2$ the linear space generated by all Kermit forms and that of all LS of $F_n^{(L)}$ coincide. Moreover, by Subsection \ref{subsec: Kermit forms in terms of simple compatible LS} this space is also the same as the one generated by only simple compatible LS \eqref{LS_simple_formulae}.
In particular, using eqs. \eqref{six_kermit_in_LS} and \eqref{four_kermit_in_LS}, we can rewrite the basis \eqref{Kermit_basis} in terms of simple compatible LS. Note that this basis involves then non-trivial linear combinations of simple compatible LS. 
We can instead find a basis consisting of individual simple compatible LS.

As a warm-up, consider the one-loop case. From the result for $F^{(1)}_n$ from ref. \cite{Chicherin:2022bov}, we know that a basis of simple compatible LS at $L=1$ is given by $\Omega_n(ij)$. There are $n (n-3)/2$ such terms.

For $L \geq 2$,
 we conjecture the following.
\begin{tcolorbox}[colback=white]
A basis of simple compatible LS at $L \geq 2$ is given by 
\begin{equation}\label{LS_lin_dependent0}
    \Omega_n \,, \quad  \Omega_{n}(ij,kl) \quad \text{with} \quad i<j \leq k<l \,,
\end{equation}
without all those of the form
\begin{equation}\label{LS_lin_dependent}
    \Omega_n(ij,ik) \quad \text{with} \quad 1 \leq i < j <k \leq n \ \text{and} \ |j-i|,|k-j|,|i-k|>1 \,.
\end{equation}
\end{tcolorbox}
One may check that indeed the number of terms in eqs. \eqref{LS_lin_dependent0}  and \eqref{LS_lin_dependent} is equal to eq. \eqref{number_lin_ind}.

\subsection{All-loop leading singularity basis 
at five and six points}
\label{subapp:LS basis at n=5 and n=6}

Here we give the explicit basis of LS of $F^{(L)}_n$ in terms of Kermit forms, according to eq.~\eqref{Kermit_basis}.

Let $n=5$. At $L=1$ we have five linearly independent LS, given by $\Omega_{5}(13)$, cf. eq. (\ref{R5,13}),
and its four cyclic shifts.
For $L \geq 2$ there are six linearly independent LS. From the discussion we deduce that we can choose a
basis given by $\Omega_5$, as well as $[1234]$ and its five cyclic shifts. This confirms that the five-point basis choice considered in \cite{Chicherin:2022zxo} is indeed valid at all loop orders. 
We remark that if we replace $\Omega_5$ by $[123;145]$, we obtain a basis of only Kermit forms, as shown in Fig.~\ref{Fig:basisn=5intermsofkermitformsonly}.
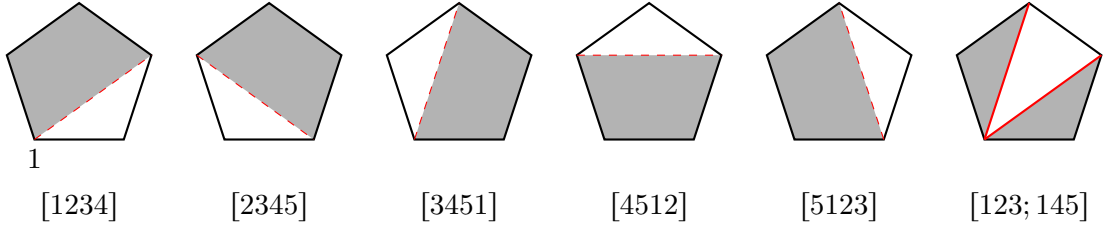
\begin{figure}[t]
\centering
\begin{tikzpicture}[scale = 1]

    \begin{scope}[xshift = 0cm]

    \node at (0,-1.7) {$[1234]$};

    \draw[thick] (90:1) 
    \foreach \x in {162,234,306,18} 
        { -- (\x:1) } -- cycle;

    \coordinate (D) at (18:1);
    \coordinate (C) at (90:1);
    \coordinate (B) at (162:1);
    \coordinate (A) at (234:1);
    \coordinate (E) at (306:1);

    \fill[black, opacity=0.3] (A) -- (B) -- (C) -- (D) -- cycle;

    \draw[red, dashed] (A) -- (D);

    \node[below] at (A) {1};
    
    \end{scope}


    \begin{scope}[xshift = 2.5cm]

    \node at (0,-1.7) {$[2345]$};

    \draw[thick] (90:1) 
    \foreach \x in {162,234,306,18} 
        { -- (\x:1) } -- cycle;

    \coordinate (D) at (18:1);
    \coordinate (C) at (90:1);
    \coordinate (B) at (162:1);
    \coordinate (A) at (234:1);
    \coordinate (E) at (306:1);

    \fill[black, opacity=0.3] (B) -- (C) -- (D) -- (E) -- cycle;

    \draw[red, dashed] (B) -- (E);

    
    \end{scope}


    \begin{scope}[xshift = 5cm]

    \node at (0,-1.7) {$[3451]$};

    \draw[thick] (90:1) 
    \foreach \x in {162,234,306,18} 
        { -- (\x:1) } -- cycle;

    \coordinate (D) at (18:1);
    \coordinate (C) at (90:1);
    \coordinate (B) at (162:1);
    \coordinate (A) at (234:1);
    \coordinate (E) at (306:1);

    \fill[black, opacity=0.3] (C) -- (D) -- (E) -- (A) -- cycle;

    \draw[red, dashed] (A) -- (C);

    
    \end{scope}


    \begin{scope}[xshift = 7.5cm]

    \node at (0,-1.7) {$[4512]$};

    \draw[thick] (90:1) 
    \foreach \x in {162,234,306,18} 
        { -- (\x:1) } -- cycle;

    \coordinate (D) at (18:1);
    \coordinate (C) at (90:1);
    \coordinate (B) at (162:1);
    \coordinate (A) at (234:1);
    \coordinate (E) at (306:1);

    \fill[black, opacity=0.3] (D) -- (E) -- (A) -- (B) -- cycle;

    \draw[red, dashed] (B) -- (D);

    
    \end{scope}

    \begin{scope}[xshift = 10cm]

    \node at (0,-1.7) {$[5123]$};

    \draw[thick] (90:1) 
    \foreach \x in {162,234,306,18} 
        { -- (\x:1) } -- cycle;

    \coordinate (D) at (18:1);
    \coordinate (C) at (90:1);
    \coordinate (B) at (162:1);
    \coordinate (A) at (234:1);
    \coordinate (E) at (306:1);

    \fill[black, opacity=0.3] (E) -- (A) -- (B) -- (C) -- cycle;

    \draw[red, dashed] (E) -- (C);

    
    \end{scope}

    \begin{scope}[xshift = 12.5cm]

    \node at (0,-1.7) {$[123;145] $};

    \draw[thick] (90:1) 
    \foreach \x in {162,234,306,18} 
        { -- (\x:1) } -- cycle;

    \coordinate (D) at (18:1);
    \coordinate (C) at (90:1);
    \coordinate (B) at (162:1);
    \coordinate (A) at (234:1);
    \coordinate (E) at (306:1);

    \fill[black, opacity=0.3] (A) -- (B) -- (C)  -- cycle;
    \fill[black, opacity=0.3] (A) -- (D) -- (E) -- cycle;

    \draw[red, thick] (A) -- (D);
    \draw[red, thick] (A) -- (C);

    
    \end{scope}

\end{tikzpicture}
\caption{A LS basis for $F_5^{(L)}$ at $L \geq 2$ that involves Kermit forms only.}
\label{Fig:basisn=5intermsofkermitformsonly}
\end{figure}

Let $n=6$. At $L=1$ we have $9$ leading singularities,
given by $\Omega_{6}(13)$ (plus five cyclic shifts) and $\Omega_{6}(14)$ (plus two cyclic shifts).
For $L \geq 2$, the number of independent LS is $20$. A basis of Kermit forms is shown in Fig.~\ref{Fig:basisn=6intermsofkermitformsonly}.

\begin{figure}[t]
\centering
\begin{tikzpicture}

    \node at (0,-1.7) {$[1234]$};
    
    \draw[thick] (0:1) 
        \foreach \x in {60,120,180,240,300} 
            { -- (\x:1) } -- cycle;

    \coordinate (E) at (0:1);
    \coordinate (D) at (60:1);
    \coordinate (C) at (120:1);
    \coordinate (B) at (180:1);
    \coordinate (A) at (240:1);
    \coordinate (F) at (300:1);

    \fill[black, opacity=0.3] (A) -- (B) -- (C) -- (D) -- cycle;

    \draw[red, dashed] (A) -- (D);
    
    \node[below] at (A) {1};

    
    \begin{scope}[xshift = 3 cm]
    
    \node at (0,-1.7) {$[1345]$};

    \draw[thick] (0:1) 
        \foreach \x in {60,120,180,240,300} 
            { -- (\x:1) } -- cycle;

    \coordinate (E) at (0:1);
    \coordinate (D) at (60:1);
    \coordinate (C) at (120:1);
    \coordinate (B) at (180:1);
    \coordinate (A) at (240:1);
    \coordinate (F) at (300:1);

    \fill[black, opacity=0.3] (A) -- (C) -- (D) -- (E) -- cycle;

    \draw[red, dashed] (A) -- (C);
    \draw[red, dashed] (A) -- (E);

    \end{scope}

    
    \begin{scope}[xshift = 6cm]
    
    \node at (0,-1.7) {$[1346]$};
    \node at (3,0) {$+$ cycl.};

    \draw[thick] (0:1) 
        \foreach \x in {60,120,180,240,300} 
            { -- (\x:1) } -- cycle;

    \coordinate (E) at (0:1);
    \coordinate (D) at (60:1);
    \coordinate (C) at (120:1);
    \coordinate (B) at (180:1);
    \coordinate (A) at (240:1);
    \coordinate (F) at (300:1);

    \fill[black, opacity=0.3] (A) -- (C) -- (D) -- (F) -- cycle;
    
    \draw[red, dashed] (A) -- (C);
    \draw[red, dashed] (D) -- (F);
 
    \end{scope}

    
    \begin{scope}[xshift = 0cm, yshift = -3.5 cm]
    
    \node at (0,-1.7) {$[123;145]$};

    \draw[thick] (0:1) 
        \foreach \x in {60,120,180,240,300} 
            { -- (\x:1) } -- cycle;

    \coordinate (E) at (0:1);
    \coordinate (D) at (60:1);
    \coordinate (C) at (120:1);
    \coordinate (B) at (180:1);
    \coordinate (A) at (240:1);
    \coordinate (F) at (300:1);

    \fill[black, opacity=0.3] (A) -- (B) -- (C) -- cycle;
    \fill[black, opacity=0.3] (A) -- (D) -- (E) -- cycle;

    \draw[red, thick] (A) -- (C);
    \draw[red, thick] (A) -- (D);
    \draw[red, dashed] (A) -- (E);
    
    \end{scope}
    
    
    \begin{scope}[xshift = 3 cm, yshift = -3.5 cm]
    
    \node at (0,-1.7) {$[123;156]$};

    \draw[thick] (0:1) 
        \foreach \x in {60,120,180,240,300} 
            { -- (\x:1) } -- cycle;

    \coordinate (E) at (0:1);
    \coordinate (D) at (60:1);
    \coordinate (C) at (120:1);
    \coordinate (B) at (180:1);
    \coordinate (A) at (240:1);
    \coordinate (F) at (300:1);

    \fill[black, opacity=0.3] (A) -- (B) -- (C) -- cycle;
    \fill[black, opacity=0.3] (A) -- (E) -- (F) -- cycle;

    \draw[red, thick] (A) -- (C);
    \draw[red, dashed] (A) -- (E);
    
    \end{scope}
    
    
    \begin{scope}[xshift = 6cm, yshift = -3.5 cm]
    
    \node at (0,-1.7) {$[123;146]$};

    \draw[thick] (0:1) 
        \foreach \x in {60,120,180,240,300} 
            { -- (\x:1) } -- cycle;

    \coordinate (E) at (0:1);
    \coordinate (D) at (60:1);
    \coordinate (C) at (120:1);
    \coordinate (B) at (180:1);
    \coordinate (A) at (240:1);
    \coordinate (F) at (300:1);

    \fill[black, opacity=0.3] (A) -- (B) -- (C) -- cycle;
    \fill[black, opacity=0.3] (A) -- (D) -- (F) -- cycle;

    \draw[red, thick] (A) -- (C);
    \draw[red, thick] (A) -- (D);
    \draw[red, dashed] (D) -- (F);
    
    \end{scope}
    
    \begin{scope}[xshift = 9cm, yshift = -3.5 cm]
    
    \node at (0,-1.7) {$[124;156]$};

    \draw[thick] (0:1) 
        \foreach \x in {60,120,180,240,300} 
            { -- (\x:1) } -- cycle;

    \coordinate (E) at (0:1);
    \coordinate (D) at (60:1);
    \coordinate (C) at (120:1);
    \coordinate (B) at (180:1);
    \coordinate (A) at (240:1);
    \coordinate (F) at (300:1);

    \fill[black, opacity=0.3] (A) -- (B) -- (D) -- cycle;
    \fill[black, opacity=0.3] (A) -- (E) -- (F) -- cycle;

    \draw[red, dashed] (B) -- (D);
    \draw[red, thick] (A) -- (D);
    \draw[red, thick] (A) -- (E);
    
    \end{scope}
    
    \begin{scope}[xshift = 12cm, yshift = -3.5 cm]
    
    \node at (0,-1.7) {$[134;156]$};

    \draw[thick] (0:1) 
        \foreach \x in {60,120,180,240,300} 
            { -- (\x:1) } -- cycle;

    \coordinate (E) at (0:1);
    \coordinate (D) at (60:1);
    \coordinate (C) at (120:1);
    \coordinate (B) at (180:1);
    \coordinate (A) at (240:1);
    \coordinate (F) at (300:1);

    \fill[black, opacity=0.3] (A) -- (C) -- (D) -- cycle;
    \fill[black, opacity=0.3] (A) -- (E) -- (F) -- cycle;

    \draw[red, dashed] (A) -- (C);
    \draw[red, thick] (A) -- (D);
    \draw[red, thick] (A) -- (E);
    
    \end{scope}

\end{tikzpicture}
\caption{A LS basis for $F_6^{(L)}$ at $L \geq 2$.
``+ cycl'' in the first row means that there are $6$ elements of [1234] type, $6$ of [1345] type, $3$ of [1346] type, so that together with the 5 elements from the second row we have $20$ basis elements in total.}
\label{Fig:basisn=6intermsofkermitformsonly}
\end{figure}
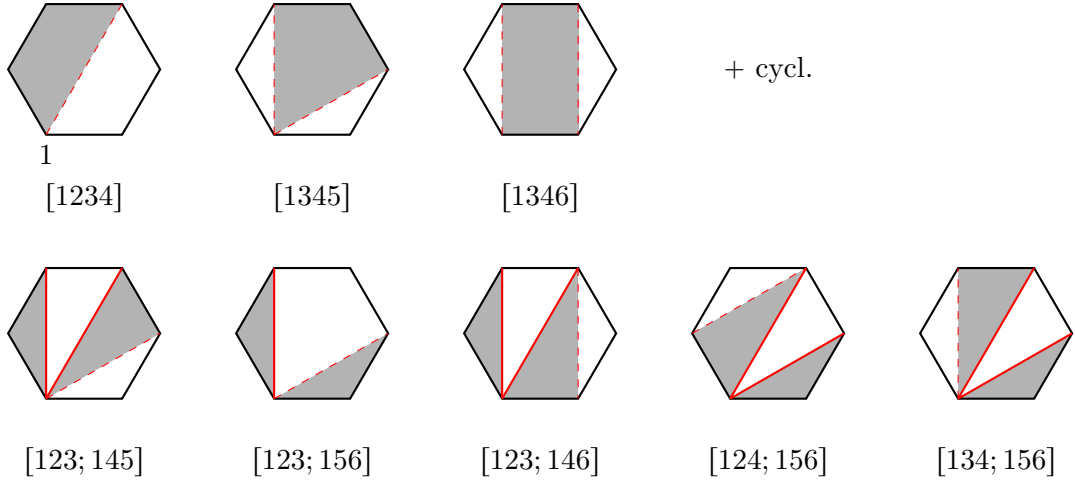

\bibliographystyle{JHEP}
\bibliography{main}

\providecommand{\href}[2]{#2}\begingroup\raggedright\begin{thebibliography}{10}

\bibitem{Henn:2020omi}
J.~M. Henn, \emph{{What Can We Learn About QCD and Collider Physics from N=4 Super Yang\textendash{}Mills?}}, \href{http://dx.doi.org/10.1146/annurev-nucl-102819-100428}{\emph{Ann. Rev. Nucl. Part. Sci.} {\bf 71} (2021) 87--112}, [\href{http://arxiv.org/abs/2006.00361}{{\tt 2006.00361}}].

\bibitem{Bern:2022jnl}
Z.~Bern and J.~Trnka, \emph{{Snowmass TF04 Report: Scattering Amplitudes and their Applications}},  \href{http://arxiv.org/abs/2210.03146}{{\tt 2210.03146}}.

\bibitem{Travaglini:2022uwo}
G.~Travaglini et~al., \emph{{The SAGEX review on scattering amplitudes}}, \href{http://dx.doi.org/10.1088/1751-8121/ac8380}{\emph{J. Phys. A} {\bf 55} (2022) 443001}, [\href{http://arxiv.org/abs/2203.13011}{{\tt 2203.13011}}].

\bibitem{Alday:2008yw}
L.~F. Alday and R.~Roiban, \emph{{Scattering Amplitudes, Wilson Loops and the String/Gauge Theory Correspondence}}, \href{http://dx.doi.org/10.1016/j.physrep.2008.08.002}{\emph{Phys. Rept.} {\bf 468} (2008) 153--211}, [\href{http://arxiv.org/abs/0807.1889}{{\tt 0807.1889}}].

\bibitem{Henn:2009bd}
J.~M. Henn, \emph{{Duality between Wilson loops and gluon amplitudes}}, \href{http://dx.doi.org/10.1002/prop.200900048}{\emph{Fortsch. Phys.} {\bf 57} (2009) 729--822}, [\href{http://arxiv.org/abs/0903.0522}{{\tt 0903.0522}}].

\bibitem{Drummond:2008vq}
J.~M. Drummond, J.~Henn, G.~P. Korchemsky and E.~Sokatchev, \emph{{Dual superconformal symmetry of scattering amplitudes in N=4 super-Yang-Mills theory}}, \href{http://dx.doi.org/10.1016/j.nuclphysb.2009.11.022}{\emph{Nucl. Phys. B} {\bf 828} (2010) 317--374}, [\href{http://arxiv.org/abs/0807.1095}{{\tt 0807.1095}}].

\bibitem{Berkovits:2008ic}
N.~Berkovits and J.~Maldacena, \emph{{Fermionic T-Duality, Dual Superconformal Symmetry, and the Amplitude/Wilson Loop Connection}}, \href{http://dx.doi.org/10.1088/1126-6708/2008/09/062}{\emph{JHEP} {\bf 09} (2008) 062}, [\href{http://arxiv.org/abs/0807.3196}{{\tt 0807.3196}}].

\bibitem{Drummond:2009fd}
J.~M. Drummond, J.~M. Henn and J.~Plefka, \emph{{Yangian symmetry of scattering amplitudes in N=4 super Yang-Mills theory}}, \href{http://dx.doi.org/10.1088/1126-6708/2009/05/046}{\emph{JHEP} {\bf 05} (2009) 046}, [\href{http://arxiv.org/abs/0902.2987}{{\tt 0902.2987}}].

\bibitem{Caron-Huot:2020bkp}
S.~Caron-Huot, L.~J. Dixon, J.~M. Drummond, F.~Dulat, J.~Foster, O.~G\"urdo\u{g}an, M.~von Hippel, A.~J. McLeod and G.~Papathanasiou, \emph{{The Steinmann Cluster Bootstrap for $N$ = 4 Super Yang-Mills Amplitudes}}, \href{http://dx.doi.org/10.22323/1.376.0003}{\emph{PoS} {\bf CORFU2019} (2020) 003}, [\href{http://arxiv.org/abs/2005.06735}{{\tt 2005.06735}}].

\bibitem{Goncharov:2010jf}
A.~B. Goncharov, M.~Spradlin, C.~Vergu and A.~Volovich, \emph{{Classical Polylogarithms for Amplitudes and Wilson Loops}}, \href{http://dx.doi.org/10.1103/PhysRevLett.105.151605}{\emph{Phys. Rev. Lett.} {\bf 105} (2010) 151605}, [\href{http://arxiv.org/abs/1006.5703}{{\tt 1006.5703}}].

\bibitem{Arkani-Hamed:2009ljj}
N.~Arkani-Hamed, F.~Cachazo, C.~Cheung and J.~Kaplan, \emph{{A Duality For The S Matrix}}, \href{http://dx.doi.org/10.1007/JHEP03(2010)020}{\emph{JHEP} {\bf 03} (2010) 020}, [\href{http://arxiv.org/abs/0907.5418}{{\tt 0907.5418}}].

\bibitem{Arkani-Hamed:2012zlh}
N.~Arkani-Hamed, J.~L. Bourjaily, F.~Cachazo, A.~B. Goncharov, A.~Postnikov and J.~Trnka, \emph{{Grassmannian Geometry of Scattering Amplitudes}}.
\newblock Cambridge University Press, 4, 2016.
\newblock 10.1017/CBO9781316091548.

\bibitem{Arkani-Hamed:2014bca}
N.~Arkani-Hamed, J.~L. Bourjaily, F.~Cachazo, A.~Postnikov and J.~Trnka, \emph{{On-Shell Structures of MHV Amplitudes Beyond the Planar Limit}}, \href{http://dx.doi.org/10.1007/JHEP06(2015)179}{\emph{JHEP} {\bf 06} (2015) 179}, [\href{http://arxiv.org/abs/1412.8475}{{\tt 1412.8475}}].

\bibitem{Paranjape:2022ymg}
S.~Paranjape, J.~Trnka and M.~Zheng, \emph{{Non-planar BCFW Grassmannian geometries}}, \href{http://dx.doi.org/10.1007/JHEP12(2022)084}{\emph{JHEP} {\bf 12} (2022) 084}, [\href{http://arxiv.org/abs/2208.02262}{{\tt 2208.02262}}].

\bibitem{Brown:2022wqr}
T.~V. Brown, U.~Oktem and J.~Trnka, \emph{{Poles at infinity in on-shell diagrams}}, \href{http://dx.doi.org/10.1007/JHEP02(2023)003}{\emph{JHEP} {\bf 02} (2023) 003}, [\href{http://arxiv.org/abs/2212.06840}{{\tt 2212.06840}}].

\bibitem{Arkani-Hamed:2013jha}
N.~Arkani-Hamed and J.~Trnka, \emph{{The Amplituhedron}}, \href{http://dx.doi.org/10.1007/JHEP10(2014)030}{\emph{JHEP} {\bf 10} (2014) 030}, [\href{http://arxiv.org/abs/1312.2007}{{\tt 1312.2007}}].

\bibitem{Arkani-Hamed:2013kca}
N.~Arkani-Hamed and J.~Trnka, \emph{{Into the Amplituhedron}}, \href{http://dx.doi.org/10.1007/JHEP12(2014)182}{\emph{JHEP} {\bf 12} (2014) 182}, [\href{http://arxiv.org/abs/1312.7878}{{\tt 1312.7878}}].

\bibitem{Arkani-Hamed:2017vfh}
N.~Arkani-Hamed, H.~Thomas and J.~Trnka, \emph{{Unwinding the Amplituhedron in Binary}}, \href{http://dx.doi.org/10.1007/JHEP01(2018)016}{\emph{JHEP} {\bf 01} (2018) 016}, [\href{http://arxiv.org/abs/1704.05069}{{\tt 1704.05069}}].

\bibitem{Damgaard:2019ztj}
D.~Damgaard, L.~Ferro, T.~Lukowski and M.~Parisi, \emph{{The Momentum Amplituhedron}}, \href{http://dx.doi.org/10.1007/JHEP08(2019)042}{\emph{JHEP} {\bf 08} (2019) 042}, [\href{http://arxiv.org/abs/1905.04216}{{\tt 1905.04216}}].

\bibitem{Ferro:2022abq}
L.~Ferro and T.~Lukowski, \emph{{The Loop Momentum Amplituhedron}}, \href{http://dx.doi.org/10.1007/JHEP05(2023)183}{\emph{JHEP} {\bf 05} (2023) 183}, [\href{http://arxiv.org/abs/2210.01127}{{\tt 2210.01127}}].

\bibitem{Herrmann:2022nkh}
E.~Herrmann and J.~Trnka, \emph{{The SAGEX review on scattering amplitudes Chapter 7: Positive geometry of scattering amplitudes}}, \href{http://dx.doi.org/10.1088/1751-8121/ac8709}{\emph{J. Phys. A} {\bf 55} (2022) 443008}, [\href{http://arxiv.org/abs/2203.13018}{{\tt 2203.13018}}].

\bibitem{De:2024bpk}
S.~De, D.~Pavlov, M.~Spradlin and A.~Volovich, \emph{{From Feynman diagrams to the amplituhedron: a gentle review}},  in \emph{{Positive Geometry}}, 10, 2024.
\newblock \href{http://arxiv.org/abs/2410.11757}{{\tt 2410.11757}}.

\bibitem{Ranestad:2025qay}
K.~Ranestad, B.~Sturmfels and S.~Telen, \emph{{What is Positive Geometry?}},  \href{http://arxiv.org/abs/2502.12815}{{\tt 2502.12815}}.

\bibitem{Franco:2014csa}
S.~Franco, D.~Galloni, A.~Mariotti and J.~Trnka, \emph{{Anatomy of the Amplituhedron}}, \href{http://dx.doi.org/10.1007/JHEP03(2015)128}{\emph{JHEP} {\bf 03} (2015) 128}, [\href{http://arxiv.org/abs/1408.3410}{{\tt 1408.3410}}].

\bibitem{Dennen:2016mdk}
T.~Dennen, I.~Prlina, M.~Spradlin, S.~Stanojevic and A.~Volovich, \emph{{Landau Singularities from the Amplituhedron}}, \href{http://dx.doi.org/10.1007/JHEP06(2017)152}{\emph{JHEP} {\bf 06} (2017) 152}, [\href{http://arxiv.org/abs/1612.02708}{{\tt 1612.02708}}].

\bibitem{Arkani-Hamed:2018rsk}
N.~Arkani-Hamed, C.~Langer, A.~Yelleshpur~Srikant and J.~Trnka, \emph{{Deep Into the Amplituhedron: Amplitude Singularities at All Loops and Legs}}, \href{http://dx.doi.org/10.1103/PhysRevLett.122.051601}{\emph{Phys. Rev. Lett.} {\bf 122} (2019) 051601}, [\href{http://arxiv.org/abs/1810.08208}{{\tt 1810.08208}}].

\bibitem{Herrmann:2020qlt}
E.~Herrmann, C.~Langer, J.~Trnka and M.~Zheng, \emph{{Positive geometry, local triangulations, and the dual of the Amplituhedron}}, \href{http://dx.doi.org/10.1007/JHEP01(2021)035}{\emph{JHEP} {\bf 01} (2021) 035}, [\href{http://arxiv.org/abs/2009.05607}{{\tt 2009.05607}}].

\bibitem{Dian:2022tpf}
G.~Dian, P.~Heslop and A.~Stewart, \emph{{Internal boundaries of the loop amplituhedron}}, \href{http://dx.doi.org/10.21468/SciPostPhys.15.3.098}{\emph{SciPost Phys.} {\bf 15} (2023) 098}, [\href{http://arxiv.org/abs/2207.12464}{{\tt 2207.12464}}].

\bibitem{He:2023rou}
S.~He, Y.-t. Huang and C.-K. Kuo, \emph{{The ABJM Amplituhedron}}, \href{http://dx.doi.org/10.1007/JHEP09(2023)165}{\emph{JHEP} {\bf 09} (2023) 165}, [\href{http://arxiv.org/abs/2306.00951}{{\tt 2306.00951}}].

\bibitem{Arkani-Hamed:2023epq}
N.~Arkani-Hamed, W.~Flieger, J.~M. Henn, A.~Schreiber and J.~Trnka, \emph{{Coulomb Branch Amplitudes from a Deformed Amplituhedron Geometry}}, \href{http://dx.doi.org/10.1103/PhysRevLett.132.211601}{\emph{Phys. Rev. Lett.} {\bf 132} (2024) 211601}, [\href{http://arxiv.org/abs/2311.10814}{{\tt 2311.10814}}].

\bibitem{Ferro:2024vwn}
L.~Ferro, R.~Glew, T.~Lukowski and J.~Stalknecht, \emph{{The Two-loop MHV Momentum Amplituhedron from Fibrations of Fibrations}},  \href{http://arxiv.org/abs/2407.12906}{{\tt 2407.12906}}.

\bibitem{Dian:2024hil}
G.~Dian, E.~Mazzucchelli and F.~Tellander, \emph{{The two-loop Amplituhedron}},  \href{http://arxiv.org/abs/2410.11501}{{\tt 2410.11501}}.

\bibitem{Galashin:2018fri}
P.~Galashin and T.~Lam, \emph{{Parity duality for the amplituhedron}}, \href{http://dx.doi.org/10.1112/S0010437X20007411}{\emph{Compos. Math.} {\bf 156} (2020) 2207--2262}, [\href{http://arxiv.org/abs/1805.00600}{{\tt 1805.00600}}].

\bibitem{Lukowski:2020dpn}
T.~Lukowski, M.~Parisi and L.~K. Williams, \emph{{The Positive Tropical Grassmannian, the Hypersimplex, and the m = 2 Amplituhedron}}, \href{http://dx.doi.org/10.1093/imrn/rnad010}{\emph{Int. Math. Res. Not.} {\bf 2023} (2023) 16778--16836}, [\href{http://arxiv.org/abs/2002.06164}{{\tt 2002.06164}}].

\bibitem{Parisi:2021oql}
M.~Parisi, M.~Sherman-Bennett and L.~K. Williams, \emph{{The m=2 amplituhedron and the hypersimplex: signs, clusters, tilings, Eulerian numbers.}}, \href{http://dx.doi.org/10.1090/cams/23}{\emph{Commun. Am. Math. Soc.} {\bf 3} (2023) 329--399}, [\href{http://arxiv.org/abs/2104.08254}{{\tt 2104.08254}}].

\bibitem{Even-Zohar:2023del}
C.~Even-Zohar, T.~Lakrec, M.~Parisi, R.~Tessler, M.~Sherman-Bennett and L.~Williams, \emph{{Cluster algebras and tilings for the m=4 amplituhedron}},  \href{http://arxiv.org/abs/2310.17727}{{\tt 2310.17727}}.

\bibitem{Akhmedova:2023wcf}
E.~Akhmedova and R.~J. Tessler, \emph{{The Tropical Amplituhedron}},  \href{http://arxiv.org/abs/2312.12319}{{\tt 2312.12319}}.

\bibitem{Even-Zohar:2024nvw}
C.~Even-Zohar, T.~Lakrec, M.~Parisi, M.~Sherman-Bennett, R.~Tessler and L.~Williams, \emph{{A cluster of results on amplituhedron tiles}}, \href{http://dx.doi.org/10.1007/s11005-024-01854-4}{\emph{Lett. Math. Phys.} {\bf 114} (2024) 111}, [\href{http://arxiv.org/abs/2402.15568}{{\tt 2402.15568}}].

\bibitem{Lam:2024gyg}
T.~Lam, \emph{{On the face stratification of the $m=2$ amplituhedron}},  \href{http://arxiv.org/abs/2403.06948}{{\tt 2403.06948}}.

\bibitem{Parisi:2024psm}
M.~Parisi, M.~Sherman-Bennett, R.~Tessler and L.~Williams, \emph{{The Magic Number Conjecture for the $m=2$ amplituhedron and Parke-Taylor identities}},  \href{http://arxiv.org/abs/2404.03026}{{\tt 2404.03026}}.

\bibitem{Galashin:2024ttp}
P.~Galashin, \emph{{Amplituhedra and origami}},  \href{http://arxiv.org/abs/2410.09574}{{\tt 2410.09574}}.

\bibitem{Beisert:2006ez}
N.~Beisert, B.~Eden and M.~Staudacher, \emph{{Transcendentality and Crossing}}, \href{http://dx.doi.org/10.1088/1742-5468/2007/01/P01021}{\emph{J. Stat. Mech.} {\bf 0701} (2007) P01021}, [\href{http://arxiv.org/abs/hep-th/0610251}{{\tt hep-th/0610251}}].

\bibitem{Alday:2007hr}
L.~F. Alday and J.~M. Maldacena, \emph{{Gluon scattering amplitudes at strong coupling}}, \href{http://dx.doi.org/10.1088/1126-6708/2007/06/064}{\emph{JHEP} {\bf 06} (2007) 064}, [\href{http://arxiv.org/abs/0705.0303}{{\tt 0705.0303}}].

\bibitem{Drummond:2007au}
J.~M. Drummond, J.~Henn, G.~P. Korchemsky and E.~Sokatchev, \emph{{Conformal Ward identities for Wilson loops and a test of the duality with gluon amplitudes}}, \href{http://dx.doi.org/10.1016/j.nuclphysb.2009.10.013}{\emph{Nucl. Phys. B} {\bf 826} (2010) 337--364}, [\href{http://arxiv.org/abs/0712.1223}{{\tt 0712.1223}}].

\bibitem{Bern:2005iz}
Z.~Bern, L.~J. Dixon and V.~A. Smirnov, \emph{{Iteration of planar amplitudes in maximally supersymmetric Yang-Mills theory at three loops and beyond}}, \href{http://dx.doi.org/10.1103/PhysRevD.72.085001}{\emph{Phys. Rev. D} {\bf 72} (2005) 085001}, [\href{http://arxiv.org/abs/hep-th/0505205}{{\tt hep-th/0505205}}].

\bibitem{Alday:2009zm}
L.~F. Alday, J.~M. Henn, J.~Plefka and T.~Schuster, \emph{{Scattering into the fifth dimension of N=4 super Yang-Mills}}, \href{http://dx.doi.org/10.1007/JHEP01(2010)077}{\emph{JHEP} {\bf 01} (2010) 077}, [\href{http://arxiv.org/abs/0908.0684}{{\tt 0908.0684}}].

\bibitem{Flieger:2025ekn}
W.~Flieger, J.~Henn, A.~Schreiber and J.~Trnka, \emph{{Two-loop four-point amplitudes on the Coulomb branch of ${\mathcal{N}}=4$ super Yang-Mills}},  \href{http://arxiv.org/abs/2501.09454}{{\tt 2501.09454}}.

\bibitem{Alday:2011pf}
L.~F. Alday and A.~A. Tseytlin, \emph{{On strong-coupling correlation functions of circular Wilson loops and local operators}}, \href{http://dx.doi.org/10.1088/1751-8113/44/39/395401}{\emph{J. Phys. A} {\bf 44} (2011) 395401}, [\href{http://arxiv.org/abs/1105.1537}{{\tt 1105.1537}}].

\bibitem{Alday:2011ga}
L.~F. Alday, E.~I. Buchbinder and A.~A. Tseytlin, \emph{{Correlation function of null polygonal Wilson loops with local operators}}, \href{http://dx.doi.org/10.1007/JHEP09(2011)034}{\emph{JHEP} {\bf 09} (2011) 034}, [\href{http://arxiv.org/abs/1107.5702}{{\tt 1107.5702}}].

\bibitem{Henn:2019swt}
J.~M. Henn, G.~P. Korchemsky and B.~Mistlberger, \emph{{The full four-loop cusp anomalous dimension in $\mathcal{N}=4$ super Yang-Mills and QCD}}, \href{http://dx.doi.org/10.1007/JHEP04(2020)018}{\emph{JHEP} {\bf 04} (2020) 018}, [\href{http://arxiv.org/abs/1911.10174}{{\tt 1911.10174}}].

\bibitem{Alday:2012hy}
L.~F. Alday, P.~Heslop and J.~Sikorowski, \emph{{Perturbative correlation functions of null Wilson loops and local operators}}, \href{http://dx.doi.org/10.1007/JHEP03(2013)074}{\emph{JHEP} {\bf 03} (2013) 074}, [\href{http://arxiv.org/abs/1207.4316}{{\tt 1207.4316}}].

\bibitem{Alday:2013ip}
L.~F. Alday, J.~M. Henn and J.~Sikorowski, \emph{{Higher loop mixed correlators in N=4 SYM}}, \href{http://dx.doi.org/10.1007/JHEP03(2013)058}{\emph{JHEP} {\bf 03} (2013) 058}, [\href{http://arxiv.org/abs/1301.0149}{{\tt 1301.0149}}].

\bibitem{Chicherin:2022zxo}
D.~Chicherin and J.~Henn, \emph{{Pentagon Wilson loop with Lagrangian insertion at two loops in $ \mathcal{N} $ = 4 super Yang-Mills theory}}, \href{http://dx.doi.org/10.1007/JHEP07(2022)038}{\emph{JHEP} {\bf 07} (2022) 038}, [\href{http://arxiv.org/abs/2204.00329}{{\tt 2204.00329}}].

\bibitem{Arkani-Hamed:2021iya}
N.~Arkani-Hamed, J.~Henn and J.~Trnka, \emph{{Nonperturbative negative geometries: amplitudes at strong coupling and the amplituhedron}}, \href{http://dx.doi.org/10.1007/JHEP03(2022)108}{\emph{JHEP} {\bf 03} (2022) 108}, [\href{http://arxiv.org/abs/2112.06956}{{\tt 2112.06956}}].

\bibitem{Brown:2023mqi}
T.~V. Brown, U.~Oktem, S.~Paranjape and J.~Trnka, \emph{{Loops of Loops Expansion in the Amplituhedron}},  \href{http://arxiv.org/abs/2312.17736}{{\tt 2312.17736}}.

\bibitem{Chicherin:2024hes}
D.~Chicherin, J.~Henn, J.~Trnka and S.-Q. Zhang, \emph{{Positivity properties of five-point two-loop Wilson loops with Lagrangian insertion}},  \href{http://arxiv.org/abs/2410.11456}{{\tt 2410.11456}}.

\bibitem{Glew:2024zoh}
R.~Glew and T.~Lukowski, \emph{{Positive and Negative Ladders in Loop Space}},  \href{http://arxiv.org/abs/2411.14989}{{\tt 2411.14989}}.

\bibitem{He:2022cup}
S.~He, C.-K. Kuo, Z.~Li and Y.-Q. Zhang, \emph{{All-Loop Four-Point Aharony-Bergman-Jafferis-Maldacena Amplitudes from Dimensional Reduction of the Amplituhedron}}, \href{http://dx.doi.org/10.1103/PhysRevLett.129.221604}{\emph{Phys. Rev. Lett.} {\bf 129} (2022) 221604}, [\href{http://arxiv.org/abs/2204.08297}{{\tt 2204.08297}}].

\bibitem{Henn:2023pkc}
J.~M. Henn, M.~Lagares and S.-Q. Zhang, \emph{{Integrated negative geometries in ABJM}}, \href{http://dx.doi.org/10.1007/JHEP05(2023)112}{\emph{JHEP} {\bf 05} (2023) 112}, [\href{http://arxiv.org/abs/2303.02996}{{\tt 2303.02996}}].

\bibitem{Lagares:2024epo}
M.~Lagares and S.-Q. Zhang, \emph{{Higher-loop integrated negative geometries in ABJM}}, \href{http://dx.doi.org/10.1007/JHEP05(2024)142}{\emph{JHEP} {\bf 05} (2024) 142}, [\href{http://arxiv.org/abs/2402.17432}{{\tt 2402.17432}}].

\bibitem{Arkani-Hamed:2010pyv}
N.~Arkani-Hamed, J.~L. Bourjaily, F.~Cachazo and J.~Trnka, \emph{{Local Integrals for Planar Scattering Amplitudes}}, \href{http://dx.doi.org/10.1007/JHEP06(2012)125}{\emph{JHEP} {\bf 06} (2012) 125}, [\href{http://arxiv.org/abs/1012.6032}{{\tt 1012.6032}}].

\bibitem{Henn:2013pwa}
J.~M. Henn, \emph{{Multiloop integrals in dimensional regularization made simple}}, \href{http://dx.doi.org/10.1103/PhysRevLett.110.251601}{\emph{Phys. Rev. Lett.} {\bf 110} (2013) 251601}, [\href{http://arxiv.org/abs/1304.1806}{{\tt 1304.1806}}].

\bibitem{Cachazo:2008vp}
F.~Cachazo, \emph{{Sharpening The Leading Singularity}},  \href{http://arxiv.org/abs/0803.1988}{{\tt 0803.1988}}.

\bibitem{Chicherin:2022bov}
D.~Chicherin and J.~M. Henn, \emph{{Symmetry properties of Wilson loops with a Lagrangian insertion}}, \href{http://dx.doi.org/10.1007/JHEP07(2022)057}{\emph{JHEP} {\bf 07} (2022) 057}, [\href{http://arxiv.org/abs/2202.05596}{{\tt 2202.05596}}].

\bibitem{Arkani-Hamed:2014dca}
N.~Arkani-Hamed, A.~Hodges and J.~Trnka, \emph{{Positive Amplitudes In The Amplituhedron}}, \href{http://dx.doi.org/10.1007/JHEP08(2015)030}{\emph{JHEP} {\bf 08} (2015) 030}, [\href{http://arxiv.org/abs/1412.8478}{{\tt 1412.8478}}].

\bibitem{Dixon:2016apl}
L.~J. Dixon, M.~von Hippel, A.~J. McLeod and J.~Trnka, \emph{{Multi-loop positivity of the planar $ \mathcal{N} $ = 4 SYM six-point amplitude}}, \href{http://dx.doi.org/10.1007/JHEP02(2017)112}{\emph{JHEP} {\bf 02} (2017) 112}, [\href{http://arxiv.org/abs/1611.08325}{{\tt 1611.08325}}].

\bibitem{Henn:2024qwe}
J.~Henn and P.~Raman, \emph{{Positivity properties of scattering amplitudes}},  \href{http://arxiv.org/abs/2407.05755}{{\tt 2407.05755}}.

\bibitem{Arkani-Hamed:2009kmp}
N.~Arkani-Hamed, J.~Bourjaily, F.~Cachazo and J.~Trnka, \emph{{Unification of Residues and Grassmannian Dualities}}, \href{http://dx.doi.org/10.1007/JHEP01(2011)049}{\emph{JHEP} {\bf 01} (2011) 049}, [\href{http://arxiv.org/abs/0912.4912}{{\tt 0912.4912}}].

\bibitem{Dixon:2011pw}
L.~J. Dixon, J.~M. Drummond and J.~M. Henn, \emph{{Bootstrapping the three-loop hexagon}}, \href{http://dx.doi.org/10.1007/JHEP11(2011)023}{\emph{JHEP} {\bf 11} (2011) 023}, [\href{http://arxiv.org/abs/1108.4461}{{\tt 1108.4461}}].

\bibitem{Dixon:2011nj}
L.~J. Dixon, J.~M. Drummond and J.~M. Henn, \emph{{Analytic result for the two-loop six-point NMHV amplitude in N=4 super Yang-Mills theory}}, \href{http://dx.doi.org/10.1007/JHEP01(2012)024}{\emph{JHEP} {\bf 01} (2012) 024}, [\href{http://arxiv.org/abs/1111.1704}{{\tt 1111.1704}}].

\bibitem{Caron-Huot:2019vjl}
S.~Caron-Huot, L.~J. Dixon, F.~Dulat, M.~von Hippel, A.~J. McLeod and G.~Papathanasiou, \emph{{Six-Gluon amplitudes in planar $ \mathcal{N} $ = 4 super-Yang-Mills theory at six and seven loops}}, \href{http://dx.doi.org/10.1007/JHEP08(2019)016}{\emph{JHEP} {\bf 08} (2019) 016}, [\href{http://arxiv.org/abs/1903.10890}{{\tt 1903.10890}}].

\bibitem{Prlina:2017azl}
I.~Prlina, M.~Spradlin, J.~Stankowicz, S.~Stanojevic and A.~Volovich, \emph{{All-Helicity Symbol Alphabets from Unwound Amplituhedra}}, \href{http://dx.doi.org/10.1007/JHEP05(2018)159}{\emph{JHEP} {\bf 05} (2018) 159}, [\href{http://arxiv.org/abs/1711.11507}{{\tt 1711.11507}}].

\bibitem{Prlina:2017tvx}
I.~Prlina, M.~Spradlin, J.~Stankowicz and S.~Stanojevic, \emph{{Boundaries of Amplituhedra and NMHV Symbol Alphabets at Two Loops}}, \href{http://dx.doi.org/10.1007/JHEP04(2018)049}{\emph{JHEP} {\bf 04} (2018) 049}, [\href{http://arxiv.org/abs/1712.08049}{{\tt 1712.08049}}].

\bibitem{Penrose}
R.~Penrose, \emph{{Twistor Algebra}}, {\emph{Journal of Mathematical Physics} {\bf 8} (02, 1967) 345--366}.

\bibitem{Parke:1986gb}
S.~J. Parke and T.~R. Taylor, \emph{{An Amplitude for $n$ Gluon Scattering}}, \href{http://dx.doi.org/10.1103/PhysRevLett.56.2459}{\emph{Phys. Rev. Lett.} {\bf 56} (1986) 2459}.

\bibitem{Bourjaily:2015jna}
J.~L. Bourjaily and J.~Trnka, \emph{{Local Integrand Representations of All Two-Loop Amplitudes in Planar SYM}}, \href{http://dx.doi.org/10.1007/JHEP08(2015)119}{\emph{JHEP} {\bf 08} (2015) 119}, [\href{http://arxiv.org/abs/1505.05886}{{\tt 1505.05886}}].

\bibitem{Arkani-Hamed:2010zjl}
N.~Arkani-Hamed, J.~L. Bourjaily, F.~Cachazo, S.~Caron-Huot and J.~Trnka, \emph{{The All-Loop Integrand For Scattering Amplitudes in Planar N=4 SYM}}, \href{http://dx.doi.org/10.1007/JHEP01(2011)041}{\emph{JHEP} {\bf 01} (2011) 041}, [\href{http://arxiv.org/abs/1008.2958}{{\tt 1008.2958}}].

\bibitem{Hodges:2010kq}
A.~Hodges, \emph{{The Box Integrals in Momentum-Twistor Geometry}}, \href{http://dx.doi.org/10.1007/JHEP08(2013)051}{\emph{JHEP} {\bf 08} (2013) 051}, [\href{http://arxiv.org/abs/1004.3323}{{\tt 1004.3323}}].

\bibitem{Arkani_Hamed_2017}
N.~Arkani-Hamed, Y.~Bai and T.~Lam, \emph{Positive geometries and canonical forms}, \href{http://dx.doi.org/10.1007/jhep11(2017)039}{\emph{Journal of High Energy Physics} {\bf 2017} (Nov., 2017) }.

\bibitem{Ranestad:2024svp}
K.~Ranestad, R.~Sinn and S.~Telen, \emph{{Adjoints and Canonical Forms of Tree Amplituhedra}},  \href{http://arxiv.org/abs/2402.06527}{{\tt 2402.06527}}.

\bibitem{Arkani_Hamed_2018}
N.~Arkani-Hamed, H.~Thomas and J.~Trnka, \emph{Unwinding the amplituhedron in binary}, \href{http://dx.doi.org/10.1007/jhep01(2018)016}{\emph{Journal of High Energy Physics} {\bf 2018} (Jan., 2018) }.

\bibitem{Eden:2010ce}
B.~Eden, G.~P. Korchemsky and E.~Sokatchev, \emph{{More on the duality correlators/amplitudes}}, \href{http://dx.doi.org/10.1016/j.physletb.2012.02.014}{\emph{Phys. Lett. B} {\bf 709} (2012) 247--253}, [\href{http://arxiv.org/abs/1009.2488}{{\tt 1009.2488}}].

\bibitem{Henn:2013nsa}
J.~M. Henn, A.~V. Smirnov and V.~A. Smirnov, \emph{{Evaluating single-scale and/or non-planar diagrams by differential equations}}, \href{http://dx.doi.org/10.1007/JHEP03(2014)088}{\emph{JHEP} {\bf 03} (2014) 088}, [\href{http://arxiv.org/abs/1312.2588}{{\tt 1312.2588}}].

\bibitem{Lukowski:2019sxw}
T.~\L{}ukowski, M.~Parisi, M.~Spradlin and A.~Volovich, \emph{{Cluster Adjacency for $m=2$ Yangian Invariants}}, \href{http://dx.doi.org/10.1007/JHEP10(2019)158}{\emph{JHEP} {\bf 10} (2019) 158}, [\href{http://arxiv.org/abs/1908.07618}{{\tt 1908.07618}}].

\bibitem{Dian_2023}
G.~Dian, P.~Heslop and A.~Stewart, \emph{Internal boundaries of the loop amplituhedron}, \href{http://dx.doi.org/10.21468/scipostphys.15.3.098}{\emph{SciPost Physics} {\bf 15} (Sept., 2023) }.

\bibitem{Henn:2019mvc}
J.~Henn, B.~Power and S.~Zoia, \emph{{Conformal Invariance of the One-Loop All-Plus Helicity Scattering Amplitudes}}, \href{http://dx.doi.org/10.1007/JHEP02(2020)019}{\emph{JHEP} {\bf 02} (2020) 019}, [\href{http://arxiv.org/abs/1911.12142}{{\tt 1911.12142}}].

\bibitem{Gurdogan:2020tip}
O.~G\"urdo\u{g}an and M.~Parisi, \emph{{Cluster patterns in Landau and leading singularities via the amplituhedron}}, \href{http://dx.doi.org/10.4171/aihpd/155}{\emph{Ann. Inst. H. Poincare D Comb. Phys. Interact.} {\bf 10} (2023) 299--336}, [\href{http://arxiv.org/abs/2005.07154}{{\tt 2005.07154}}].

\bibitem{progress}
D.~Chicherin, J.~Henn, E.~Mazzucchelli, J.~Trnka, Q.~Yang and S.-Q. Zhang, \emph{{Integrated Negative Geometries from Landau Analysis}},  \href{http://arxiv.org/abs/in progress}{{\tt in progress}}.

\end{thebibliography}\endgroup

\end{document}